\documentclass[letterpaper,USenglish,numberwithinsect]{lipics-v2016}
 
\usepackage{microtype}
\usepackage{algpseudocode,algorithm}

\numberwithin{figure}{section}

\newcommand*{\Cdot}{{\centerdot}}
\newcommand{\OPT}{\mbox{\sc OPT}}
\newcommand{\SOL}{\mbox{\sc SOL}}
\def\mcL{\mathcal{L}}
\def\mcN{\mathcal{N}}
\def\mcP{\mathcal{P}}
\def\mcS{\mathcal{S}}

\title{A local search 2.917-approximation algorithm for duo-preservation string mapping%
\footnote{This work was partially supported by NSERC Canada and NSF China.}}
\titlerunning{Approximate duo-preservation through local search} 

\author[1]{Yao Xu}
\author[1,2]{Yong Chen}
\author[1,3]{Taibo Luo}
\author[1]{Guohui Lin}
\affil[1]{Department of Computing Science, University of Alberta.  Edmonton, Alberta T6G 2E8, Canada.\\
  \texttt{\{xu2,yong5,taibo,guohui\}@ualberta.ca}}
\affil[2]{Department of Mathematics, Hangzhou Dianzi University.  Hangzhou, Zhejiang 310018, China.}
\affil[3]{Business School, Sichuan University.  Chengdu, Sichuan 610065, China.}
\authorrunning{Y. Xu {\it et al.} (\today)} 

\Copyright{Yao Xu, Yong Chen, Taibo Luo, and Guohui Lin}

\subjclass{Dummy classification -- please refer to \url{http://www.acm.org/about/class/ccs98-html}}
\keywords{Approximation algorithm, duo-preservation string mapping, string partition, local search, amortized analysis}

\begin{document}

\maketitle

\begin{abstract}
We study the {\em maximum duo-preservation string mapping} ({\sc Max-Duo}) problem,
which is the complement of the well studied {\em minimum common string partition} ({\sc MCSP}) problem.
Both problems have applications in many fields including text compression and bioinformatics.
Motivated by an earlier local search algorithm,
we present an improved approximation and show that its performance ratio is no greater than ${35}/{12} < 2.917$.
This beats the current best $3.25$-approximation for {\sc Max-Duo}.
The performance analysis of our algorithm is done through a complex yet interesting amortization.
Two lower bounds on the locality gap of our algorithm are also provided.
\end{abstract}

\section{Introduction}
The {\em minimum common string partition} ({\sc MCSP}) problem is a well-studied problem in computer science,
with applications in the fields such as text compression and bioinformatics.
{\sc MCSP} was first introduced by Goldstein {\it et al.} \cite{GKZ04} as follows:
Consider two length-$n$ strings $A = (a_1, a_2, \ldots, a_n)$ and $B = (b_1, b_2, \ldots, b_n)$ over some alphabet $\Sigma$,
such that $B$ is a permutation of $A$.
A {\em partition} of $A$, denoted as $\mcP_A$, is a multi-set of substrings whose concatenation in a certain order becomes $A$.
The number of substrings in $\mcP_A$ is the {\em cardinality} of $\mcP_A$.
The {\sc MCSP} problem asks for a minimum cardinality partition $\mcP_A$ of $A$ that is also a partition of $B$.
When every letter of the alphabet $\Sigma$ occurs at most $k$ times in each of the two strings, the restricted version of {\sc MCSP} is denoted as $k$-{\sc MCSP}.

The {\sc MCSP} problem is NP-hard and APX-hard even when $k = 2$ \cite{GKZ04}.
Several approximation algorithms \cite{CZF05,CKS04,CM07,GKZ04,KW06,KW07} have been presented since 2004.
The current best result is an $O(\log n \log^* n)$-approximation for the general {\sc MCSP} and an $O(k)$-approximation for $k$-{\sc MCSP}.
On the other hand, {\sc MCSP} is proved to be {\em fixed parameter tractable} (FPT),
with respect to $k$ and/or to the cardinality of the optimal partition \cite{Dam08,JZZ12,BFK13,BK14}.

Given a string, an ordered pair of consecutive letters is called a {\em duo} \cite{GKZ04};
a length-$\ell$ substring in a partition {\em preserves} $\ell - 1$ duos of the given string.
The complementary objective to that of {\sc MCSP} is to maximize the number of duos preserved in the common partition,
which is referred to as the {\em maximum duo-preservation string mapping} ({\sc Max-Duo}) problem by Chen {\it et al.} \cite{CCS14}
and is our target problem in this paper.
Analogously, $k$-{\sc Max-Duo} is the restricted version of {\sc Max-Duo} when
every letter of the alphabet $\Sigma$ occurs at most $k$ times in each of the two given strings.

We next give a graphical view on a common partition of the two given strings $A = (a_1, a_2, \ldots, a_n)$ and $B = (b_1, b_2, \ldots, b_n)$.
Construct a bipartite graph $G = (A, B, E)$, where the vertices of $A$ ($B$, respectively) are $a_1, a_2, \ldots, a_n$ in order
($b_1, b_2, \ldots, b_n$ in order, respectively)
and there is an edge between $a_i$ and $b_j$ if they are the same letter.
A common partition $\mcP$ of the strings $A$ and $B$ one-to-one corresponds to a perfect matching $M$ in the graph $G$ (see Fig.~\ref{fig101} for an example),
and the number of duos preserved by the partition is exactly the number of pairs of {\em parallel} edges in the matching;
if both $(a_i, b_j), (a_{i+1}, b_{j+1}) \in E$, then they form a pair of parallel edges.

\begin{figure}[htb]
\centering
\includegraphics[width=0.5\linewidth]{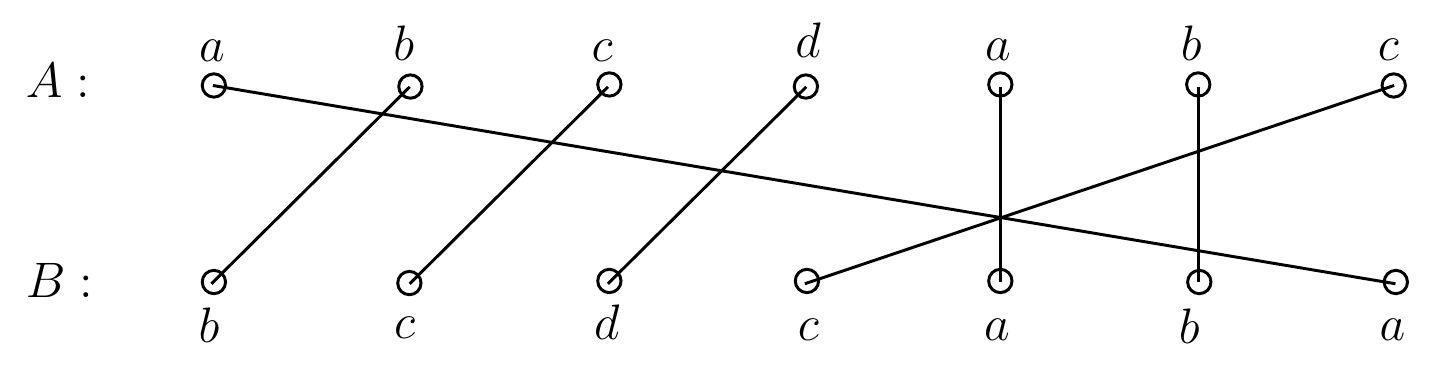} 
\caption{An instance of the {\sc Max-Duo} problem with two strings $A = (a, b, c, d, a, b, c)$ and $B = (b, c, d, c, a, b, a)$,
	and a common partition $\{a, bcd, ab, c\}$ that preserves three duos $(b, c)$, $(c, d)$ and $(a, b)$,
	corresponding to the perfect matching shown in the figure.\label{fig101}}
\end{figure}

Along with {\sc Max-Duo}, Chen {\it et al.} \cite{CCS14} introduced the {\em constrained maximum induced subgraph} ({\sc CMIS}) problem,
in which one is given an $m$-partite graph $G = (V_1, V_2, \ldots, V_m, E)$,
with each $V_i$ having $n_i^2$ vertices arranged in an $n_i \times n_i$ matrix,
and the goal is to select $n_i$ vertices of each $V_i$ in different rows and different columns such that the induced subgraph contains the maximum number of edges.
The restricted version of {\sc CMIS} when $n_i \le k$ for all $i$ is denoted as $k$-{\sc CMIS}.

For an instance of the {\sc Max-Duo} problem, one can first set $m$ to be the number of distinct letters in the string $A$,
set $n_i$ to be the number of occurrences of the $i$-th distinct letter,
and the $(s,t)$-vertex in the $n_i \times n_i$ matrix ``means''
mapping the $s$-th occurrence of the $i$-th distinct letter in the string $A$ to its $t$-th occurrence in the string $B$;
and then set an edge connecting a vertex of $V_i$ and a vertex of $V_j$ if the two vertices together preserve a duo.
This way, the {\sc Max-Duo} problem becomes a special case of the {\sc CMIS} problem, and furthermore the $k$-{\sc Max-Duo} is a special case of the $k$-{\sc CMIS}.

Chen {\it et al.} \cite{CCS14} presented a $k^2$-approximation for $k$-{\sc CMIS} and a $2$-approximation for $2$-{\sc CMIS},
based on linear programming and a randomized rounding.
These imply that $k$-{\sc Max-Duo} can also be approximated within a ratio of $k^2$ and $2$-{\sc Max-Duo} can be approximated within a ratio of $2$.

Continuing on the graphical view as shown in Fig.~\ref{fig101} on a common partition of the two given strings $A$ and $B$,
we can construct another graph $H = (V, F)$ in which every vertex of $V$ corresponds to a pair of parallel edges in the bipartite graph $G = (A, B, E)$,
and two vertices of $V$ are adjacent if the two pairs of parallel edges of $E$ cannot co-exist in any perfect matching of $G$
(called {\em conflicting}, which can be determined in constant time, see Section~\ref{sec2}).
This way, a set of duos that can be preserved by some perfect matching of $G$ (called {\em compatible}, see Section~\ref{sec2})
one-to-one corresponds to an independent set of $H$ \cite{GKZ04,BKL14}.
Therefore, the {\sc Max-Duo} problem can be cast as a special case of the well-known {\em maximum independent set} ({\sc MIS}) problem~\cite{GJ79};
in particular, Boria {\it et al.} \cite{BKL14} showed that an instance of $k$-{\sc Max-Duo} translates to a graph with the maximum degree $\Delta \le 6(k-1)$.
Since {\sc MIS} can be approximated arbitrarily close to ${(\Delta+3)}/{5}$~\cite{BF99},
$k$-{\sc Max-Duo} can now be better approximated within a ratio of ${(6k-3)}/{5} + \epsilon$, for any $\epsilon > 0$, using the same algorithm.
Especially, $2$-{\sc Max-Duo} and $3$-{\sc Max-Duo} can be approximated within a ratio of $1.8 + \epsilon$ and $3 + \epsilon$, respectively.
Boria {\it et al.} \cite{BKL14} also proved that {\sc Max-Duo} is APX-hard, even when $k = 2$.

In Section~\ref{sec2}, we will construct another bipartite graph for an instance of the {\sc Max-Duo} problem,
and thus cast {\sc Max-Duo} as a special case of the {\em maximum compatible bipartite matching} ({\sc MCBM}) problem.
Such a reduction was first shown by Boria {\it et al.} \cite{BKL14},
who presented a $4$-approximation for the {\sc MCBM} problem, implying that {\sc Max-Duo} can also be approximated within a ratio of $4$.
Boria {\it et al.} \cite{BCC16} also used this reduction, with the word {\em consecutive} in place of {\em compatible},
to present a local search $3.5$-approximation for the {\sc MCBM} problem.

Most recently, Brubach \cite{Bru16} presented a $3.25$-approximation for the {\sc Max-Duo} based on a novel {\em combinatorial triplet matching}.
This $3.25$-approximation is the current best for the general {\sc Max-Duo} problem.

The basic idea in the local search $3.5$-approximation for the {\sc MCBM} problem by Boria {\it et al.} \cite{BCC16} is
to swap one edge of the current matching out for two compatible edges,
thus to increase the size of the matching till a local optimum is reached.
The performance ratio $3.5$ is shown to be tight.
We extend this idea to allow swapping five edges of the current matching out for six compatible edges,
and we also allow a new operation of swapping five edges of the current matching out for five compatible edges
if the number of {\em singleton edges} (to be defined in Section~\ref{sec2}) is strictly decreased.
Through a complex yet interesting amortized analysis,
we prove that our local search heuristics has an approximation ratio of at most ${35}/{12} < 2.917$,
which improves the current best $3.25$-approximation algorithm and breaks the barrier of $3$.
In a companion paper \cite{XCLL17b}, we propose a $(1.4 + \epsilon)$-approximation for the $2$-{\sc Max-Duo};
thus together we improve all the current best approximability results.

The rest of the paper is organized as follows:
We provide some preliminaries in Section~\ref{sec2}, including the formal description of the {\sc MCBM} problem
and the terminologies and notations to be used throughout the paper.
Our local search heuristics is presented in Section~\ref{sec3}.
In Section~\ref{sec4}, we analyze the approximation ratio of our heuristics through amortization.
In Section~\ref{sec5}, we show a lower bound of ${13}/{6} > 2.166$ on the locality gap of our algorithm for the {\sc MCBM} problem,
and a lower bound of $5/3 > 1.666$ on the locality gap of our algorithm for the {\sc Max-Duo} problem.
We conclude the paper in Section~\ref{sec6}.

\section{Preliminaries}
\label{sec2}
Recall that in an instance of the {\sc Max-Duo} problem,
we have two length-$n$ strings $A = (a_1, a_2, \ldots, a_n)$ and $B = (b_1, b_2, \ldots, b_n)$ such that $B$ is a permutation of $A$.
We use $d^A_i = (a_i, a_{i+1})$ and $d^B_i = (b_i, b_{i+1})$ to denote the $i$-th duo of $A$ and $B$, respectively, for $i = 1, 2, \ldots, n-1$;
and $D^A = \{d^A_1, d^A_2, \ldots, d^A_{n-1}\}$ and $D^B = \{d^B_1, d^B_2, \ldots, d^B_{n-1}\}$.
We construct a bipartite graph $G = (D^A, D^B, E)$, where there is an edge $e_{i,j}$ connecting $d^A_i$ and $d^B_j$ if $a_i = b_j$ and $a_{i+1} = b_{j+1}$,
suggesting that the duo $d^A_i$ is preserved if the edge $e_{i,j} = (d^A_i, d^B_j)$ is selected into the solution matching.
(See Fig.~\ref{fig203a} for the bipartite graph constructed from the two strings shown in Fig.~\ref{fig101}.)
Note that selecting the edge $e_{i,j}$ rules out all the other edges incident at $d^A_i$ and all the other edges incident at $d^B_j$,
and some more edges described in the next paragraph.

Formally, the two edges $e_{i,j}$ and $e_{i,j'}$ with $j \ne j'$ are called {\em adjacent},
and they are {\em conflicting} since they cannot be both selected into a feasible solution matching.
Similarly, two adjacent edges $e_{i,j}$ and $e_{i',j}$ with $i \ne i'$ are conflicting.
The two edges $e_{i,j}$ and $e_{i+1,j+1}$ are called {\em parallel};
while the two edges $e_{i,j}$ and $e_{i+1,j'}$ with $j' \ne j, j+1$ are called {\em neighboring}.
Two neighboring edges are conflicting too since they cannot be both selected.
Similarly, the two edges $e_{i,j}$ and $e_{i',j+1}$ with $i' \ne i, i+1$ are neighboring and conflicting.
Any two unconflicting edges are said {\em compatible} to each other,
and a {\em compatible} set of edges contains edges that are pairwise compatible,
which is consequently a feasible solution matching (called a {\em compatible matching}).
(See Fig.~\ref{fig203b} for a compatible matching found in the bipartite graph in Fig.~\ref{fig203a}.)
The goal of the {\em maximum compatible bipartite matching} ({\sc MCBM}) problem is
to find a maximum cardinality compatible matching in the bipartite graph $G = (D^A, D^B, E)$.

\begin{figure}[htb]
\begin{subfigure}{0.48\textwidth}
\includegraphics[width=\linewidth]{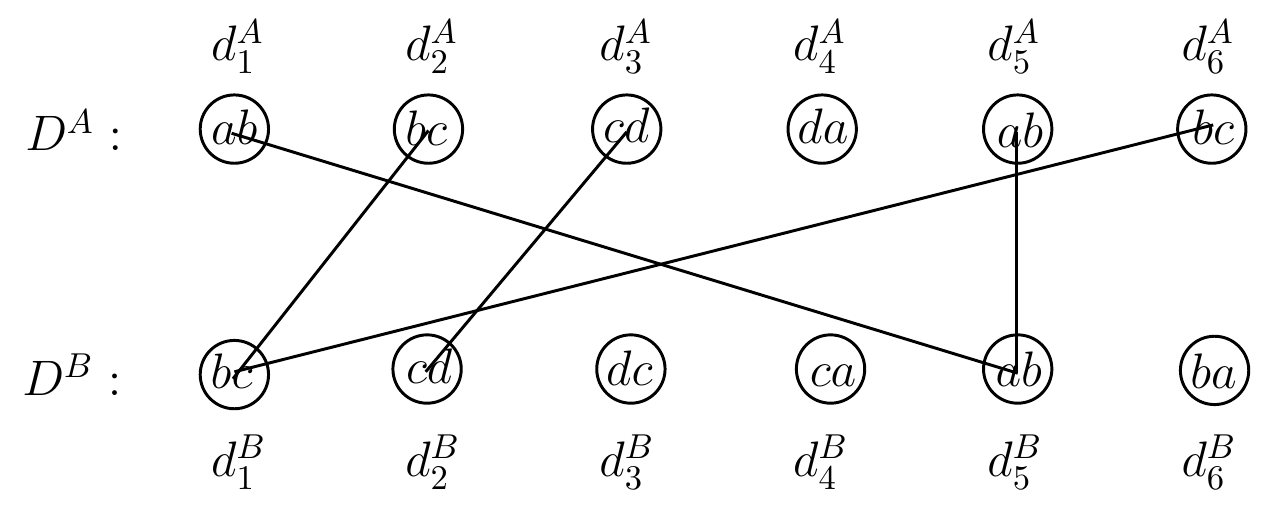} 
\caption{The constructed bipartite graph.}\label{fig203a}
\end{subfigure}
\hspace*{\fill}
\begin{subfigure}{0.48\textwidth}
\includegraphics[width=\linewidth]{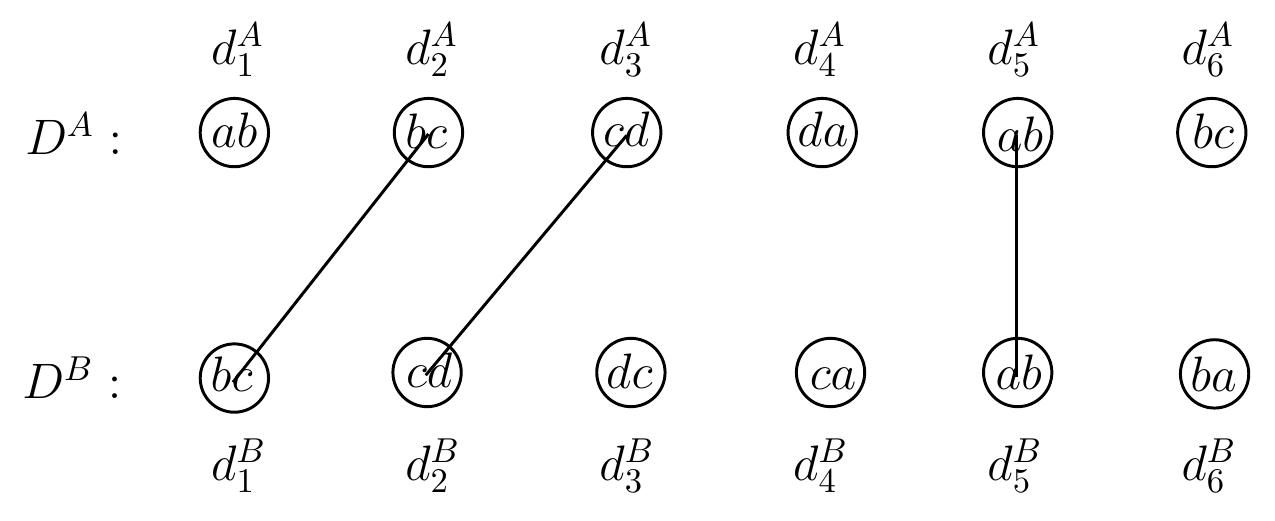} 
\caption{A compatible matching in the graph.} \label{fig203b}
\end{subfigure}
\caption{A bipartite graph $G = (D^A, D^B, E)$ constructed from the two strings $A = (a, b, c, d, a, b, c)$ and $B = (b, c, d, c, a, b, a)$,
	and a compatible matching in $G$ containing three edges $e_{2,1}, e_{3,2}, e_{5,5}$.\label{fig203}}
\end{figure}

Clearly, the bipartite graph $G = (D^A, D^B, E)$ in the {\sc MCBM} problem does not have to be constructed out of two given strings in the {\sc Max-Duo} problem,
and therefore {\sc Max-Duo} is a special case of {\sc MCBM}.
Nevertheless, when restricted to {\sc Max-Duo}, the cardinality of a compatible matching is exactly the number of duos preserved by the matching.
An edge in a compatible matching $M$ is called {\em singleton} if it is not parallel to any other edge in the matching.
This way, the matching $M$ is partitioned into two parts: $s(M)$ containing all the singleton edges and $p(M)$ containing all the parallel edges.
A series of pairs of parallel edges $e_{i,j}, e_{i+1,j+1}, \ldots, e_{i+p,j+p}$, for some $p \ge 2$,
is referred to as {\em consecutive parallel edges}.

Except towards the end we show a lower bound on the locality gap of our local search heuristics for the {\sc Max-Duo} problem,
all discussion in the sequel is on the {\sc MCBM} problem.
The obtained approximability results on the {\sc MCBM} problem also apply to the {\sc Max-Duo} problem.

\begin{obs}
\label{obs2.1}
Any edge $e_{i,j} \in E$ can be conflicting with at most $6$ edges that are pairwise compatible,
which are $e_{i,j'}$, $e_{i-1,j''-1}$, $e_{i+1,j'''+1}$, $e_{i',j}$, $e_{i''-1,j-1}$, $e_{i'''+1,j+1}$
incident at $d^A_{i-1}, d^A_i, d^A_{i+1}, d^B_{j-1}, d^B_j, d^B_{j+1}$, respectively,
where none of $i', i'', i'''$ can be $i$ and none of $j', j'', j'''$ can be $j$.
\end{obs}

We remark that in Observation~\ref{obs2.1} by ``at most'', some of the six edges could be void, that is, non-existent in $E$;
also, when $e_{i,j'}$ and $e_{i-1,j''-1}$ both present, then they have to be parallel suggesting that $j' = j''$
(the same applies to $e_{i,j'}$ and $e_{i+1,j'''+1}$, $e_{i',j}$ and $e_{i''-1,j-1}$, $e_{i',j}$ and $e_{i'''+1,j+1}$).

In the sequel, in general, the subscript of a vertex of $D^A$ has an $i$ or $h$, and the subscript of a vertex of $D^B$ has a $j$ or $\ell$.

\section{A local search heuristics $\mcL\mcS$}
\label{sec3}
Given a bipartite graph $G = (D^A, D^B, E)$,
the $3.5$-approximation algorithm presented by Boria {\it et al.} \cite{BCC16} starts with an arbitrary maximal compatible matching,
iteratively seeks swapping one edge in the current matching out for two compatible edges,
and terminates when the expansion by such swapping is impossible.

Our local search heuristics is an extension of the above algorithm,
to iteratively apply two different swapping operations to increase the size of the matching and to decrease the number of singleton edges in the matching,
respectively.
We present the heuristics in details in the following.
Note that we also start with an arbitrary maximal compatible matching, which by Observation~\ref{obs2.1} can be obtained in $O(n^2)$-time,
where $n$ is the number of vertices in one side of the bipartite graph (or more precisely, $|D^A| = |D^B| = n-1$).

Let $M$ denote the current compatible matching in hand.
For any edge $e_{i,j} \in M$, let ${C}(e_{i, j})$ be the set of all the edges of $E$ conflicting with $e_{i,j}$;
then ($q = -1, 0, +1$ in the following set unions)
\begin{equation}
\label{eq1}
C(e_{i, j}) = \bigcup_{q=-1}^{+1}\left\{e_{i+q, j'+q} \in E \mid j' \ne j\right\} \cup
			  \bigcup_{q=-1}^{+1}\left\{e_{i'+q, j+q} \in E \mid i' \ne i\right\}.
\end{equation}

Clearly, $|{C}(e_{i, j})| \le 6(n-1)$.
Recall that $|E| \in O(n^2)$.
We have the following observation, which essentially narrows down the candidate edges for swapping with the edge $e_{i,j}$.

\begin{obs}
\label{obs3.1}
For a maximal compatible matching $M$ and an edge $e_{i,j} \in M$,
the edges compatible with all the edges of $M - \{e_{i, j}\}$ must be in ${C}(e_{i, j}) \cup \{e_{i, j}\}$.
\end{obs}

We next describe the two different swapping operations.
Both of them apply to a maximal compatible match $M$.
One operation is to replace five edges of $M$ by six edges, denoted as {\sc Replace-5-by-6}, thus to increase the size of the matching;
and the other operation is to replace five edges of $M$ by five edges with the resulted matching having strictly less singleton edges, denoted as {\sc Reduce-5-by-5}.
Note that in each iteration, the operation {\sc Reduce-5-by-5} applies only when the operation {\sc Replace-5-by-6} fails to expand the current matching $M$.

\subsection{Operation {\sc Replace-5-by-6}}
The operation {\sc Replace-5-by-6} seeks to expand the current maximal compatible matching $M$ by swapping five edges of $M$ out for six compatible edges.
It does so by scanning all size-$5$ subsets of $M$ and terminates at a successful expansion.
If no such expansion is possible, it also terminates but without making any change to the matching $M$.

Let $X = \{e_1, e_2, \ldots, e_5\}$ be a subset of $M$
(in the special case where $|M| \le 5$, we seek for a compatible matching of size $|M| + 1$ directly by an exhaustive search).
The operation composes a set $E' = X \cup C(X)$,
where $C(X)$ contains all the edges each conflicting with an edge of $X$ but compatible with (all the edges of) $M - X$;
it then checks every size-$6$ subset $X'$ of $E'$ for compatibility and, if affirmative, swaps $X$ out for $X'$ to expand $M$.

Recall that $|M| < n$.
The number of size-$5$ subsets of $M$ is $O(n^5)$.
For each size-$5$ subset $X$, composing the set $E'$ takes $O(n^2)$ time and $|E'| < 30n$.
It follows that the number of size-$6$ subsets of $E'$ is $O(n^6)$.
Lastly, checking the compatibility of each size-$6$ subset $X'$ takes $O(1)$ time.
Therefore, the time complexity of the operation {\sc Replace-5-By-6} is $O(n^{11})$.

\subsection{Operation {\sc Reduce-5-by-5}}
From Equation~\ref{eq1}, one sees that given a maximal compatible matching $M$,
a pair of parallel edges of $M$ are expected to conflict much less edges outside of $M$ than two singleton edges of $M$ do.
This hints that for two compatible matchings of the same cardinality, the one with more parallel edges more likely can be expanded,
and motivates the new operation {\sc Reduce-5-by-5}.

When the operation {\sc Replace-5-By-6} fails to expand the current maximal compatible matching $M$,
the operation {\sc Reduce-5-by-5} seeks to decrease the number of singleton edges in $M$, by swapping five edges of $M$ out for five compatible edges.
Similarly, it does so by scanning all size-$5$ subsets of $M$, and terminates at a successful reduction.
If no such reduction is possible, it also terminates but without making any change to the matching $M$.

Recall that $M$ is partitioned into $p(M)$ and $s(M)$, containing all the parallel edges and all the singleton edges, respectively.
Let $X = \{e_1, e_2, \ldots, e_5\}$ be a subset of $M$ 
(in the special case where $|M| \le 5$,
we seek for a compatible matching of the same size but containing strictly less singleton edges directly by an exhaustive search).
The operation composes a set $E' = X \cup C(X)$, where $C(X)$ contains all the edges each conflicting with an edge of $X$ but compatible with $M - X$;
it then checks every size-$5$ subset $X'$ of $E'$ for compatibility and subsequently checks whether $|s(M - X \cup X')| < |s(M)|$,
if both affirmative, swaps $X$ out for $X'$ to reduce the number of singleton edges in $M$.

For the time complexity of the operation {\sc Reduce-5-by-5}, similarly we recall that $|M| < n$.
Partitioning $M$ into $p(M)$ and $s(M)$ takes at most $O(n^2)$ time.
There are $O(n^5)$ size-$5$ subsets of $M$. 
For each such size-$5$ subset $X$, composing the set $E'$ takes $O(n^2)$ time and $|E'| < 30n$.
It follows that the number of size-$5$ subsets of $E'$ is $O(n^5)$.
Lastly, checking the compatibility of each size-$5$ subset $X'$ takes $O(1)$ time and counting the singleton edges of $M - X \cup X'$ can be done in $O(n)$ time.
Therefore, the time complexity of the operation {\sc Reduce-5-By-5} is $O(n^{11})$ too.

\subsection{The local search heuristics $\mcL\mcS$}
Our local search heuristics is iterative.
The compatible matching $M$ is initialized to $\emptyset$.

At the beginning of each iteration, we greedily expand the current compatible matching $M$ to the maximal, by adding one edge at a time.
Next, with the current maximal compatible matching $M$, the operation {\sc Replace-5-By-6} is applied to expand $M$.
If successful, the iteration ends.
Otherwise, $M$ is not modified by the operation {\sc Replace-5-By-6} and the operation {\sc Reduce-5-By-5} is applied to reduce the number of singleton edges in $M$.
If successful, the iteration ends;
otherwise the entire algorithm terminates and returns the current $M$ as the solution.

Clearly, the step of greedy expansion takes $O(n^2)$ time.
The running time of the rest of the iteration is $O(n^{11})$, which is dominant.

Note that every iteration, except the last,
either increases the cardinality of the compatible matching or decreases the number of the singleton edges in the compatible matching.
We thus conclude that there are $O(n^2)$ iterations in the entire algorithm, which we denote as $\mcL\mcS$.
It follows that the time complexity of the algorithm $\mcL\mcS$ is $O(n^{13})$.
We state this result in the following theorem.

\begin{theorem}
\label{thm301}
The time complexity of the local search heuristics $\mcL\mcS$ for the {\sc MCBM} problem is $O(n^{13})$,
where $n$ is the number of vertices in one side of the bipartite graph.
\end{theorem}

\section{Approximation ratio analysis for the heuristics $\mcL\mcS$}
\label{sec4}
We analyze the performance ratio of the heuristics $\mcL\mcS$ through {\em amortization}.
The main result is to prove that the heuristics $\mcL\mcS$ is a $35/12$-approximation for the {\sc MCBM} problem,
and thus it is also a $35/12$-approximation for the {\sc Max-Duo} problem.

\subsection{The amortization scheme}
\label{sec4.1}
Let $M^*$ be the optimal compatible matching to the {\sc MCBM} problem and $\OPT = |M^*|$,
and $M$ be the maximal compatible matching returned by the algorithm $\mcL\mcS$ and $\SOL = |M|$.
We partition $M$ into $s(M)$ and $p(M)$.
(In the sequel, notations with a superscript $^*$ are associated with $M^*$;
notations without a superscript are associated with $M$.
In general, the subscript of a vertex of $D^A$ has an $i$ or $h$, and the subscript of a vertex of $D^B$ has a $j$ or $\ell$.)

In the amortization scheme, we assign one token to each edge $e^* \in M^*$, and thus the total amount of tokens is $\OPT$.
The edge $e^*$ will be conflicting to a number of edges of $M$ (including the case where $e^*$ is in $M$, then $e^*$ is conflicting to itself only);
it then splits the token evenly and distributes a fraction to every conflicting edge of $M$.
To the end, the total amount of tokens received by all the edges of $M$ is exactly $\OPT$.
Our main task is to estimate an upper bound (which is expected to be $35/12$) on the amount of tokens received by an edge of $M$,
thereby to give a lower bound on $\SOL$.

Formally, we define the function $\tau(e \gets e^*) \ge 0$ to be the amount of token $e^* \in M^*$ gives to $e \in M$.
For the edge $e^* \in M^*$, let $C(e^*) \subseteq M$ be the subset of edges of $M$ conflicting with $e^*$,
and for the edge $e \in M$, let $C^*(e) \subseteq M^*$ be the subset of edges of $M^*$ conflicting with $e$.
From the maximality, we know that both $|C(e^*)|, |C^*(e)| \ge 1$, for any $e^*, e$.
Then, $\tau(e \gets e^*) = \frac{1}{|C(e^*)|}$, if $e \in C(e^*)$; or otherwise $\tau(e \gets e^*) = 0$.
The total amount of tokens $e \in M$ receives is denoted as
\begin{equation}
\label{eq2}
\omega(e) := \sum_{e^* \in C^*(e)} \frac{1}{|C(e^*)|}, \forall e \in M.
\end{equation}
And we have
\[
\OPT = \sum_{e \in M} \omega(e) \le \max_{e \in M} \omega(e) \cdot \SOL.
\]

Therefore, the quantity $\max_{e \in M} \omega(e)$ is an upper bound on the performance ratio of the algorithm $\mcL\mcS$.
We thus aim to estimate $\max_{e \in M} \omega(e)$.
In the following, we will see that $\max_{e \in M} \omega(e) = {10}/3$, which is larger than our target ratio ${35}/{12}$.
We then switch to enumerate all possible cases where an edge $e$ has $\omega(e) \ge 3$ and
amortize some fraction of its token to certain provably existing edges $e'$ with $\omega(e') < 3$.
In other words, we will estimate the average value of $\omega(\cdot)$ for all the edges of $M$, denoted as $\overline{\omega(e)}$,
and prove an upper bound (which is shown to be $35/12$) on $\overline{\omega(e)}$ that is also an upper bound on the performance ratio of the algorithm $\mcL\mcS$.

To this purpose, we may assume without loss of generality that $M \cap M^* = \emptyset$ since their $\omega(\cdot)$'s are all $1$.
According to Observation~\ref{obs2.1} in Section~\ref{sec2}, we have $|C^*(e)| \le 6$ and $|C(e^*)| \le 6$ for any $e \in M$ and $e^* \in M^*$.
Consider an arbitrary edge $e_{i,j} \in M$, we have
\begin{equation*}
C^*(e_{i, j}) = \{e^*_{i-1, j''-1}, e^*_{i, j'}, e^*_{i+1, j'''+1}, e^*_{i''-1, j-1}, e^*_{i', j}, e^*_{i'''+1, j+1}\},
\end{equation*}
where $e^*_{i,j'}$ ($e^*_{i-1, j''-1}, e^*_{i+1, j'''+1}$, respectively) denotes the edge of $M^*$ incident at $d^A_i$ ($d^A_{i-1}, d^A_{i+1}$, respectively),
if it exists, or otherwise it is a void edge;
$e^*_{i',j}$ ($e^*_{i''-1, j-1}, e^*_{i'''+1, j+1}$, respectively) denotes the edge of $M^*$ incident at $d^B_j$ ($d^B_{j-1}, d^B_{j+1}$, respectively),
if it exists, or otherwise it is a void edge;
and none of $i', i'', i'''$ can be $i$ and none of $j', j'', j'''$ can be $j$.
(It is important to point out that $C^*(e_{i, j})$ does not necessarily contain $6$ edges, due to the possible void edges.)
We partition $C^*(e_{i, j})$ into two parts $C^*(e_{i, \Cdot})$ and $C^*(e_{\Cdot, j})$:
\begin{align*}
C^*(e_{i, \Cdot}) &
= \{e^*_{i-1, j''-1}, e^*_{i, j'}, e^*_{i+1, j'''+1}\}, \\
C^*(e_{\Cdot, j}) &
= \{e^*_{i''-1, j-1}, e^*_{i', j}, e^*_{i'''+1, j+1}\}.
\end{align*}
(Again, each of $C^*(e_{i, \Cdot})$ and $C^*(e_{\Cdot, j})$ does not necessarily contain $3$ edges, due to the possible void edges.)
We extend the function notation to let $\tau(e_{i, j} \gets C^*(e_{i, j}))$ be the multi-set of the $\tau(e_{i, j} \gets e^*)$ values,
where $e^* \in C^*(e_{i, j})$, that is,
\begin{align}
\tau(e_{i, j} \gets C^*(e_{i, \Cdot})) &
= \left\{\frac{1}{|C(e^*)|} \ \bigg| \ e^* \in C^*(e_{i, \Cdot}) \right\},\\
\tau(e_{i, j} \gets C^*(e_{\Cdot, j})) &
= \left\{\frac{1}{|C(e^*)|} \ \bigg| \ e^* \in C^*(e_{\Cdot, j}) \right\},\\
\tau(e_{i, j} \gets C^*(e_{i, j})) &
= \tau(e_{i, j} \gets C^*(e_{i, \Cdot})) \cup \tau(e_{i, j} \gets C^*(e_{\Cdot, j})).
\end{align}
Then $\omega(e_{i, j})$ is the sum of all the (at most six) values in the set $\tau(e_{i, j} \gets C^*(e_{i, j}))$;
each of these values can be any of $1, \frac 12, \frac 13, \frac 14, \frac 15, \frac 16$, since $1 \le |C(e^*)| \le 6$ for any $e^* \in C^*(e_{i, j})$.
We also use the following vectors to represent the {\em ordered values} of 
$\tau(e_{i, j} \gets C^*(e_{i, \Cdot}))$ and $\tau(e_{i, j} \gets C^*(e_{\Cdot, j}))$, respectively:
\begin{equation*}
\big(\tau(e_{i, j} \gets e^*_{i-1, j''-1}), \tau(e_{i, j} \gets e^*_{i, j'}), \tau(e_{i, j} \gets e^*_{i+1, j'''+1})\big),
\end{equation*}
\begin{equation*}
\big(\tau(e_{i, j} \gets e^*_{i''-1, j-1}), \tau(e_{i, j} \gets e^*_{i', j}), \tau(e_{i, j} \gets e^*_{i'''+1, j+1})\big).
\end{equation*}
We need the following three more subsets of $M$, all of which are associated with $e_{i,j} \in M$.
\begin{align*}
C(C^*(e_{i, \Cdot})) &
= \bigcup_{e^* \in C^*(e_{i, \Cdot})} C(e^*),\\
C(C^*(e_{\Cdot, j})) &
= \bigcup_{e^* \in C^*(e_{\Cdot, j})} C(e^*), \\
C(C^*(e_{i, j})) &
= C(C^*(e_{i, \Cdot})) \cup C(C^*(e_{\Cdot, j})).
\end{align*}

\subsection{Value combinations of $\tau(e_{i, j} \gets C^*(e_{i, j}))$ with $\omega(e_{i, j}) \ge 3$}
\label{sec4.2}
Note that the operation {\sc Replace-5-by-6} actually executes swapping $p$ edges of the current compatible matching out
for $p+1$ compatible edges to expand the matching, for $p = 1, 2, 3, 4, 5$.
Therefore, for any edge $e_{i,j} \in M$,
we can never have two edges $e^*_{i_1, j_1}, e^*_{i_2, j_2} \in C^*(e_{i, j})$ such that $|C(e^*_{i_1, j_1})| = |C(e^*_{i_2, j_2})| = 1$,
that is, both of them conflict with only the edge $e_{i,j}$ in $M$.
Thus we immediately have the following lemma, which has also been observed in \cite{BCC16}.

\begin{lemma}
\label{lemma401}
{\rm \cite{BCC16}}
For any edge $e_{i,j} \in M$, there is at most one edge $e^*_{i_1, j_1} \in C^*(e_{i, j})$ such that $|C(e^*_{i_1, j_1})| = 1$.
\end{lemma}

\begin{lemma}
\label{lemma402}
For any edge $e_{i,j} \in M$, and for any pair of parallel edges $e^*_{i_1, j_1}, e^*_{i_1+1, j_1+1} \in C^*(e_{i, j})$,
$| |C(e^*_{i_1, j_1})| - |C(e^*_{i_1+1, j_1+1})| | \le 2$.
\end{lemma}
\begin{proof}
Since the edges of $C(e^*_{i_1, j_1}) \cup C(e^*_{i_1+1, j_1+1}) \subseteq M$ are pairwise compatible,
we have
\begin{align*}
C(e^*_{i_1, j_1}) - C(e^*_{i_1 + 1, j_1 + 1}) & 
\subseteq \{e_{i_1-1, \diamond}, e_{\diamond, j_1-1}\}, \\
C(e^*_{i_1+1, j_1+1}) - C(e^*_{i_1, j_1}) & 
\subseteq \{e_{i_1+2, \diamond}, e_{\diamond, j_1+2}\},
\end{align*}
where $e_{i_1-1, \diamond}$ ($e_{\diamond, j_1-1}, e_{i_1+2, \diamond}, e_{\diamond, j_1+2}$, respectively) denotes the edge of $M$ incident at
$d^A_{i_1-1}$ ($d^B_{j_1-1}, d^A_{i_1+2}, d^B_{j_1+2}$, respectively), if it exists, or otherwise it is a void edge.
Thus, $|C(e^*_{i_1, j_1}) - C(e^*_{i_1+1, j_1+1})| \le 2$ and $|C(e^*_{i_1+1, j_1+1}) - C(e^*_{i_1, j_1})| \le 2$,
which together imply $| |C(e^*_{i_1, j_1})| - |C(e^*_{i_1+1, j_1+1})| | \le 2$.
\end{proof}

\begin{lemma}
\label{lemma403}
Suppose $|C^*(e_{i, \Cdot})| = 3$, then $C^*(e_{i, \Cdot}) = \{e^*_{i-1, j'-1}, e^*_{i, j'}, e^*_{i+1, j'+1}\}$ for some $j' \ne j$.
In this case we can never have $|C(e^*_{i-1, j'-1})| = |C(e^*_{i, j'})| = |C(e^*_{i+1, j'+1})| = 2$,
if one of the following three conditions holds:
\begin{enumerate}
\item there is an edge $e^*_{i_1, j_1} \in C^*(e_{\Cdot, j})$ such that $|C(e^*_{i_1, j_1})| = 1$;
\item $|C(C^*(e_{\Cdot, j}))| \le |C^*(e_{\Cdot, j})|$;
\item there is at least one singleton edge of $M$ in $C(C^*(e_{i, \Cdot}))$.
\end{enumerate}
\end{lemma}
\begin{proof}
Recall that $C^*(e_{i, \Cdot}) = \{e^*_{i-1, j''-1}, e^*_{i, j'}, e^*_{i+1, j'''+1}\}$ for some $j', j'', j''' (\ne j)$.
When all these three edges of $M^*$ exist, they are consecutive parallel edges, that is,
$j' = j'' = j'''$ and thus $C^*(e_{i, \Cdot}) = \{e^*_{i-1, j'-1}, e^*_{i, j'}, e^*_{i+1, j'+1}\}$ for some $j' \ne j$.
This proves the first half of the lemma.

Next, assume $|C(e^*_{i-1, j'-1})| = |C(e^*_{i, j'})| = |C(e^*_{i+1, j'+1})| = 2$, and we will show none of the three conditions holds.

Since $e_{i, j} \in C(e^*_{i-1, j'-1}) \cap C(e^*_{i, j'}) \cap C(e^*_{i+1, j'+1})$,
each of $C(e^*_{i-1, j'-1}), C(e^*_{i, j'}), C(e^*_{i+1, j'+1})$ contains exactly one edge other than $e_{i,j}$.
Observe that any edge of $M$ conflicting with $e^*_{i, j'}$ must be conflicting with either $e^*_{i-1, j'-1}$ or $e^*_{i+1, j'+1}$.
We conclude that either $C(e^*_{i-1, j'-1}) = C(e^*_{i, j'})$ or $C(e^*_{i+1, j'+1}) = C(e^*_{i, j'})$,
implying that $2 \le |C(C^*(e_{i, \Cdot}))| \le 3$.

If $|C(C^*(e_{i, \Cdot}))| = 2$, then the algorithm $\mcL\mcS$ would have replaced these two edges of $C(C^*(e_{i, \Cdot}))$ by the three edges of $C^*(e_{i, \Cdot})$,
contradicting to the fact that $M$ is the solution by $\mcL\mcS$.
Therefore, $|C(C^*(e_{i, \Cdot}))| = 3$.

If the first condition holds, the algorithm $\mcL\mcS$ would have replaced the three edges of $C(C^*(e_{i, \Cdot}))$ by
the edge $e^*_{i_1, j_1}$ and the three edges of $C^*(e_{i, \Cdot})$ to expand $M$, again a contradiction.

If the second condition holds, 
we have $|C(C^*(e_{i, j}))| \le |C(C^*(e_{i, \Cdot}))| + |C(C^*(e_{\Cdot, j}))| - 1 \le 2 + |C^*(e_{\Cdot, j})| < |C^*(e_{i, j})| \le 6$.
Then, the algorithm $\mcL\mcS$ would have replaced all the edges of $C(C^*(e_{i, j}))$ by all the edges of $C^*(e_{i, j})$ to expand $M$, also a contradiction.

When there is at least one singleton edge of $M$ in $C(C^*(e_{i, \Cdot}))$, we distinguish two cases where $e_{i,j}$ is singleton or not.
If $e_{i,j}$ is not a singleton, then we may assume the edge $e_{i+1,j+1} \in M$ and thus $e_{i+1,j+1} \in C(C^*(e_{i, \Cdot}))$ too;
it follows from $|C(e^*_{i-1, j'-1})| = |C(e^*_{i, j'})| = |C(e^*_{i+1, j'+1})| = 2$ that these two edges form an isolated pair of parallel edges in $M$.
In this case, the algorithm $\mcL\mcS$ would have replaced the three edges in $C(C^*(e_{i, \Cdot}))$ by the three parallel edges of $C^*(e_{i, \Cdot})$
to decrease the number of singleton edges by at least one, a contradiction.
If $e_{i,j}$ is a singleton, then the other edge conflicting with $e^*_{i,j'}$ must also be a singleton.
The algorithm $\mcL\mcS$ would still have replaced the three edges in $C(C^*(e_{i, \Cdot}))$ by the three parallel edges of $C^*(e_{i, \Cdot})$
to decrease the number of singleton edges by at least one, again a contradiction.

In summary, we conclude that none of the three conditions would hold.
This proves the second half of the lemma.
\end{proof}

For an edge $e_{i,j} \in M$ with $\omega(e_{i, j}) \ge 3$,
we can now characterize the multi-set $\tau(e_{i, j} \gets C^*(e_{i, j}))$ of six values, in which an entry of $0$ represents a void edge in $C^*(e_{i, j})$.
We arrange these six values in a non-increasing order.
Using the above three Lemmas \ref{lemma401}--\ref{lemma403}, we have the following conclusion:

\begin{lemma}
\label{lemma404}
For an edge $e_{i,j} \in M$ with $\omega(e_{i, j}) \ge 3$, there are $8$ possible value combinations of $\tau(e_{i, j} \gets C^*(e_{i, j}))$,
which are
$\big\{1, \frac 12, \frac 12, \frac 12, \frac 12, \frac 13\big\}$, 
$\big\{1, \frac 12, \frac 12, \frac 12, \frac 12, \frac 14\big\}$,
$\big\{1, \frac 12, \frac 12, \frac 12, \frac 13, \frac 13\big\}$,
$\big\{1, \frac 12, \frac 12, \frac 12, \frac 13, \frac 14\big\}$,
$\big\{1, \frac 12, \frac 12, \frac 12, \frac 13, \frac 15\big\}$, 
$\big\{1, \frac 12, \frac 12, \frac 12, \frac 14, \frac 14\big\}$,
$\big\{1, \frac 12, \frac 12, \frac 13, \frac 13, \frac 13\big\}$, and
$\big\{1, \frac 12, \frac 12, \frac 12, \frac 12, 0\big\}$.
These combinations give rise to $\omega(e_{i, j}) = \frac{10}{3}$, $\frac{13}{4}$, $\frac{19}{6}$, $\frac{37}{12}$, $\frac{91}{30}$, $3$, $3$ and $3$ respectively.
\end{lemma}
%

We remark that in Lemma~\ref{lemma404}, $|C^*(e_{i, j})| = 6$ except for the last combination where $|C^*(e_{i, j})| = 5$.
Also, we see that $\max_{e \in M} \omega(e) \le 10/3$, implying that the algorithm $\mcL\mcS$ is a $10/3$-approximation.
(This is worse than the current best $3.25$-approximation though.)

\subsection{Ordered value combinations of $\tau(e_{i, j} \gets C^*(e_{i, \Cdot}))$ with $\omega(e_{i, j}) \ge 3$}
\label{sec4.2-3}
We discuss the possible ordered value combinations of $\tau(e_{i, j} \gets C^*(e_{i, \Cdot}))$ in this section.

Using the first condition of Lemma \ref{lemma403},
we can rule out $\big\{\frac 12, \frac 12, \frac 12\big\}$ for $\tau(e_{i, j} \gets C^*(e_{i, \Cdot}))$.
From the $8$ possible value combinations of $\tau(e_{i, j} \gets C^*(e_{i, j}))$ stated in Lemma~\ref{lemma404},
by Lemma~\ref{lemma402} we can identify in total $12$ possible value combinations of $\tau(e_{i, j} \gets C^*(e_{i, \Cdot}))$ with $\omega(e_{i, j}) \ge 3$,
stated in the following lemma.

\begin{lemma}
\label{lemma405}
For an edge $e_{i,j} \in M$ with $\omega(e_{i, j}) \ge 3$, there are $12$ possible value combinations of $\tau(e_{i, j} \gets C^*(e_{i, \Cdot}))$,
which are
$\big\{1, \frac 12, \frac 12\big\}$,
$\big\{1, \frac 12, \frac 13\big\}$,
$\big\{1, \frac 12, \frac 14\big\}$,
$\big\{1, \frac 13, \frac 13\big\}$,
$\big\{\frac 12, \frac 12, \frac 13\big\}$,
$\big\{\frac 12, \frac 12, \frac 14\big\}$,
$\big\{\frac 12, \frac 13, \frac 13\big\}$,
$\big\{\frac 12, \frac 13, \frac 14\big\}$,
$\big\{\frac 12, \frac 13, \frac 15\big\}$,
$\big\{\frac 12, \frac 14, \frac 14\big\}$,
$\big\{\frac 13, \frac 13, \frac 13\big\}$, and
$\big\{\frac 12, \frac 12, 0\big\}$.
\end{lemma}

\begin{lemma}
\label{lemma406}
Suppose $|C^*(e_{i, \Cdot})| = 3$ and $C^*(e_{i, \Cdot}) = \{e^*_{i-1, j'-1}$, $e^*_{i, j'}$, $e^*_{i+1, j'+1}\}$ for some $j' \ne j$.
We have
\begin{align}
C(e^*_{i, j'}) \subseteq & \
C(e^*_{i-1, j'-1}) \cup C(e^*_{i+1, j'+1}), \label{eq6} \\
|C(C^*(e_{i, \Cdot}))| \le & \
|C(e^*_{i-1, j'-1})| + |C(e^*_{i+1, j'+1})| - 1, \label{eq7} \\
|C(C^*(e_{i, \Cdot}))| \ge & \
\max
\begin{cases}
3, \\
|C(e^*_{i-1, j'-1})| + |C(e^*_{i+1, j'+1})| - 2, \\
|C(e^*_{i-1, j'-1})| + |C(e^*_{i+1, j'+1})| - |C(e^*_{i, j'})|. \label{eq8}
\end{cases}
\end{align}
\end{lemma}
\begin{proof}
Observe that any edge of $M$ conflicting with $e^*_{i, j'}$ must also conflict with either $e^*_{i-1, j'-1}$ or $e^*_{i+1, j'+1}$.
We have $C(e^*_{i, j'}) \subseteq C(e^*_{i-1, j'-1}) \cup C(e^*_{i+1, j'+1})$,
which proves the inequality (\ref{eq6}) and also indicates that
$C(C^*(e_{i, \Cdot})) = C(e^*_{i-1, j'-1}) \cup C(e^*_{i+1, j'+1})$.
Since $e_{i, j} \in C(e^*_{i-1, j'-1}) \cap C(e^*_{i, j'}) \cap C(e^*_{i+1, j'+1}) \subseteq \{e_{i, j}, e_{\diamond, j'}\}$,
where $e_{\diamond, j'}$ is a possible edge of $M$ incident at $d^B_{j'}$,
we have
\begin{equation*}
|C(e^*_{i-1, j'-1})| + |C(e^*_{i+1, j'+1})| - 2
\le |C(C^*(e_{i, \Cdot}))|
\le |C(e^*_{i-1, j'-1})| + |C(e^*_{i+1, j'+1})| - 1.
\end{equation*}
This proves the inequality (\ref{eq7}) and the second inequality in (\ref{eq8}).

Also observe that any edge of $M$ conflicting with both $e^*_{i-1, j'-1}$ and $e^*_{i+1, j'+1}$ must conflict with $e^*_{i, j'}$ too.
We have $C(e^*_{i-1, j'-1}) \cap C(e^*_{i+1, j'+1}) \subseteq C(e^*_{i, j'})$.
Therefore,
\begin{equation*}
|C(C^*(e_{i, \Cdot}))| \ge |C(e^*_{i-1, j'-1})| + |C(e^*_{i+1, j'+1})| - |C(e^*_{i, j'})|,
\end{equation*}
proving the last inequality in (\ref{eq8}).
$|C(C^*(e_{i, \Cdot}))| \ge 3$ can be proven by a simple contradiction,
since otherwise the algorithm $\mcL\mcS$ would replace all the edges of $C(C^*(e_{i, \Cdot}))$ by the three edges of $C^*(e_{i, \Cdot})$ to expand $M$.
\end{proof}

\begin{lemma}
\label{lemma407}
Suppose $|C^*(e_{i, \Cdot})| = 3$ and $C^*(e_{i, \Cdot}) = \{e^*_{i-1, j'-1}$, $e^*_{i, j'}$, $e^*_{i+1, j'+1}\}$ for some $j' \ne j$,
and there is an edge $e^*_{i_3, j_3} \in C^*(e_{\Cdot, j})$ such that $|C(e^*_{i_3, j_3})| = 1$.
For any two edges $e^*_{i_1, j_1}, e^*_{i_2, j_2} \in C^*(e_{i, \Cdot})$, 
if $|C(e^*_{i_1, j_1})| = |C(e^*_{i_2, j_2})| = 2$, then the following two statements hold:
\begin{enumerate}
\item $e^*_{i_1, j_1}$ and $e^*_{i_2, j_2}$ are parallel, that is, either $i_2 = i_1 + 1, j_2 = j_1 + 1$ or $i_2 = i_1 - 1, j_2 = j_1 - 1$.
\item $C(e^*_{i_1, j_1}) \cap C(e^*_{i_2, j_2}) = \{e_{i, j}\}$.
\end{enumerate}
\end{lemma}
\begin{proof}
Using $|C(e^*_{i_3, j_3})| = 1$, we know from Lemma \ref{lemma401} that $|C(e^*_{i-1, j'-1})| \ge 2$, $|C(e^*_{i, j'})| \ge 2$ and $|C(e^*_{i+1, j'+1})| \ge 2$.

To prove the first statement, we suppose to the contrary that $i_1 = i-1$ and $i_2 = i+1$,
and thus $|C(e^*_{i-1, j'-1})| = |C(e^*_{i+1, j'+1})| = 2$.
From the inequality (\ref{eq7}) of Lemma \ref{lemma406}, we have $|C(e^*_{i, j'})| \le |C(C^*(e_{i, \Cdot}))| \le 3$.
Since Lemma \ref{lemma403} has ruled out the possibility of $|C(e^*_{i, j'})| = 2$,
we have $|C(e^*_{i, j'})| = |C(C^*(e_{i, \Cdot}))| = 3$.
However, the algorithm $\mcL\mcS$ would replace the three edges of $C(C^*(e_{i, \Cdot}))$ by $e^*_{i_3, j_3}$ and the three edges of $C^*(e_{i, \Cdot})$
to expand $M$, a contradiction.

Based on the first statement, we assume without loss of generality that $|C(e^*_{i, j'})| = |C(e^*_{i-1, j'-1})| = 2$.
Note that $e_{i, j} \in C(e^*_{i, j'}) \cap C(e^*_{i-1, j'-1})$.
If $C(e^*_{i, j'}) = C(e^*_{i-1, j'-1})$, then the algorithm $\mcL\mcS$ would replace the two edges of $C(e^*_{i, j'})$ by the three edges
$e^*_{i, j'}$, $e^*_{i-1, j'-1}$, $e^*_{i_3, j_3}$ to expand $M$, a contradiction.
Therefore, $C(e^*_{i, j'}) \ne C(e^*_{i-1, j'-1})$, which implies the second statement $C(e^*_{i_1, j_1}) \cap C(e^*_{i_2, j_2}) = \{e_{i, j}\}$.
\end{proof}

Note that each value combination $\{\tau_1, \tau_2, \tau_3\}$ of $\tau(e_{i, j} \gets C^*(e_{i, \Cdot}))$ in Lemma~\ref{lemma405}
gives rise to $3! = 6$ different ordered value combinations.
Due to symmetry, we consider only three of them: $(\tau_2, \tau_1, \tau_3)$, $(\tau_1, \tau_2, \tau_3)$, and $(\tau_1, \tau_3, \tau_2)$,
in the following to determine whether or not they can be possible ordered value combinations for $\tau(e_{i, j} \gets C^*(e_{i, \Cdot}))$.

\begin{enumerate}
\item $\tau(e_{i, j} \gets C^*(e_{i, \Cdot})) = \big\{1, \frac 12, \frac 12\big\}$. \\
The case of $\big(1, \frac 12, \frac 12\big)$ can be ruled out by the inequalities (\ref{eq7}) and (\ref{eq8}) of Lemma \ref{lemma406}. \\
Then the only possible case left is $\big(\frac 12, 1, \frac 12\big)$.

\item $\tau(e_{i, j} \gets C^*(e_{i, \Cdot})) = \big\{1, \frac 12, \frac 13\big\}$. \\
The case of $\big(1, \frac 13, \frac 12\big)$ can immediately be ruled out by the inequality (\ref{eq7}) of Lemma \ref{lemma406}. \\
Then the two possible cases left are $\big(\frac 12, 1, \frac 13\big)$ and $\big(1, \frac 12, \frac 13\big)$.

\item $\tau(e_{i, j} \gets C^*(e_{i, \Cdot})) = \big\{1, \frac 12, \frac 14\big\}$. \\
Both cases of $\big(\frac 12, 1, \frac 14\big)$ and $\big(1, \frac 14, \frac 12\big)$ can immediately be ruled out by Lemma \ref{lemma402}. \\
Then the only possible case left is $\big(1, \frac 12, \frac 14\big)$.

\item $\tau(e_{i, j} \gets C^*(e_{i, \Cdot})) = \big\{1, \frac 13, \frac 13\big\}$. \\
Consider the case of $\tau(e_{i, j} \gets C^*(e_{i, \Cdot})) = \big(1, \frac 13, \frac 13\big)$.
In this case, we have $C(C^*(e_{i, \Cdot})) = C(e^*_{i, j'}) = C(e^*_{i+1, j'+1})$ with $|C(C^*(e_{i, \Cdot}))| = 3$,
indicating that one of the three edges in $C(C^*(e_{i, \Cdot}))$ must be a singleton edge of $M$ and
there is no edge in $M - C(C^*(e_{i, \Cdot}))$ parallel with any edge in $C(C^*(e_{i, \Cdot}))$.
However, the algorithm $\mcL\mcS$ would replace the three edges of $C(C^*(e_{i, \Cdot}))$ by the three parallel edges of $C^*(e_{i, \Cdot})$
to reduce the singleton edges in $M$, a contradiction. \\
Thus the only possible case left is $\big(\frac 13, 1, \frac 13\big)$.

\item $\tau(e_{i, j} \gets C^*(e_{i, \Cdot})) = \big\{\frac 12, \frac 12, \frac 13\big\}$. \\
The case of $\big(\frac 12, \frac 13, \frac 12\big)$ can immediately be ruled out by Lemma \ref{lemma407}. \\
Then the only possible case left is $\big(\frac 12, \frac 12, \frac 13\big)$.

\item $\tau (C^*(e_{i, \Cdot})) = \big\{\frac 12, \frac 12, \frac 14\big\}$. \\
The case of $\big(\frac 12, \frac 14, \frac 12\big)$ can immediately be ruled out by the inequality (\ref{eq7}) of Lemma \ref{lemma406}. \\
Then the only possible case left is $\big(\frac 12, \frac 12, \frac 14\big)$.

\item $\tau(e_{i, j} \gets C^*(e_{i, \Cdot})) = \big\{\frac 12, \frac 13, \frac 13\big\}$. \\
Both cases of $\big(\frac 13, \frac 12, \frac 13\big)$ and $\big(\frac 12, \frac 13, \frac 13\big)$ are possible.

\item $\tau(e_{i, j} \gets C^*(e_{i, \Cdot})) = \big\{\frac 12, \frac 13, \frac 14\big\}$. \\
All three cases of $\big(\frac 13, \frac 12, \frac 14\big)$, $\big(\frac 12, \frac 13, \frac 14\big)$, and $\big(\frac 12, \frac 14, \frac 13\big)$ are possible.

\item $\tau(e_{i, j} \gets C^*(e_{i, \Cdot})) = \big\{\frac 12, \frac 13, \frac 15\big\}$. \\
Both cases of $\big(\frac 13, \frac 12, \frac 15\big)$ and $\big(\frac 12, \frac 15, \frac 13\big)$ can immediately be ruled out by Lemma \ref{lemma402}. \\
Then the only possible case left is $\big(\frac 12, \frac 13, \frac 15\big)$.

\item $\tau (C^*(e_{i, \Cdot})) = \big\{\frac 12, \frac 14, \frac 14\big\}$. \\
Both cases of $\big(\frac 12, \frac 14, \frac 14\big)$ and $\big(\frac 14, \frac 12, \frac 14\big)$ are possible.

\item $\tau(e_{i, j} \gets C^*(e_{i, \Cdot})) = \big\{\frac 13, \frac 13, \frac 13\big\}$. \\
The only case $\big(\frac 13, \frac 13, \frac 13\big)$ is possible.

\item $\tau(e_{i, j} \gets C^*(e_{i, \Cdot})) = \big\{\frac 12, \frac 12, 0\big\}$. \\
Both cases of $\big(\frac 12, \frac 12, 0\big)$ and $\big(\frac 12, 0, \frac 12\big)$ are possible.
\end{enumerate}

We summarize the above discussion in the following lemma:

\begin{lemma}
\label{lemma408}
For an edge $e_{i,j} \in M$ with $\omega(e_{i, j}) \ge 3$, there are $18$ possible ordered value combinations of $\tau(e_{i, j} \gets C^*(e_{i, \Cdot}))$,
which are
$\big(\frac 12, 1, \frac 12\big)$, 
$\big(\frac 12, 1, \frac 13\big)$, $\big(1, \frac 12, \frac 13\big)$, 
$\big(1, \frac 12, \frac 14\big)$,
$\big(\frac 13, 1, \frac 13\big)$,
$\big(\frac 12, \frac 12, \frac 13\big)$, 
$\big(\frac 12, \frac 12, \frac 14\big)$, 
$\big(\frac 13, \frac 12, \frac 13\big)$, $\big(\frac 12, \frac 13, \frac 13\big)$, 
$\big(\frac 13, \frac 12, \frac 14\big)$, $\big(\frac 12, \frac 13, \frac 14\big)$, $\big(\frac 12, \frac 14, \frac 13\big)$, 
$\big(\frac 12, \frac 13, \frac 15\big)$,
$\big(\frac 12, \frac 14, \frac 14\big)$, $\big(\frac 14, \frac 12, \frac 14\big)$,
$\big(\frac 13, \frac 13, \frac 13\big)$,
$\big(\frac 12, \frac 12, 0\big)$, and
$\big(\frac 12, 0, \frac 12\big)$.
\end{lemma}

\subsection{Edge combinations of $C(C^*(e_{i, j}))$ with $\omega(e_{i, j}) \ge 3$}
\label{sec4.3}
We examine all possible combinations of the edges in $C(C^*(e_{i, \Cdot}))$ with $\omega(e_{i, j}) \ge 3$.
We distinguish two cases where $e_{i, j} \in p(M)$ and $e_{i, j} \in s(M)$, respectively.
In fact, as shown in Section~\ref{sec4.3.1}, the edge $e_{i,j}$ cannot be a parallel edge in $M$.

\subsubsection{$e_{i, j}$ cannot be a parallel edge of $M$}
\label{sec4.3.1}
Recall that the number of singleton edges of the maximal compatible matching $M$ cannot be further reduced
by the algorithm $\mcL\mcS$ using the operation {\sc Reduce-5-By-5}.

We assume to the contrary that $e_{i, j} \in p(M)$, and assume that $e_{i+1, j+1} \in p(M)$ too.

From $|C^*(e_{i, j})| \ge 5$ in Lemma~\ref{lemma404},
we consider $|C^*(e_{i, \Cdot})| = 3$ and suppose that $C^*(e_{i, \Cdot}) = \{e^*_{i-1, j'-1}, e^*_{i, j'}, e^*_{i+1, j'+1}\}$ for some $j' \ne j$.

Clearly, $|C(e^*_{i, j'})| \ge 2$ and $|C(e^*_{i+1, j'+1})| \ge 2$ since both contain the edges $e_{i,j}$ and $e_{i+1,j+1}$.
It follows that the middle value in the ordered value combination of $\tau(e_{i, j} \gets C^*(e_{i, \Cdot}))$ must be $\le \frac 12$.
This rules out three of the $18$ possible ordered value combinations stated in Lemma \ref{lemma408}, each having a $1$ in the middle, which are
$\big(\frac 12, 1, \frac 12\big)$, $\big(\frac 12, 1, \frac 13\big)$, $\big(\frac 13, 1, \frac 13\big)$.
Furthermore, since $\big(\frac 12, 1, \frac 12\big)$ is the only one resulted from the (unordered) value combination $\big\{1, \frac 12, \frac 12\big\}$,
we conclude that it is impossible to have $\tau(e_{i, j} \gets C^*(e_{i, \Cdot})) = \big\{1, \frac 12, \frac 12\big\}$.
For the same reason, it is impossible to have $\tau(e_{i, j} \gets C^*(e_{i, \Cdot})) = \big\{1, \frac 13, \frac 13\big\}$.

When $|C^*(e_{\Cdot, j})| = 3$, the argument in the last paragraph applies to $C^*(e_{\Cdot, j})$ too.

Consider the $8$ possible value combinations of $\tau(e_{i, j} \gets C^*(e_{i, j}))$ such that $\omega(e_{i, j}) \ge 3$, in Lemma~\ref{lemma404}.
We observe that $\tau(e_{i, j} \gets C^*(e_{i, \Cdot})) \in$  
$\big\{\big\{\frac 12, \frac 13, \frac 14\big\}$,
$\big\{\frac 12, \frac 13, \frac 15\big\}$,
$\big\{\frac 12, \frac 14, \frac 14\big\}$,
$\big\{\frac 13, \frac 13, \frac 13\big\}$,
$\big\{\frac 12, \frac 12, 0\big\}\big\}$
only if $\tau(e_{i, j} \gets C^*(e_{\Cdot, j})) = \big\{1, \frac 12, \frac 12\big\}$, which is impossible to happen.
We thus conclude that only $5$ out of the $12$ value combinations of $\tau(e_{i, j} \gets C^*(e_{i, \Cdot}))$ in Lemma~\ref{lemma405} remain possible,
which are
$\big\{1, \frac 12, \frac 13\big\}$,
$\big\{1, \frac 12, \frac 14\big\}$,
$\big\{\frac 12, \frac 12, \frac 13\big\}$,
$\big\{\frac 12, \frac 12, \frac 14\big\}$, and
$\big\{\frac 12, \frac 13, \frac 13\big\}$.
These give $6$ possible ordered value combinations of $\tau(e_{i, j} \gets C^*(e_{i, \Cdot}))$,
which are,
$\big(1, \frac 12, \frac 13\big)$,
$\big(1, \frac 12, \frac 14\big)$,
$\big(\frac 12, \frac 12, \frac 13\big)$,
$\big(\frac 12, \frac 12, \frac 14\big)$,
$\big(\frac 13, \frac 12, \frac 13\big)$, and
$\big(\frac 12, \frac 13, \frac 13\big)$.

In the rest of this section, we have $|C^*(e_{i, j})| = 6$,
$C^*(e_{i, \Cdot}) = \{e^*_{i-1, j'-1}, e^*_{i, j'}, e^*_{i+1, j'+1}\}$ for some $j' \ne j$,
and $C^*(e_{\Cdot, j}) = \{e^*_{i'-1, j-1}, e^*_{i', j}, e^*_{i'+1, j+1}\}$ for some $i' \ne i$.

\begin{lemma}
\label{lemma409}
For the pair of parallel edges $e_{i, j}, e_{i+1, j+1} \in p(M)$,
$C^*(e_{i, j}) \cap C^*(e_{i+1, j+1}) = \{e^*_{i, j'}, e^*_{i+1, j'+1}, e^*_{i', j}, e^*_{i'+1, j+1}\}$.
If $|C(e^*_{i-1, j'-1})| = 1$, then there is at most one edge $e^*_{i_1, j_1} \in C^*(e_{i, j}) \cap C^*(e_{i+1, j+1})$ such that $|C(e^*_{i_1, j_1})| = 2$.
\end{lemma}
\begin{proof}
The first half of the lemma is trivial.
For the second half, we note that $C(e^*_{i-1, j'-1}) = \{e_{i,j}\}$;
if there is another edge $e^*_{i_2, j_2} \in C^*(e_{i, j}) \cap C^*(e_{i+1, j+1})$ such that $|C(e^*_{i_2, j_2})| = 2$, 
that is, $C(e^*_{i_1, j_1}) = C(e^*_{i_2, j_2}) = \{e_{i, j}, e_{i+1, j+1}\}$,
then the algorithm $\mcL\mcS$ would replace the two edges $e_{i, j}$ and $e_{i+1, j+1}$ by the three edges $e^*_{i-1, j'-1}$, $e^*_{i_1, j_1}$, $e^*_{i_2, j_2}$
to expand $M$, a contradiction.
\end{proof}

Lemma \ref{lemma409} states that when $e_{i,j}$ is a parallel edge of $M$,
the value combination of $\tau(e_{i, j} \gets C^*(e_{i, j}))$ contains at most two $\frac 12$'s, besides the value $1$.
Among the $8$ possible value combinations of $\tau(e_{i, j} \gets C^*(e_{i, j}))$ such that $\omega(e_{i, j}) \ge 3$, in Lemma~\ref{lemma404},
the only one with two $\frac 12$'s is $\big\{1, \frac 12, \frac 12, \frac 13, \frac 13, \frac 13\big\}$.
This leaves only two possible ordered value combinations of $\tau(e_{i, j} \gets C^*(e_{i, \Cdot}))$,
which are,
$\big(1, \frac 12, \frac 13\big)$ and
$\big(\frac 12, \frac 13, \frac 13\big)$.

Assume $\tau(e_{i, j} \gets C^*(e_{i, \Cdot})) = \big(1, \frac 12, \frac 13\big)$ and
$\tau(e_{i, j} \gets C^*(e_{\Cdot, j})) = \big(\frac 12, \frac 13, \frac 13\big)$ (or the other way around).
By the inequalities (\ref{eq7}) and (\ref{eq8}) in Lemma \ref{lemma406}, 
we have $|C(C^*(e_{i, \Cdot}))| = 3$ and $3 \le |C(C^*(e_{\Cdot, j}))| \le 4$.
Since $\{e_{i, j}, e_{i+1, j+1}\} \subseteq C(C^*(e_{i, \Cdot})) \cap C(C^*(e_{\Cdot, j}))$,
we have $4 \le |C(C^*(e_{i, j}))| \le 5$.
Thus, the algorithm $\mcL\mcS$ would replace all the edges of $C(C^*(e_{i, j}))$ by the six edges of $C^*(e_{i, j})$ to expand $M$.
This contradiction leaves no ordered value combination of $\tau(e_{i, j} \gets C^*(e_{i, \Cdot}))$.
We thus have proved the following lemma:

\begin{lemma}
\label{lemma410}
When the edge $e_{i, j}$ is a parallel edge of $M$, there is no value combination of $\tau(e_{i, j} \gets C^*(e_{i, j}))$ such that $\omega(e_{i, j}) \ge 3$.
\end{lemma}

\subsubsection{$e_{i, j}$ is a singleton edge of $M$}
\label{sec4.3.2}
With $e_{i, j} \in s(M)$, we discuss each of the $18$ possible ordered value combinations of $\tau(e_{i, j} \gets C^*(e_{i, \Cdot}))$ listed in Lemma \ref{lemma408}.

Consider an edge $e_{h, \ell} \in C(C^*(e_{i, \Cdot}))$, $e_{h, \ell} \ne e_{i,j}$.
Note that $e_{h, \ell}$ might be parallel with an edge in $M - C(C^*(e_{i, \Cdot}))$.
We define $\mcN_p(e_{h, \ell})$ to be the subset of the maximal consecutive parallel (to $e_{h,\ell}$) edges in $M - C(C^*(e_{i, \Cdot}))$.
Therefore, $\mcN_p(e_{h, \ell})$ will be either $\{ e_{h+1, \ell+1}, \ldots, e_{h+q, \ell+q} \}$ or $\{ e_{h-1, \ell-1}, \ldots, e_{h-q, \ell-q} \}$,
for some $q \ge 0$ (when $q = 0$, this set is empty).
Let
\[
\mcN_p(C(C^*(e_{i, \Cdot}))) = \bigcup_{e_{h, \ell} \in C(C^*(e_{i, \Cdot}))} \mcN_p(e_{h, \ell}),
\]
and
\[
\mcN_p[C(C^*(e_{i, \Cdot}))] = \mcN_p(C(C^*(e_{i, \Cdot}))) \cup C(C^*(e_{i, \Cdot})).
\]
Recall that, in general, the subscript of a vertex of $D^A$ has an $i$ or $h$, and the subscript of a vertex of $D^B$ has a $j$ or $\ell$.
In the sequel, for simplicity, we use $e_h$ and $e_\ell$ ($e^*_h$ and $e^*_\ell$, respectively) to denote the edges of $M$ ($M^*$, respectively)
incident at the vertices $d^A_i$ and $d^B_j$, respectively, if they exist, or otherwise they are void edges.

We next discuss all possible configurations of the edges of $C^*(e_{i, \Cdot})$ and $C(C^*(e_{i, \Cdot}))$ in figures,
associated with each of the $18$ ordered value combinations of $\tau(e_{i, j} \gets C^*(e_{i, \Cdot}))$ listed in Lemma \ref{lemma408}.
We adopt the following scheme for graphically presenting a configuration:
In each figure (for example, Fig.~\ref{fig401}),
the edge $e_{i, j}$ is in the bold solid line;
the edges in vertical bold dashed lines are in $C^*(e_{i, \Cdot})$ (for example, $e^*_{i, j'}$);
the edges in thin solid lines are edges in $C(C^*(e_{i, \Cdot}))$ (for example, $e_{i+2}$);
and the edges in thin dashed lines are edges in $\mcN_p(C(C^*(e_{i, \Cdot}))$ (for example, $e_{i+3}$);
the vertices in filled circles are surely not incident with any edge of $M$ (for example, $i-2$);
the vertices in hollow circles have uncertain incidence in $M$ (for example, $j'-2$).

We remind the readers that if there is no entry $1$ in a value combination of $\tau(e_{i, j} \gets C^*(e_{i, \Cdot}))$,
then there must be an entry $1$ in the corresponding value combination of $\tau(e_{i, j} \gets C^*(e_{\Cdot, j}))$,
that is, there is an edge $e^*_{i_1, j_1} \in C^*(e_{\Cdot, j}))$ such that $|C(e^*_{i_1, j_1})| = 1$.

\begin{enumerate}
\item $\tau(e_{i, j} \gets C^*(e_{i, \Cdot})) = \big(\frac 12, 1, \frac 12\big)$:
According to the inequalities (\ref{eq7}) and (\ref{eq8}) of Lemma \ref{lemma406}, we have $|C(C^*(e_{i, \Cdot}))| = 3$.
There is exactly one edge of $M$ incident at either of $i+2$ and $j'+2$ but not both.
We assume $e_{i+2} \in M$.
If $e_{i+2}$ is a singleton edge of $M$ or $|\mcN_p(e_{i+2})| \ge 2$,
then the algorithm $\mcL\mcS$ would replace $e_{i, j}$ and $e_{i+2}$ by the two parallel edges $e^*_{i, j'}$ and $e^*_{i+1, j'+1}$ to reduce the singleton edges,
a contradiction.
Therefore, we have $e_{i+3} \in M$ but no edge of $M$ is incident at $i+4$.
The incidence  at $i-2$ and $j'-2$ and further to the left can be symmetrically discussed.
In this sense, there is only one possible edge combination of $C(C^*(e_{i, \Cdot}))$, as shown in Fig.~\ref{fig401} with $e_{i+2}, e_{j'-2} \in M$,
where the corresponding configuration of $\mcN_p[C(C^*(e_{i, \Cdot}))]$ is also shown.

\begin{figure}[H]
\centering
\includegraphics[width=0.5\linewidth]{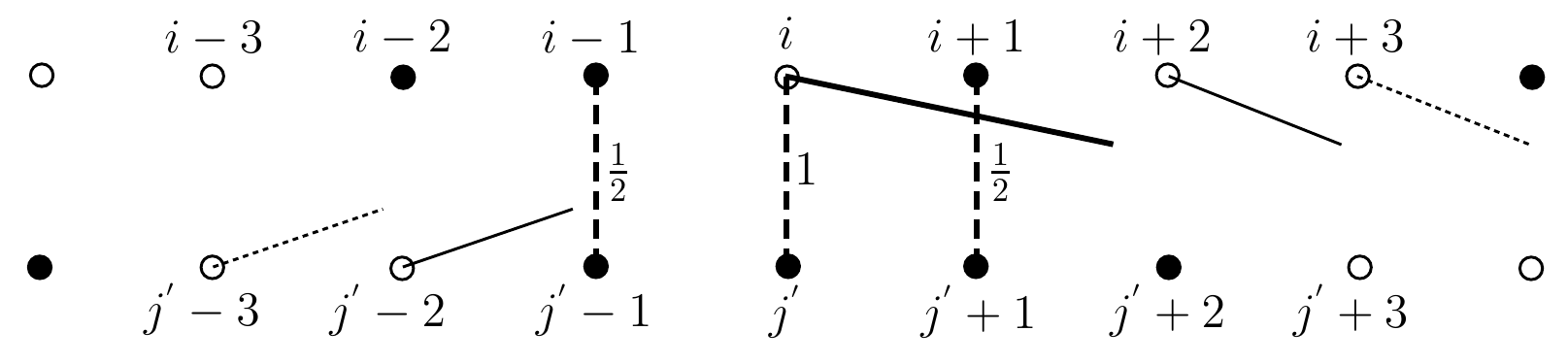} 
\caption{The only possible configuration of $\mcN_p[C(C^*(e_{i, \Cdot}))]$ when $\tau(e_{i, j} \gets C^*(e_{i, \Cdot})) = \big(\frac 12, 1, \frac 12\big)$.
We have $|C(C^*(e_{i, \Cdot}))| = 3$, and $|\mcN_p(e_{i+2})| = |\mcN_p(e_{j'-2})| = 1$ in this configuration.
It also represents the other three symmetric configurations where $|\mcN_p(e_{i+2})| = |\mcN_p(e_{i-2})| = 1$, $|\mcN_p(e_{j'+2})| = |\mcN_p(e_{j'-2})| = 1$,
and $|\mcN_p(e_{j'+2})| = |\mcN_p(e_{i-2})| = 1$, respectively.
(Recall that
the edge $e_{i, j}$ is in bold solid line, 
the edges in vertical bold dashed lines are in $C^*(e_{i, \Cdot})$,
the edges in thin solid lines are in $C(C^*(e_{i, \Cdot}))$, and
the edges in thin dashed lines are in $\mcN_p(C(C^*(e_{i, \Cdot}))$;
the vertices in filled circles are surely incident with no edges of $M$ and
the vertices in hollow circles have uncertain incidence in $M$.)\label{fig401}}
\end{figure}

\item $\tau(e_{i, j} \gets C^*(e_{i, \Cdot})) = \big(\frac 12, 1, \frac 13\big)$:
We have $C(e^*_{i, j'}) \subset C(e^*_{i-1, j'-1})$, and thus $|C(e^*_{i, j'}) \cup C(e^*_{i-1, j'-1})| = 2$ and $|C(C^*(e_{i, \Cdot}))| = 4$.
There is exactly one edge of $M$ incident at either of $i-2$ and $j'-2$ but not both.
We assume $e_{j'-2} \in M$.
If $e_{j'-2}$ is a singleton edge of $M$ or $|\mcN_p(e_{j'-2})| \ge 2$,
then the algorithm $\mcL\mcS$ would replace $e_{i, j}$ and $e_{j'-2}$ by the two parallel edges $e^*_{i, j'}$ and $e^*_{i-1, j'-1}$ to reduce the singleton edges,
a contradiction.
Therefore, we have $e_{j'-3} \in M$ but no edge of $M$ is incident at $j'-4$.
In this sense, there is only one possible edge combination of $C(C^*(e_{i, \Cdot}))$, as shown in Fig.~\ref{fig402} with $e_{j'-2} \in M$,
where the corresponding configuration of $\mcN_p[C(C^*(e_{i, \Cdot}))]$ is also shown.

\begin{figure}[H]
\centering
\includegraphics[width=0.5\linewidth]{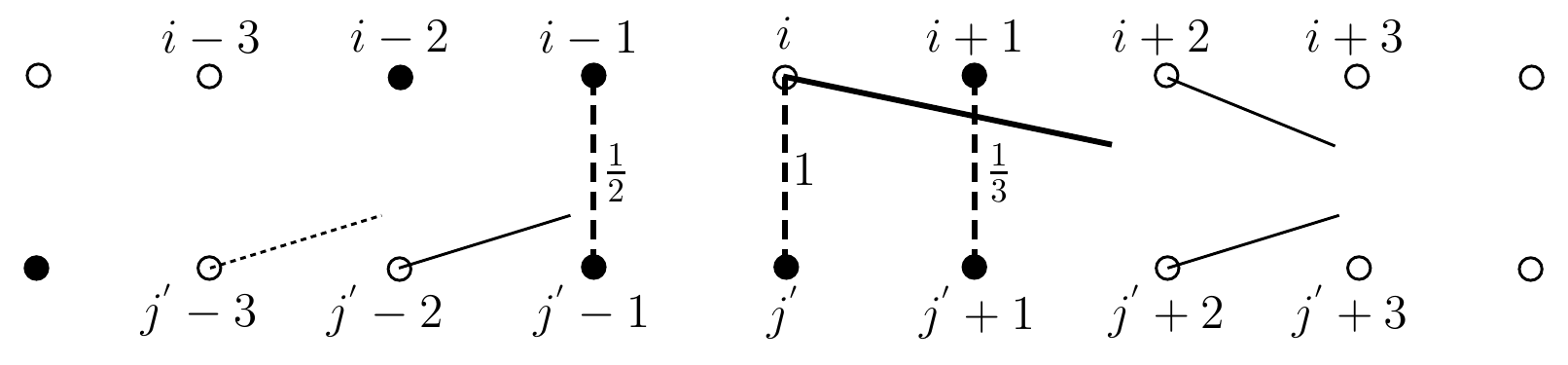} 
\caption{The only possible configuration of $\mcN_p[C(C^*(e_{i, \Cdot}))]$ when $\tau(e_{i, j} \gets C^*(e_{i, \Cdot})) = \big(\frac 12, 1, \frac 13\big)$.
We have $|C(C^*(e_{i, \Cdot}))| = 4$ and $|\mcN_p(e_{j'-2})| = 1$ in this configuration.
It also represents the symmetric configuration where $|\mcN_p(e_{i-2})| = 1$.\label{fig402}}
\end{figure}

\item $\tau(e_{i, j} \gets C^*(e_{i, \Cdot})) = \big(1, \frac 12, \frac 13\big)$:
According to the inequalities (\ref{eq7}) and (\ref{eq8}) of Lemma \ref{lemma406}, we have $|C(C^*(e_{i, \Cdot}))| = 3$.
Since $e_{i,j}$ is a singleton edge of $M$, $e_{j'+1} \in M$;
and either $e_{i+2} \in M$ or $e_{j'+2} \in M$ but no both.
If $e_{i+2} \in M$, then $e_{j'+1}$ is a singleton edge of $M$,
and thus the algorithm $\mcL\mcS$ would replace $e_{i, j}$ and $e_{j'+1}$ by the two parallel edges $e^*_{i, j'}$ and $e^*_{i-1, j'-1}$ to reduce the singleton edges,
a contradiction.
Therefore, $e_{j'+2} \in M$.
Similarly, if $\mcN_p(e_{j'+2}) = \emptyset$ or $|\mcN_p(e_{j'+2})| \ge 2$,
then the algorithm $\mcL\mcS$ would replace the three edges of $C(C^*(e_{i, \Cdot}))$ by the three parallel edges of $C^*(e_{i, \Cdot})$ to reduce the singleton edges,
a contradiction.
This leaves the only possible configuration with $|\mcN_p(e_{j'+2})| = 1$, as shown in Fig.~\ref{fig403},
where the corresponding configuration of $\mcN_p[C(C^*(e_{i, \Cdot}))]$ is also shown.

\begin{figure}[H]
\centering
\includegraphics[width=0.5\linewidth]{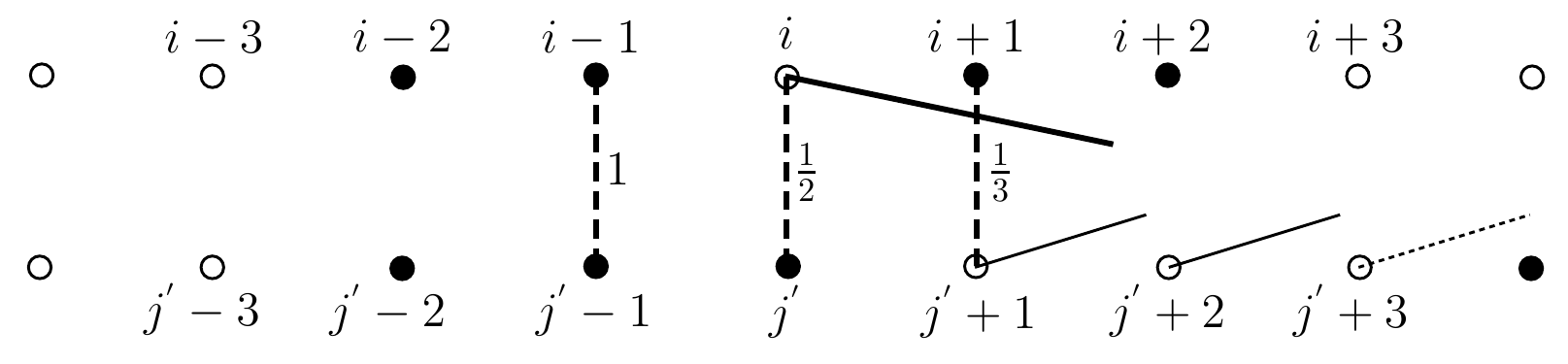} 
\caption{The only possible configuration of $\mcN_p[C(C^*(e_{i, \Cdot}))]$ when $\tau(e_{i, j} \gets C^*(e_{i, \Cdot})) = \big(1, \frac 12, \frac 13\big)$.
We have $|C(C^*(e_{i, \Cdot}))| = 3$ and $|\mcN_p(e_{j'+2})| = 1$ in this configuration.\label{fig403}}
\end{figure}

\item $\tau(e_{i, j} \gets C^*(e_{i, \Cdot})) = \big(1, \frac 12, \frac 14\big)$:
We have $e_{j'} \notin M$ and $|C(C^*(e_{i, \Cdot}))| = 4$.
Therefore, $e_{i+2}, e_{j'+1}, e_{j'+2} \in M$.
If $|\mcN_p(e_{j'+2})| \ge 1$,
then the algorithm $\mcL\mcS$ would replace $e_{i, j}$ and $e_{j'+1}$ by the two parallel edges $e^*_{i, j'}$ and $e^*_{i-1, j'-1}$ to reduce the singleton edges,
a contradiction.
Therefore, $\mcN_p(e_{j'+2}) = \emptyset$, that is, $e_{j'+3} \notin M$.
There is only one possible edge combination of $C(C^*(e_{i, \Cdot}))$, as shown in Fig.~\ref{fig404},
where the corresponding configuration of $\mcN_p[C(C^*(e_{i, \Cdot}))]$ is also shown.

\begin{figure}[H]
\centering
\includegraphics[width=0.5\linewidth]{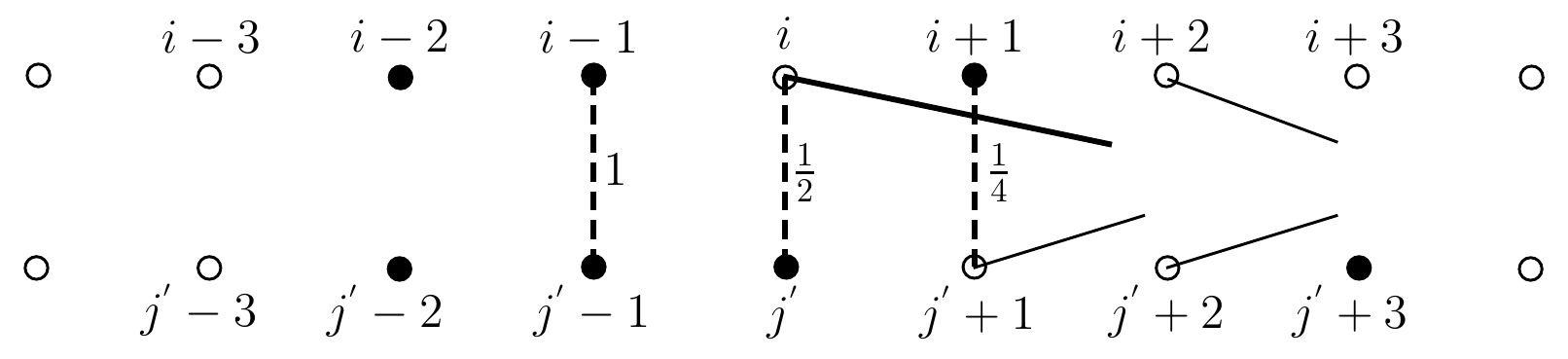} 
\caption{The only possible configuration of $\mcN_p[C(C^*(e_{i, \Cdot}))]$ when $\tau(e_{i, j} \gets C^*(e_{i, \Cdot})) = \big(1, \frac 12, \frac 14\big)$.
We have $|C(C^*(e_{i, \Cdot}))| = 4$ and $\mcN_p(e_{j'+2}) = \emptyset$ in this configuration.\label{fig404}}
\end{figure}

\item $\tau(e_{i, j} \gets C^*(e_{i, \Cdot})) = \big(\frac 13, 1, \frac 13\big)$:
According to the inequalities (\ref{eq7}) and (\ref{eq8}) of Lemma \ref{lemma406}, we have $|C(C^*(e_{i, \Cdot}))| = 5$. 
There is only one possible edge combination of $C(C^*(e_{i, \Cdot}))$, which is shown in Fig.~\ref{fig405},
where any configuration of $\mcN_p[C(C^*(e_{i, \Cdot}))]$ is possible.

\begin{figure}[H]
\centering
\includegraphics[width=0.5\linewidth]{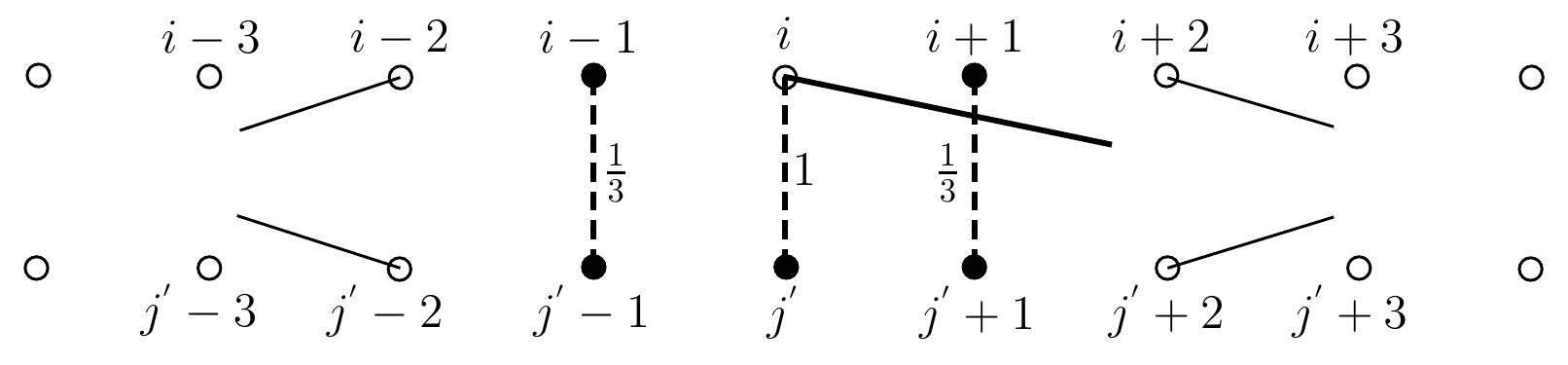} 
\caption{The only possible configuration of $\mcN_p[C(C^*(e_{i, \Cdot}))]$ when $\tau(e_{i, j} \gets C^*(e_{i, \Cdot})) = \big(\frac 13, 1, \frac 13\big)$,
where $|C(C^*(e_{i, \Cdot}))| = 5$.\label{fig405}}
\end{figure}

\item $\tau(e_{i, j} \gets C^*(e_{i, \Cdot})) = \big(\frac 12, \frac 12, \frac 13\big)$:
According to Lemma \ref{lemma407}, we have $C(e^*_{i, j'}) \cap C(e^*_{i-1, j'-1}) = \{e_{i, j}\}$;
thus $e_{j'+1} \in M$, either $e_{i'-2} \in M$ or $e_{j'-2} \in M$ but no both, either $e_{i+2} \in M$ or $e_{j'+2} \in M$ but no both,
and $|C(C^*(e_{i, \Cdot}))| = 4$.
We assume $e_{j'-2} \in M$ ($e_{i'-2} \in M$ is discussed the same).
When $e_{i+2} \in M$, $e_{j'+1}$ is a singleton edge of $M$.
If $e_{j'-2}$ is also a singleton edge of $M$,
then the algorithm $\mcL\mcS$ would replace the four edges in $C(C^*(e_{i, \Cdot}))$ by the three parallel edges in $C^*(e_{i, \Cdot})$ and $e^*_{i_1, j_1}$
to reduce the singleton edges, a contradiction.
Therefore in this case we have $|\mcN_p(e_{j'-2})| \ge 1$, that is, $e_{j'-3} \in M$.
Similarly, if $e_{i+2}$ is a singleton edge of $M$ or $|\mcN_p(e_{i+2})| \ge 2$,
then the algorithm $\mcL\mcS$ would replace the three edges $e_{i,j}, e_{j'+1}, e_{i+2}$ by the two parallel edges $e^*_{i,j'}, e^*_{i+1,j'+1}$ and $e^*_{i_1, j_1}$
to reduce the singleton edges, a contradiction.
That is, $e_{i+3} \in M$ but $e_{i+4} \notin M$.
This edge combination of $C(C^*(e_{i, \Cdot}))$ is shown in Fig.~\ref{fig406b},
where the corresponding configuration of $\mcN_p[C(C^*(e_{i, \Cdot}))]$ is also shown.

When $e_{j'+2} \in M$, for the same reason, if $|\mcN_p(e_{j'+2})| \ne 1$ then $e_{j'-2}$ must not be a singleton edge of $M$.
This edge combination of $C(C^*(e_{i, \Cdot}))$ is shown in Fig.~\ref{fig406a},
where the corresponding configuration of $\mcN_p[C(C^*(e_{i, \Cdot}))]$ is also shown.

\begin{figure}[H]
\captionsetup[subfigure]{justification=centering}
\begin{subfigure}{0.46\textwidth}
\includegraphics[width=\linewidth]{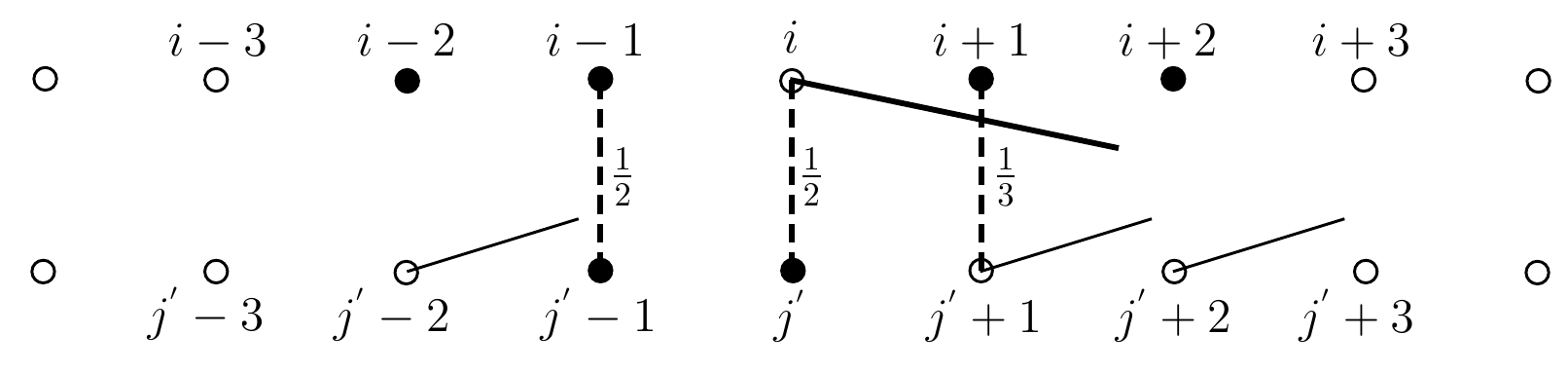} 
\caption{If $|\mcN_p(e_{j'+2})| \ne 1$ then $|\mcN_p(e_{j'-2})| \ge 1$.\label{fig406a}}
\end{subfigure}
\hspace*{\fill}
\begin{subfigure}{0.46\textwidth}
\includegraphics[width=\linewidth]{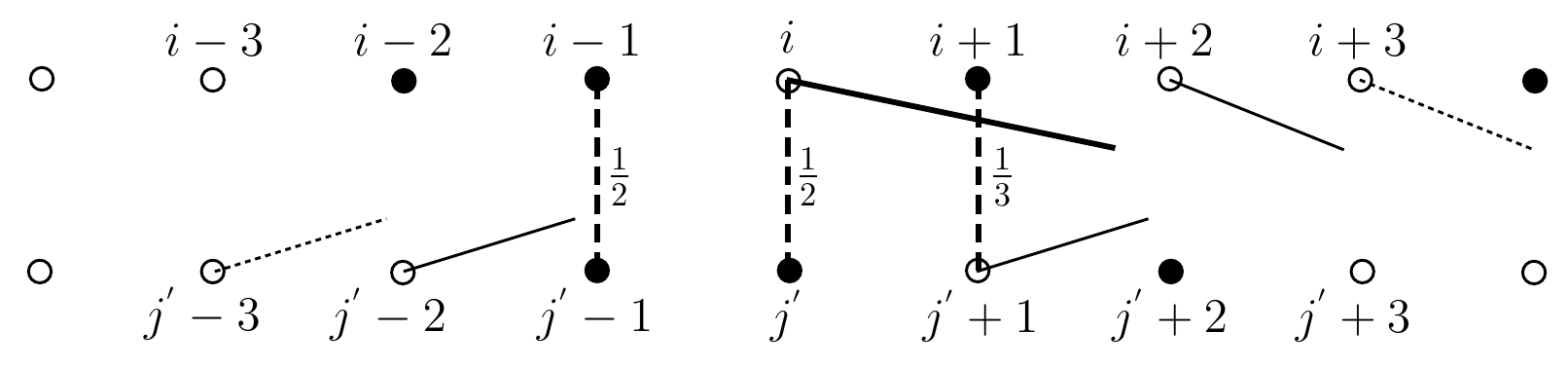} 
\caption{$|\mcN_p(e_{i+2})| = 1$ and $|\mcN_p(e_{j'-2})| \ge 1$.} \label{fig406b}
\end{subfigure}
\caption{The two possible configurations of $\mcN_p[C(C^*(e_{i, \Cdot}))]$ when $\tau(e_{i, j} \gets C^*(e_{i, \Cdot})) = \big(\frac 12, \frac 12, \frac 13\big)$.
We have $|C(C^*(e_{i, \Cdot}))| = 4$.
They also represent the symmetric case where $e_{i'-2} \in M$ instead of $e_{j'-2} \in M$.\label{fig406}}
\end{figure}

\item $\tau(e_{i, j} \gets C^*(e_{i, \Cdot})) = \big(\frac 12, \frac 12, \frac 14\big)$:
According to Lemma \ref{lemma407}, we have $C(e^*_{i, j'}) \cap C(e^*_{i-1, j'-1}) = \{e_{i, j}\}$;
thus $e_{j'+1} \in M$, either $e_{i'-2} \in M$ or $e_{j'-2} \in M$ but no both, $e_{i+2}, e_{j'+2} \in M$,
and $|C(C^*(e_{i, \Cdot}))| = 5$.
We assume $e_{j'-2} \in M$ ($e_{i'-2} \in M$ is discussed the same).
If $e_{j'-2}$ is a singleton edge of $M$ and $|\mcN_p(e_{j'+2})| \ge 1$,
then the algorithm $\mcL\mcS$ would replace the three edges $e_{i, j}$, $e_{j'-2}$, and $e_{j'+1}$ by
$e^*_{i_1, j_1}$ and the two parallel edges $e^*_{i, j'}$ and $e^*_{i-1, j'-1}$ to reduce the singleton edges, a contradiction.
Therefore, $|\mcN_p(e_{j'-2})| \ge 1$ (shown in Fig.~\ref{fig407b}) or $\mcN_p(e_{j'+2}) = \emptyset$ (shown in Fig.~\ref{fig407a}).
These two edge combinations of $C(C^*(e_{i, \Cdot}))$ are shown in Fig.~\ref{fig407a} and Fig.~\ref{fig407b}, respectively,
where the corresponding configurations of $\mcN_p[C(C^*(e_{i, \Cdot}))]$ are also shown.

Between the two configurations shown in Fig.~\ref{fig407a} and Fig.~\ref{fig407b},
we notice that for every edge $e \in C(C^*(e_{i, \Cdot})) - \{e_{i,j}\}$,
the largest possible value for $\omega(e)$ in Fig.~\ref{fig407a} is at least as large as in Fig.~\ref{fig407b}.
Since we are interested in the worst-case analysis,
we say Fig.~\ref{fig407b} is {\em shadowed} by Fig.~\ref{fig407a} and we will consider Fig.~\ref{fig407a} only in the sequel.

\begin{figure}[H]
\captionsetup[subfigure]{justification=centering}
\begin{subfigure}{0.46\textwidth}
\includegraphics[width=\linewidth]{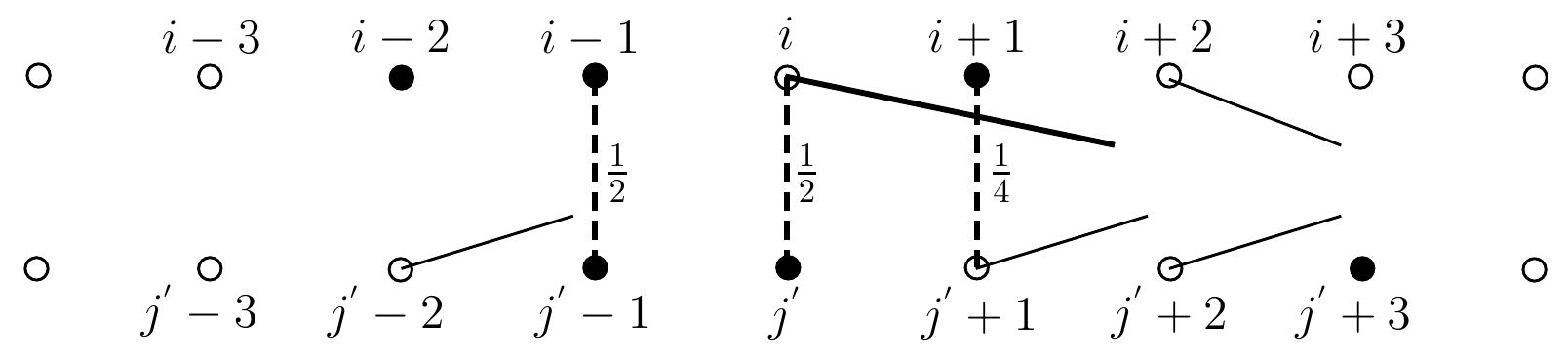} 
\caption{$\mcN_p(e_{j'+2}) = \emptyset$.\label{fig407a}}
\end{subfigure}
\hspace*{\fill}
\begin{subfigure}{0.46\textwidth}
\includegraphics[width=\linewidth]{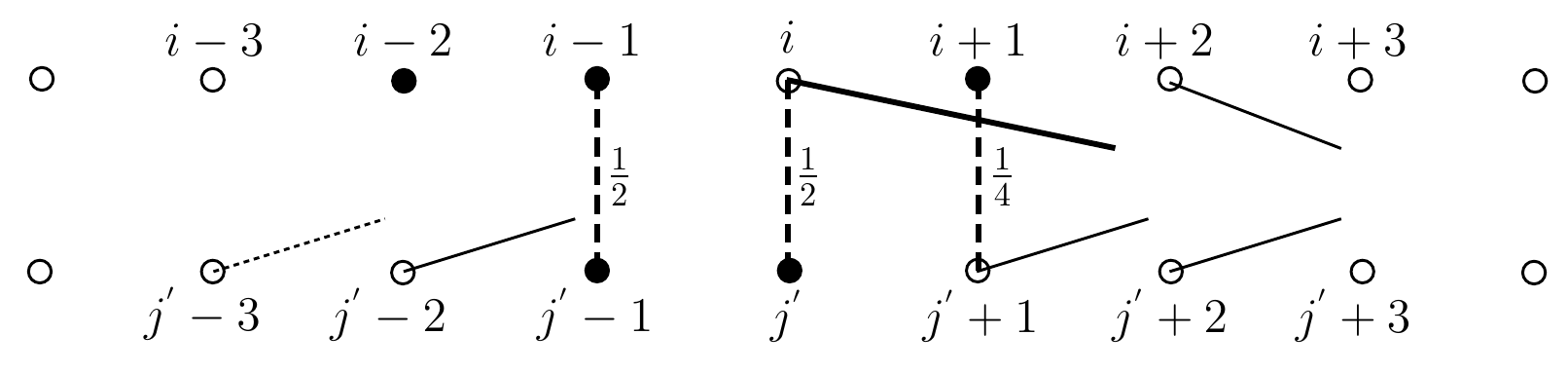} 
\caption{$|\mcN_p(e_{j'-2})| \ge 1$.\label{fig407b}}
\end{subfigure}
\caption{The two possible configurations of $\mcN_p[C(C^*(e_{i, \Cdot}))]$ when $\tau(e_{i, j} \gets C^*(e_{i, \Cdot})) = \big(\frac 12, \frac 12, \frac 14\big)$.
They are associated with the only possible edge combination of $C(C^*(e_{i, \Cdot}))$ with $|C(C^*(e_{i, \Cdot}))| = 5$,
which also represents the symmetric case where $e_{i'-2} \in M$ instead of $e_{j'-2} \in M$.
The first configuration shadows the second one.\label{fig407}}
\end{figure}

\item $\tau(e_{i, j} \gets C^*(e_{i, \Cdot})) = \big(\frac 13, \frac 12, \frac 13\big)$:
According to the inequalities (\ref{eq7}) and (\ref{eq8}) of Lemma \ref{lemma406}, we have $4 \le |C(C^*(e_{i, \Cdot}))| \le 5$.
Since $i-1$ and $i+1$ are symmetric with respect to $i$, we only discuss one of them.
We have either $e_{j'} \in M$ or $e_{j'-1} \in M$, but not both.

When $e_{j'} \in M$, then either $e_{i-2} \in M$ or $e_{j'-2} \in M$, but not both.
We assume $e_{j'-2} \in M$.
Similarly, either $e_{i+2} \in M$ or $e_{j'+2} \in M$, but not both.
We assume $e_{i+2} \in M$.
If $e_{i+2}$ is a singleton edge of $M$ or $|\mcN_p(e_{i+2})| \ge 2$,
then the algorithm $\mcL\mcS$ would replace the three edges $e_{i, j}$, $e_{j'}$, $e_{i+2}$ by
$e^*_{i_1, j_1}$ and the two parallel edges $e^*_{i, j'}$ and $e^*_{i+1, j'+1}$ to reduce the singleton edges, a contradiction.
Therefore, $|\mcN_p(e_{i+2})| = 1$;
for the same reason, $|\mcN_p(e_{j'-2})| = 1$.
This edge combination of $C(C^*(e_{i, \Cdot}))$ is shown in Fig.~\ref{fig408a},
where the corresponding configuration of $\mcN_p[C(C^*(e_{i, \Cdot}))]$ is also shown.

When $e_{j'-1} \in M$, then still either $e_{i-2} \in M$ or $e_{j'-2} \in M$, but not both.
On the other side, $e_{i+2} \in M$ and $e_{j'+2} \in M$.
When $e_{i-2} \in M$, $e_{j'-1}$ is a singleton edge of $M$;
and therefore $|\mcN_p(e_{i'-2})| = 1$.
This edge combination of $C(C^*(e_{i, \Cdot}))$ is shown in Fig.~\ref{fig408b},
where the corresponding configuration of $\mcN_p[C(C^*(e_{i, \Cdot}))]$ is also shown.

When $e_{j'-2} \in M$, the edge combination of $C(C^*(e_{i, \Cdot}))$ is shown in Fig.~\ref{fig408c},
where the corresponding configuration of $\mcN_p[C(C^*(e_{i, \Cdot}))]$ is also shown.

\begin{figure}[H]
\captionsetup[subfigure]{justification=centering}
\begin{subfigure}{0.46\textwidth}
\includegraphics[width=\linewidth]{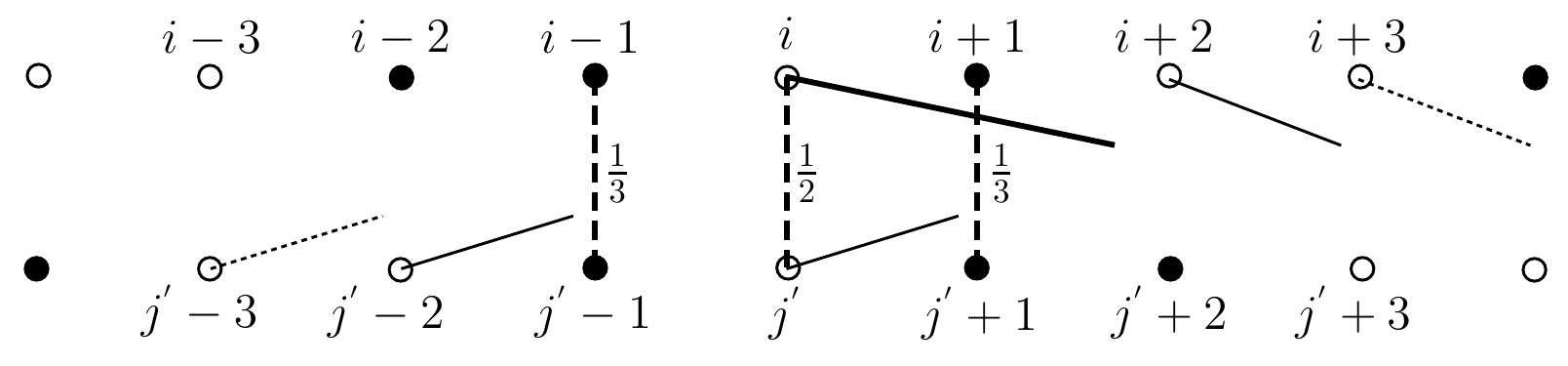} 
\caption{$|C(C^*(e_{i, \Cdot}))| = 4$ and $|\mcN_p(e_{i+2})| = |\mcN_p(e_{j'-2})| = 1$.\label{fig408a}}
\end{subfigure}
\hspace*{\fill}
\begin{subfigure}{0.46\textwidth}
\includegraphics[width=\linewidth]{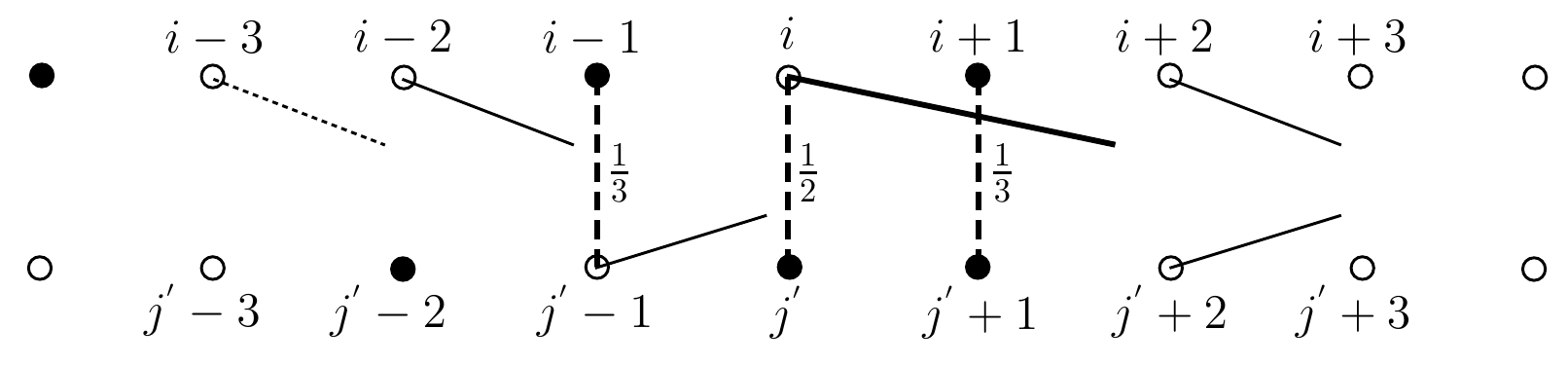} 
\caption{$|C(C^*(e_{i, \Cdot}))| = 5$ and $|\mcN_p(e_{i-2})| = 1$.\label{fig408b}}
\end{subfigure}
\centering
\begin{subfigure}{0.48\textwidth}
\includegraphics[width=\linewidth]{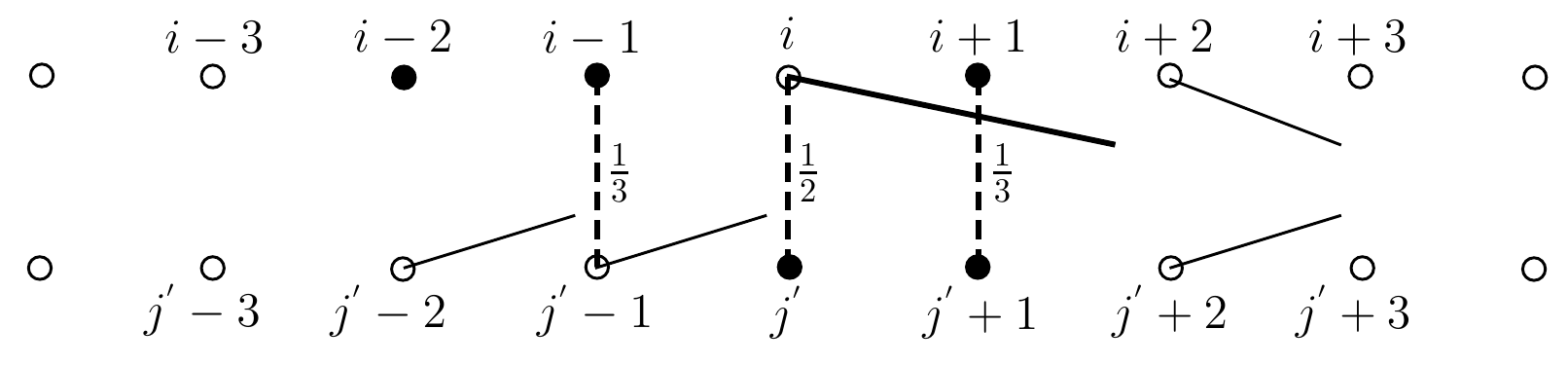} 
\caption{$|C(C^*(e_{i, \Cdot}))| = 5$.\label{fig408c}}
\end{subfigure}
\caption{The three possible configurations of $\mcN_p[C(C^*(e_{i, \Cdot}))]$ when $\tau(e_{i, j} \gets C^*(e_{i, \Cdot})) = \big(\frac 13, \frac 12, \frac 13\big)$,
associated with three possible edge combinations of $C(C^*(e_{i, \Cdot}))$ with $|C(C^*(e_{i, \Cdot}))| = 4, 5, 5$, respectively.
The configuration in Fig.~\ref{fig408a} also represents the symmetric case where $e_{i-2} \in M$ instead of $e_{j'-2} \in M$
and/or $e_{j'+2} \in M$ instead of $e_{i+2} \in M$.\label{fig408}}
\end{figure}

\item $\tau(e_{i, j} \gets C^*(e_{i, \Cdot})) = \big(\frac 12, \frac 13, \frac 13\big)$:
According to the inequalities (\ref{eq7}) and (\ref{eq8}) of Lemma \ref{lemma406}, we have $3 \le |C(C^*(e_{i, \Cdot}))| \le 4$.
If $|C(C^*(e_{i, \Cdot}))| = 3$, then the algorithm $\mcL\mcS$ would replace the three edges of $C(C^*(e_{i, \Cdot}))$ by
$e^*_{i_1, j_1}$ and the three parallel edges in $C^*(e_{i, \Cdot})$ to expand $M$, a contradiction.
Therefore, $|C(C^*(e_{i, \Cdot}))| = 4$.
From $e_{j'-1}, e_{j'+1} \in M$, we know that either $e_{i+2} \in M$ or $e_{j'+2} \in M$ but not both.
If $e_{i+2} \in M$, then all three edges $e_{j'-1}, e_{i,j}, e_{j'+1}$ are singleton edges of $M$,
and the algorithm $\mcL\mcS$ would replace the four edges of $C(C^*(e_{i, \Cdot}))$ by $e^*_{i_1, j_1}$ and the three parallel edges of $C^*(e_{i, \Cdot})$
to reduce the singleton edges, a contradiction.
Therefore, $e_{j'+2} \in M$, which then implies $|\mcN_p(e_{j'+2})| = 1$.
This only edge combination of $C(C^*(e_{i, \Cdot}))$ is shown in Fig.~\ref{fig409},
where the corresponding configuration of $\mcN_p[C(C^*(e_{i, \Cdot}))]$ is also shown.

\begin{figure}[H]
\centering
\includegraphics[width=0.5\linewidth]{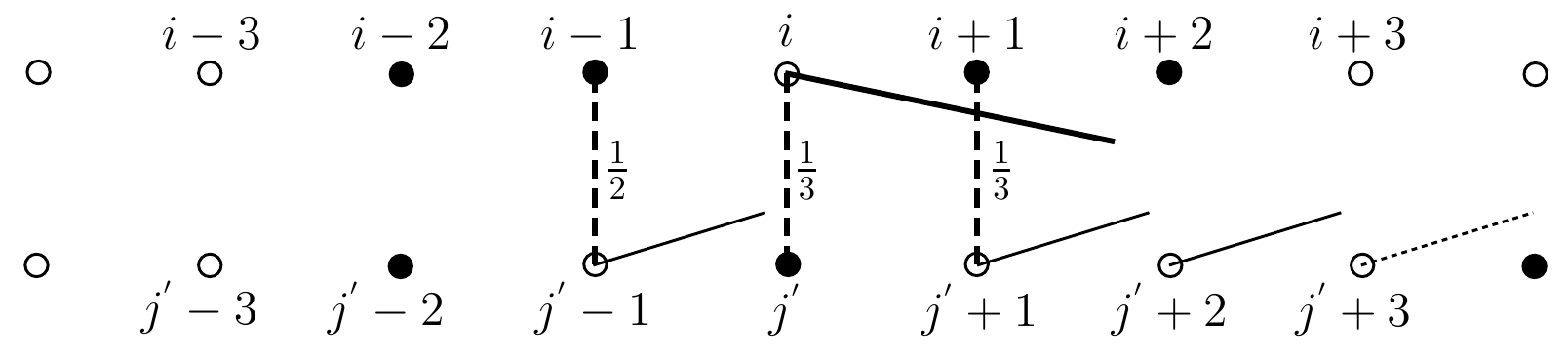} 
\caption{The only possible configuration of $\mcN_p[C(C^*(e_{i, \Cdot}))]$ when $\tau(e_{i, j} \gets C^*(e_{i, \Cdot})) = \big(\frac 12, \frac 13, \frac 13\big)$.
We have $|C(C^*(e_{i, \Cdot}))| = 4$ and $|\mcN_p(e_{j'+2})| = 1$.\label{fig409}}
\end{figure}

\item $\tau(e_{i, j} \gets C^*(e_{i, \Cdot})) = \big(\frac 13, \frac 12, \frac 14\big)$:
According to the inequalities (\ref{eq7}) and (\ref{eq8}) of Lemma \ref{lemma406}, we have $5 \le |C(C^*(e_{i, \Cdot}))| \le 6$.
Note that either $e_{j'} \in M$ or $e_{j'+1} \in M$ but not both.

When $e_{j'} \in M$, we have two symmetric cases where $e_{i-2} \in M$ and $e_{j'-2} \in M$, respectively;
and we assume $e_{j'-2} \in M$.
We conclude that $e_{j'-2}$ must not be a singleton edge of $M$ or $|\mcN_p(e_{j'-2})| \ge 2$.
This edge combination of $C(C^*(e_{i, \Cdot}))$ is shown in Fig.~\ref{fig410a},
where the corresponding configuration of $\mcN_p[C(C^*(e_{i, \Cdot}))]$ is also shown.

When $e_{j'+1} \in M$, both $e_{i-2} \in M$ and $e_{j'-2} \in M$.
This edge combination of $C(C^*(e_{i, \Cdot}))$ is shown in Fig.~\ref{fig410b},
where the corresponding configuration of $\mcN_p[C(C^*(e_{i, \Cdot}))]$ is also shown.

\begin{figure}[H]
\captionsetup[subfigure]{justification=centering}
\begin{subfigure}{0.46\textwidth}
\includegraphics[width=\linewidth]{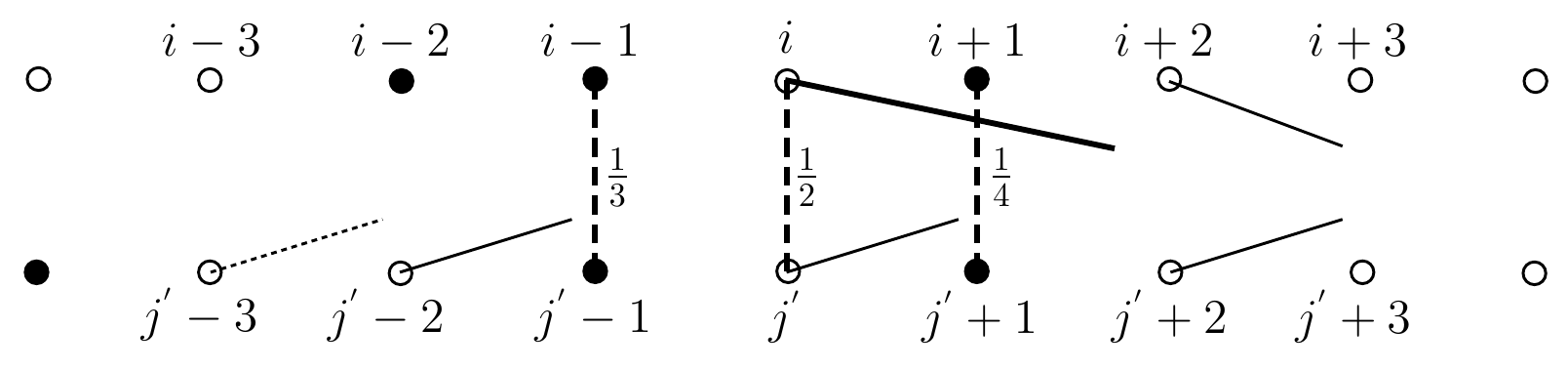} 
\caption{$|C(C^*(e_{i, \Cdot}))| = 5$ and $|\mcN_p(e_{j'-2})| = 1$.\label{fig410a}}
\end{subfigure}
\hspace*{\fill}
\begin{subfigure}{0.46\textwidth}
\includegraphics[width=\linewidth]{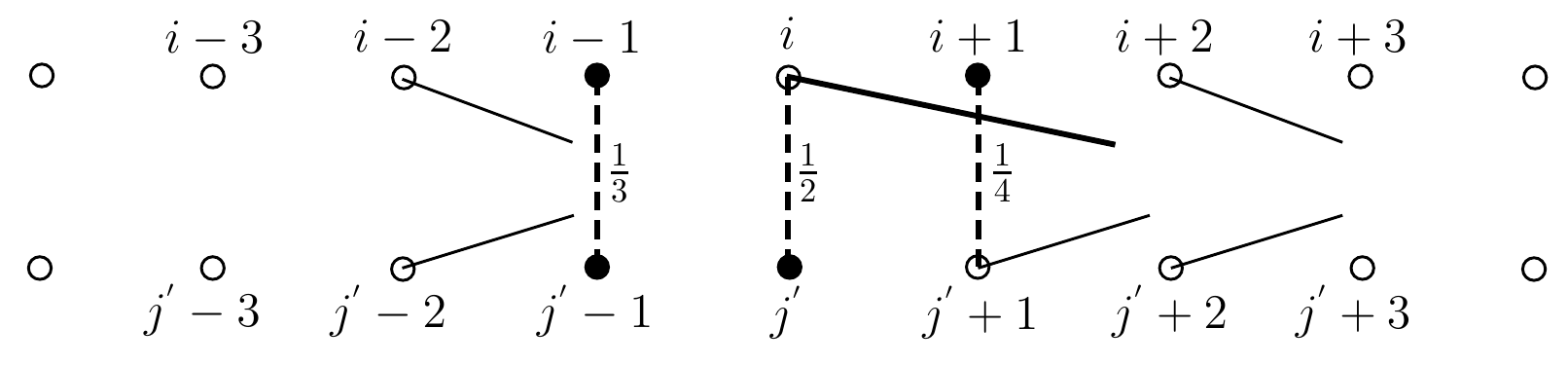} 
\caption{$|C(C^*(e_{i, \Cdot}))| = 6$.\label{fig410b}}
\end{subfigure}
\caption{The two possible configurations of $\mcN_p[C(C^*(e_{i, \Cdot}))]$ when $\tau(e_{i, j} \gets C^*(e_{i, \Cdot})) = \big(\frac 13, \frac 12, \frac 14\big)$,
associated with two possible edge combinations of $C(C^*(e_{i, \Cdot}))$ with $|C(C^*(e_{i, \Cdot}))| = 5, 6$, respectively.
The configuration in Fig.~\ref{fig410a} also represents the symmetric case where $e_{i-2} \in M$ instead of $e_{j'-2} \in M$.\label{fig410}}
\end{figure}

\item $\tau(e_{i, j} \gets C^*(e_{i, \Cdot})) = \big(\frac 12, \frac 13, \frac 14\big)$:
According to the inequalities (\ref{eq7}) and (\ref{eq8}) of Lemma \ref{lemma406}, we have $4 \le |C(C^*(e_{i, \Cdot}))| \le 5$.
Note that either $e_{j'-1} \notin M$ or $e_{j'} \notin M$ but not both, and $e_{j'+1} \in M$.

When $e_{j'-1} \notin M$, then either $e_{i+2} \in M$ or $e_{j'+2} \in M$ but not both.
When $e_{j'+2} \in M$, the edge combination of $C(C^*(e_{i, \Cdot}))$ is shown in Fig.~\ref{fig411a},
where the corresponding configuration of $\mcN_p[C(C^*(e_{i, \Cdot}))]$ is also shown.

When $e_{i+2} \in M$, we conclude that $e_{i+2}$ should not be a singleton edge of $M$;
the edge combination of $C(C^*(e_{i, \Cdot}))$ is shown in Fig.~\ref{fig411b},
where the corresponding configuration of $\mcN_p[C(C^*(e_{i, \Cdot}))]$ is also shown.

When $e_{j'} \notin M$, then both $e_{i+2} \in M$ and $e_{j'+2} \in M$;
we conclude that $\mcN_p(e_{j'+2}) = \emptyset$.
This edge combination of $C(C^*(e_{i, \Cdot}))$ is shown in Fig.~\ref{fig411c},
where the corresponding configuration of $\mcN_p[C(C^*(e_{i, \Cdot}))]$ is also shown.

\begin{figure}[H]
\captionsetup[subfigure]{justification=centering}
\begin{subfigure}{0.46\textwidth}
\includegraphics[width=\linewidth]{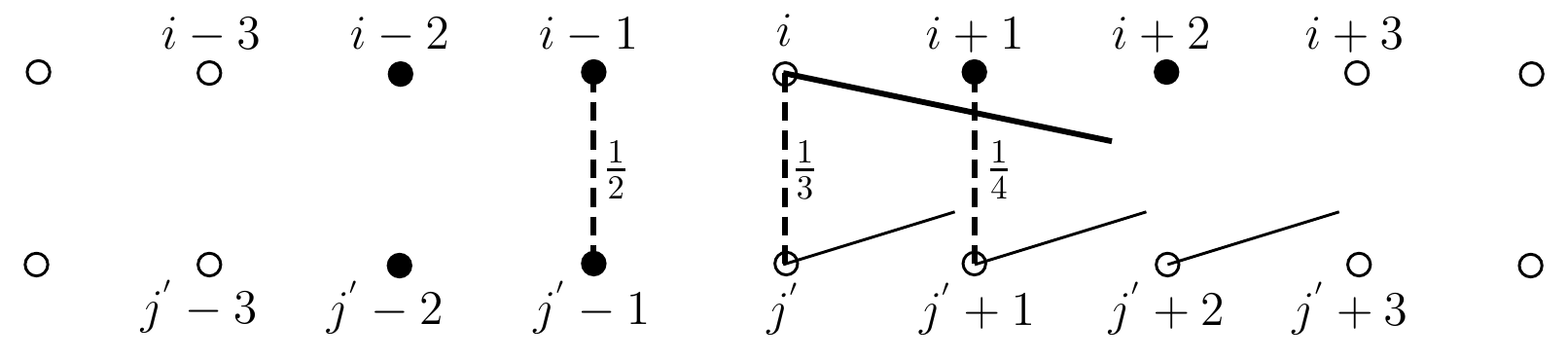} 
\caption{$|C(C^*(e_{i, \Cdot}))| = 4$.\label{fig411a}}
\end{subfigure}
\hspace*{\fill}
\begin{subfigure}{0.46\textwidth}
\includegraphics[width=\linewidth]{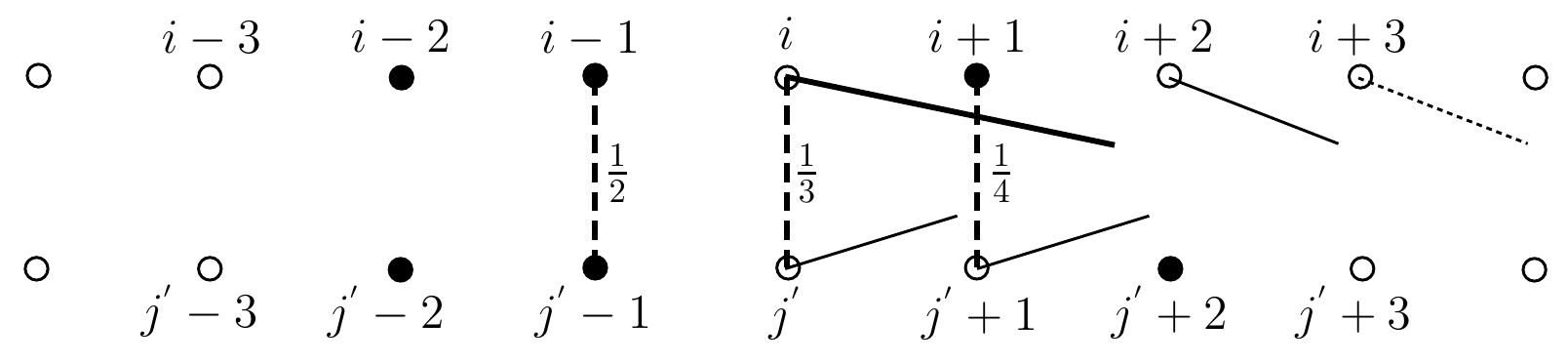} 
\caption{$|C(C^*(e_{i, \Cdot}))| = 4$ and $|\mcN_p(e_{i+2})| \ge 1$.\label{fig411b}}
\end{subfigure}
\centering
\begin{subfigure}{0.48\textwidth}
\includegraphics[width=\linewidth]{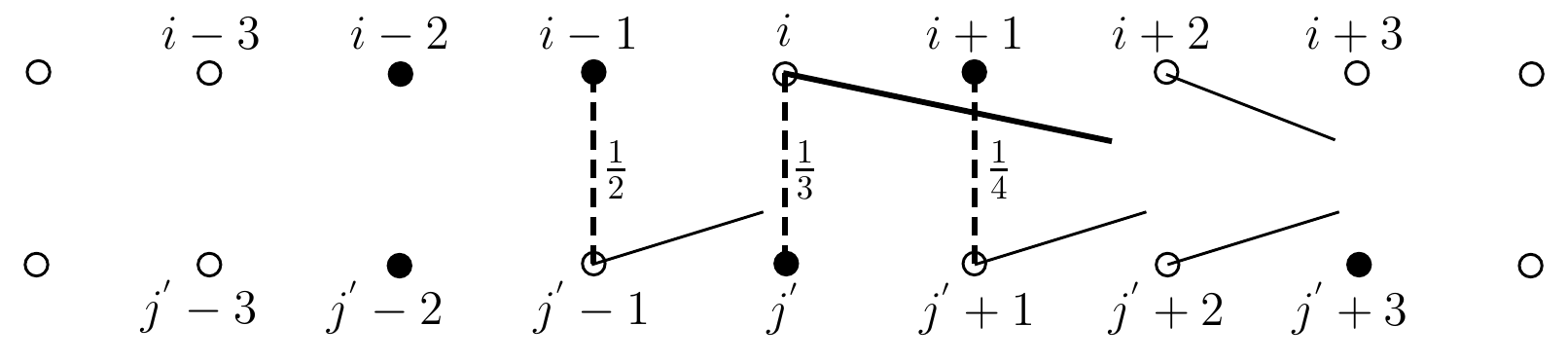} 
\caption{$|C(C^*(e_{i, \Cdot}))| = 5$ and $\mcN_p(e_{j'+2}) = \emptyset$.\label{fig411c}}
\end{subfigure}
\caption{The three possible configurations of $\mcN_p[C(C^*(e_{i, \Cdot}))]$ when $\tau(e_{i, j} \gets C^*(e_{i, \Cdot})) = \big(\frac 12, \frac 13, \frac 14\big)$,
associated with three possible edge combinations of $C(C^*(e_{i, \Cdot}))$ with $|C(C^*(e_{i, \Cdot}))| = 4, 4, 5$, respectively.\label{fig411}}
\end{figure}

\item $\tau(e_{i, j} \gets C^*(e_{i, \Cdot})) = \big(\frac 12, \frac 14, \frac 13\big)$:
This ordered value combination is impossible due to the edge $e_{i, j}$ being a singleton edge of $M$.

\item $\tau(e_{i, j} \gets C^*(e_{i, \Cdot})) = \big(\frac 12, \frac 13, \frac 15\big)$:
We have $e_{j'}, e_{j'+1}, e_{j'+2}, e_{i+2} \in M$, giving rise to $|C(C^*(e_{i, \Cdot}))| = 5$.
This only edge combination of $C(C^*(e_{i, \Cdot}))$ is shown in Fig.~\ref{fig412},
where the corresponding configuration of $\mcN_p[C(C^*(e_{i, \Cdot}))]$ is also shown.

\begin{figure}[H]
\centering
\includegraphics[width=0.5\linewidth]{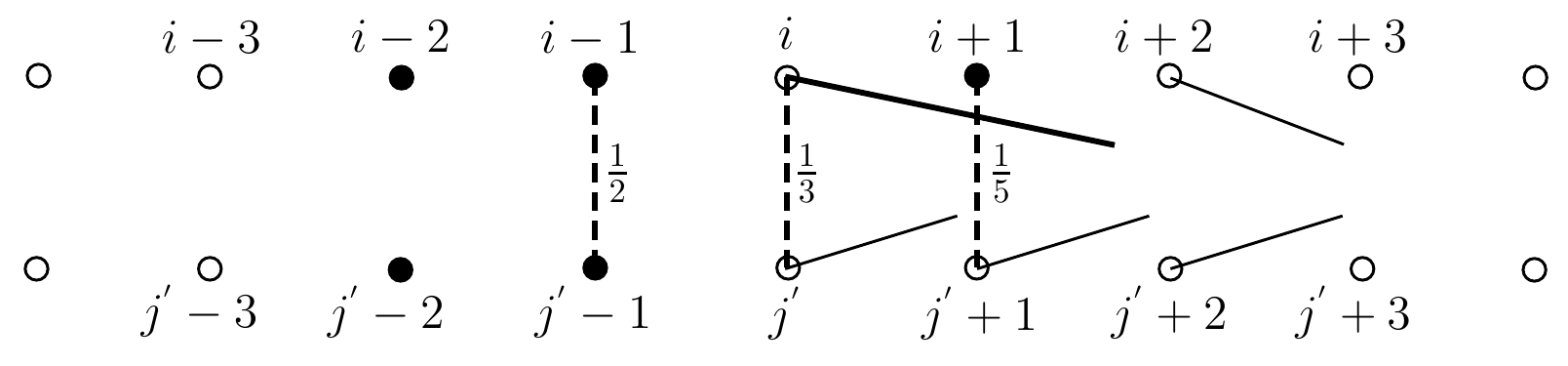} 
\caption{The only possible configuration of $\mcN_p[C(C^*(e_{i, \Cdot}))]$ when $\tau(e_{i, j} \gets C^*(e_{i, \Cdot})) = \big(\frac 12, \frac 13, \frac 15\big)$,
where $|C(C^*(e_{i, \Cdot}))| = 5$.\label{fig412}}
\end{figure}

\item $\tau(e_{i, j} \gets C^*(e_{i, \Cdot})) = \big(\frac 12, \frac 14, \frac 14\big)$:
This ordered value combination is impossible due to the edge $e_{i, j}$ being a singleton edge of $M$.

\item $\tau(e_{i, j} \gets C^*(e_{i, \Cdot})) = \big(\frac 14, \frac 12, \frac 14\big)$:
Since the edge $e_{j'}$ has to be in $M$, we have both $e_{i-2} \in M$ and $e_{j'-2} \in M$,
and both $e_{i+2} \in M$ and $e_{j'+2} \in M$, giving rise to $|C(C^*(e_{i, \Cdot}))| = 6$.
This only edge combination of $C(C^*(e_{i, \Cdot}))$ is shown in Fig.~\ref{fig413},
where the corresponding configuration of $\mcN_p[C(C^*(e_{i, \Cdot}))]$ is also shown.

\begin{figure}[H]
\centering
\includegraphics[width=0.5\linewidth]{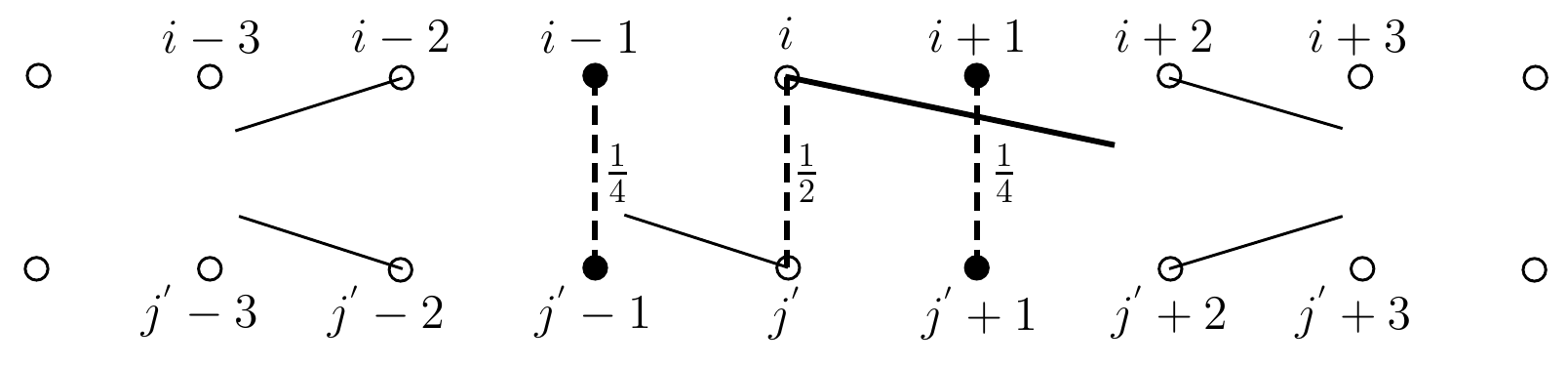} 
\caption{The only possible configuration of $\mcN_p[C(C^*(e_{i, \Cdot}))]$ when $\tau(e_{i, j} \gets C^*(e_{i, \Cdot})) = \big(\frac 14, \frac 12, \frac 14\big)$,
where $|C(C^*(e_{i, \Cdot}))| = 6$.\label{fig413}}
\end{figure}

\item $\tau(e_{i, j} \gets C^*(e_{i, \Cdot})) = \big(\frac 13, \frac 13, \frac 13\big)$:
According to the inequalities (\ref{eq7}) and (\ref{eq8}) of Lemma \ref{lemma406}, we have $4 \le |C(C^*(e_{i, \Cdot}))| \le 5$.
Note that exactly one of the three vertices $j'-1, j, j'+1$ is not incident with any edge of $M$,
we thus consider two cases where $e_{j'} \notin M$ and $e_{j'+1} \notin M$ ($e_{j'-1} \notin M$ is symmetric to $e_{j'+1} \in M$), respectively.

When $e_{j'+1} \notin M$, then either $e_{i+2} \in M$ or $e_{j'+2} \in M$ but not both,
while $e_{i-2} \notin M$ and $e_{j'-2} \notin M$.
We assume $e_{i+2} \in M$, which implies $e_{i+2}$ should not be a singleton edge of $M$.
This edge combination of $C(C^*(e_{i, \Cdot}))$ is shown in Fig.~\ref{fig414a},
where the corresponding configuration of $\mcN_p[C(C^*(e_{i, \Cdot}))]$ is also shown.

When $e_{j'} \notin M$, then either $e_{i+2} \in M$ or $e_{j'+2} \in M$ but not both,
and either $e_{i-2} \in M$ or $e_{j'-2} \in M$ but not both.
Four different combinations of their memberships of $M$ give rise to $0, 1, 2$ singleton edges between $e_{j'-1}$ and $e_{j'+1}$.
These three edge combinations of $C(C^*(e_{i, \Cdot}))$ are shown in Fig.~\ref{fig414b}, Fig.~\ref{fig414c}, Fig.~\ref{fig414d}, respectively,
where the corresponding configuration of $\mcN_p[C(C^*(e_{i, \Cdot}))]$ is also shown.

\begin{figure}[H]
\captionsetup[subfigure]{justification=centering}
\begin{subfigure}{0.46\textwidth}
\includegraphics[width=\linewidth]{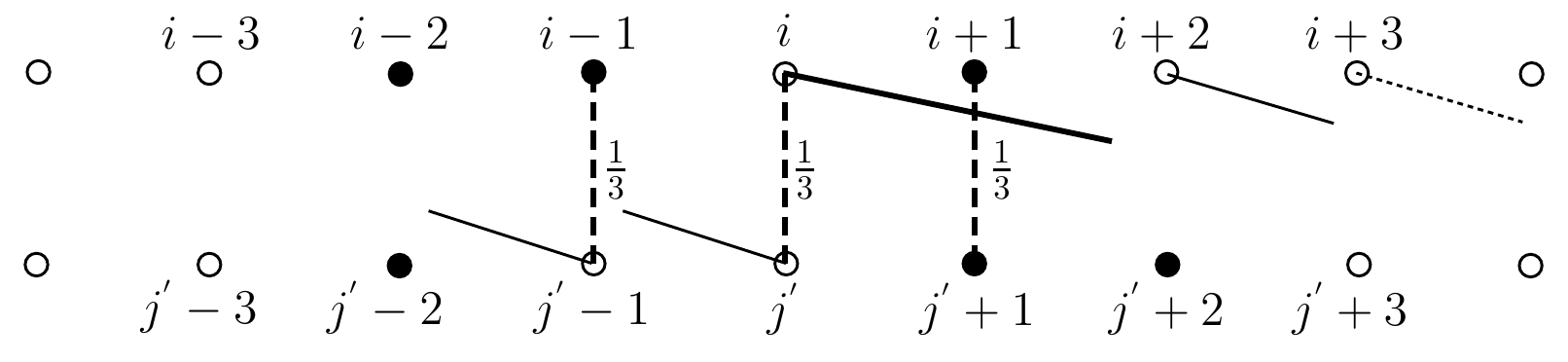} 
\caption{$|C(C^*(e_{i, \Cdot}))| = 4$ and $|\mcN_p(e_{i+2})| \ge 1$.\label{fig414a}}
\end{subfigure}
\hspace*{\fill}
\begin{subfigure}{0.46\textwidth}
\includegraphics[width=\linewidth]{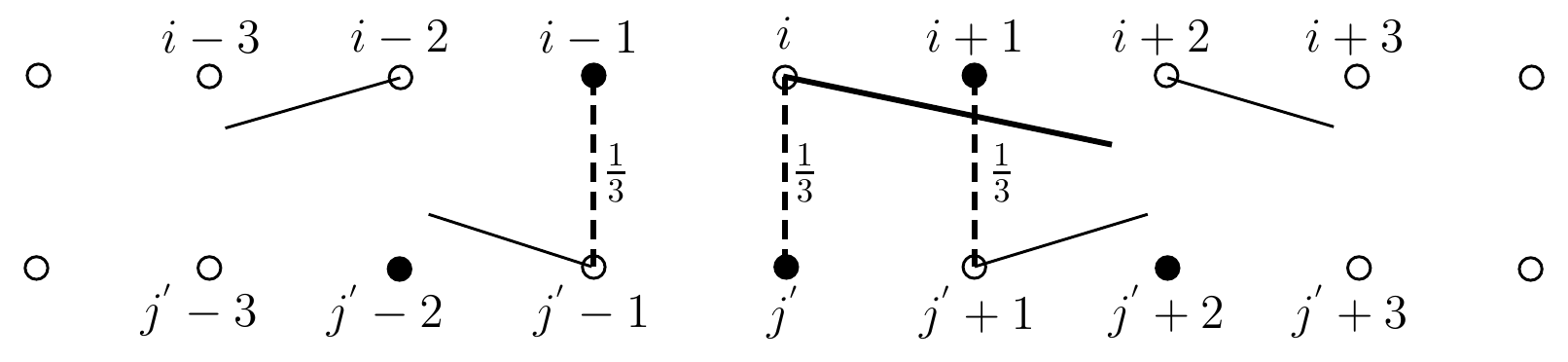} 
\caption{$|C(C^*(e_{i, \Cdot}))| = 5$.\label{fig414b}}
\end{subfigure}
\begin{subfigure}{0.48\textwidth}
\includegraphics[width=\linewidth]{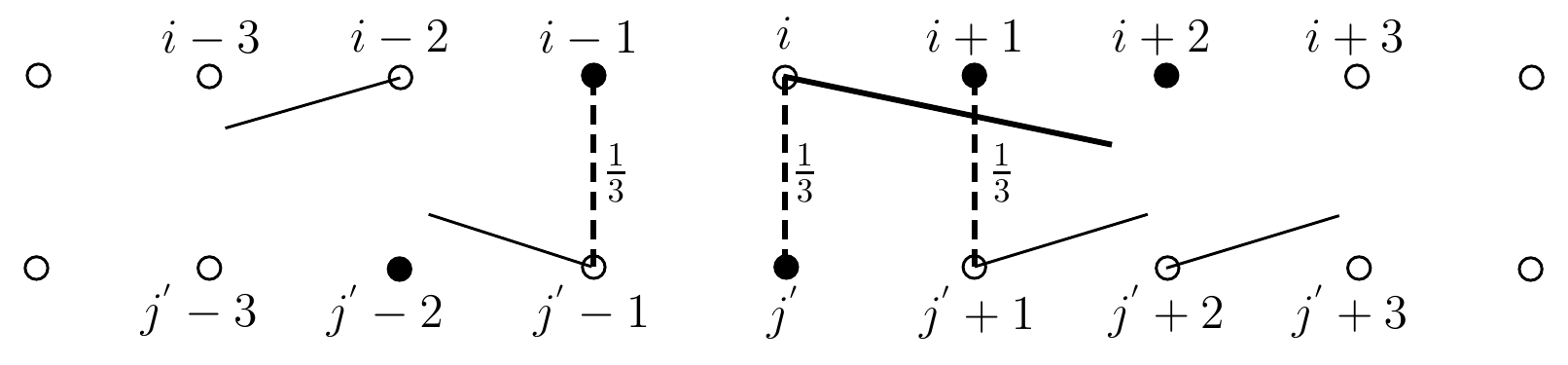} 
\caption{$|C(C^*(e_{i, \Cdot}))| = 5$.\label{fig414c}}
\end{subfigure}
\hspace*{\fill}
\begin{subfigure}{0.48\textwidth}
\includegraphics[width=\linewidth]{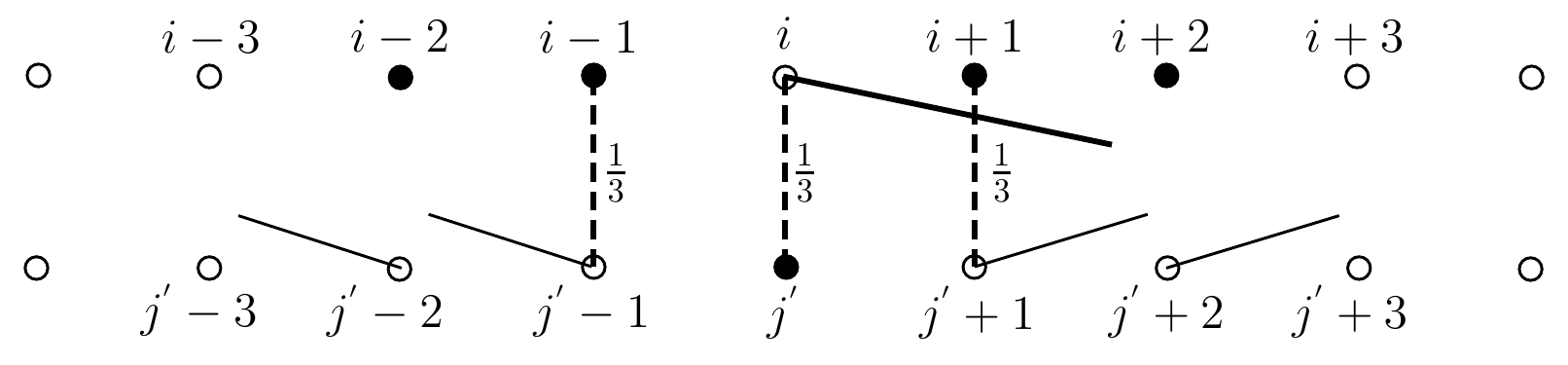} 
\caption{$|C(C^*(e_{i, \Cdot}))| = 5$.\label{fig414d}}
\end{subfigure}
\caption{The four possible configurations of $\mcN_p[C(C^*(e_{i, \Cdot}))]$ when $\tau(e_{i, j} \gets C^*(e_{i, \Cdot})) = \big(\frac 13, \frac 13, \frac 13\big)$,
associated with four possible edge combinations of $C(C^*(e_{i, \Cdot}))$ with $|C(C^*(e_{i, \Cdot}))| = 4, 5, 5, 5$, respectively.
Fig.~\ref{fig414a} also represents the case where $e_{j'+2} \in M$ instead of $e_{i+2} \in M$;
Fig.~\ref{fig414c} also represents the case where $e_{j'-2}, e_{i+2} \in M$ instead of $e_{i-2}, e_{j'+2} \in M$.\label{fig414}}
\end{figure}

\item $\tau(e_{i, j} \gets C^*(e_{i, \Cdot})) = \big(\frac 12, \frac 12, 0\big)$:
We have $2 \le |C(C^*(e_{i, \Cdot}))| \le 3$.
If $|C(C^*(e_{i, \Cdot}))| = 2$, then the algorithm $\mcL\mcS$ would replace the two edges of $C(C^*(e_{i, \Cdot}))$ by
$e^*_{i_1, j_1}$ and the two edges of $C^*(e_{i, \Cdot})$ to expand $M$, a contradiction.
Therefore, $|C(C^*(e_{i, \Cdot}))| = 3$, and furthermore $e_{j'+1} \in M$, and either $e_{i-2} \in M$ or $e_{j'-2} \in M$ but not both.
Due to symmetry we assume $e_{j'-2} \in M$.
We conclude that at most one of $e_{j'+1}$ and $e_{j'-2}$ can be a singleton edge of $M$.
The edge combination of $C(C^*(e_{i, \Cdot}))$ when $e_{j'-2}$ is not a singleton edge is shown in Fig.~\ref{fig415a},
and the edge combination of $C(C^*(e_{i, \Cdot}))$ when $e_{j'+1}$ is not a singleton edge is shown in Fig.~\ref{fig415b},
where the corresponding configuration of $\mcN_p[C(C^*(e_{i, \Cdot}))]$ is also shown, respectively.

\begin{figure}[H]
\captionsetup[subfigure]{justification=centering}
\begin{subfigure}{0.46\textwidth}
\includegraphics[width=\linewidth]{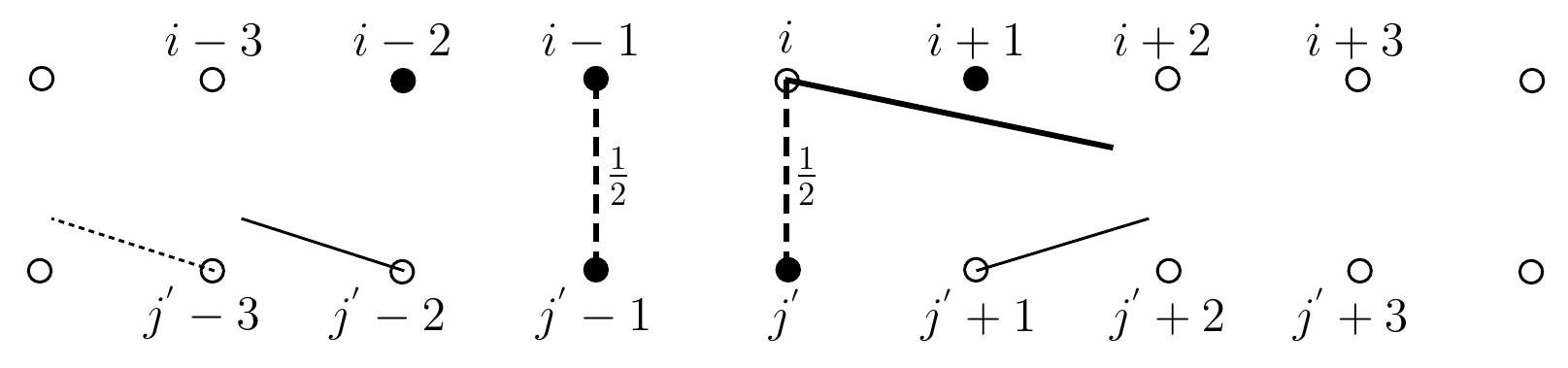} 
\caption{$|C(C^*(e_{i, \Cdot}))| = 3$ and $|\mcN_p(e_{j'-2})| \ge 1$.\label{fig415a}}
\end{subfigure}
\hspace*{\fill}
\begin{subfigure}{0.46\textwidth}
\includegraphics[width=\linewidth]{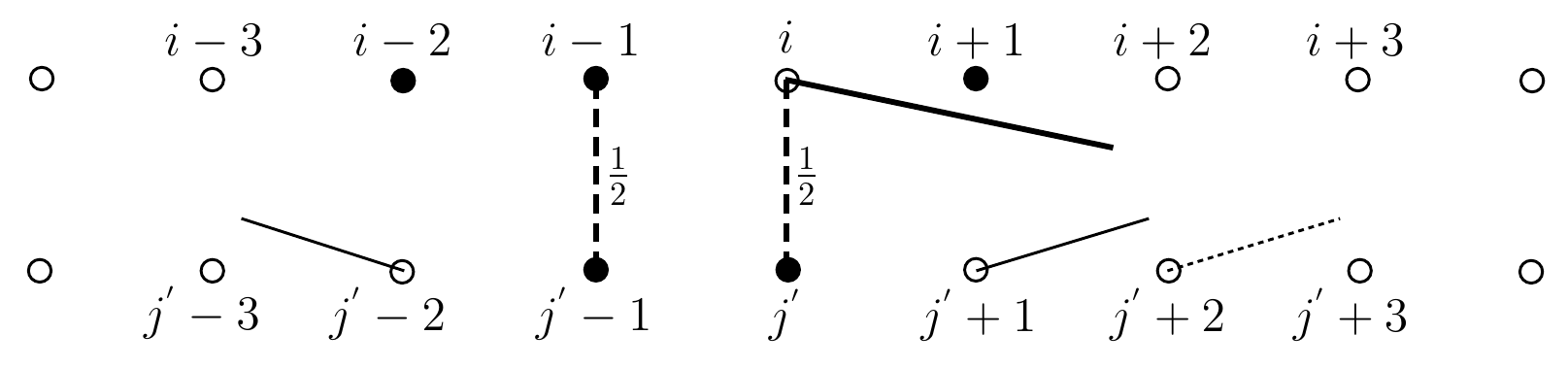} 
\caption{$|C(C^*(e_{i, \Cdot}))| = 3$ and $|\mcN_p(e_{j'+1})| \ge 1$.\label{fig415b}}
\end{subfigure}
\caption{The two possible configurations of $\mcN_p[C(C^*(e_{i, \Cdot}))]$ when $\tau(e_{i, j} \gets C^*(e_{i, \Cdot})) = \big(\frac 12, \frac 12, 0\big)$,
associated with the only possible edge combinations of $C(C^*(e_{i, \Cdot}))$ with $|C(C^*(e_{i, \Cdot}))| = 3$.
Each of them also represents the case where $e_{i-2} \in M$ instead of $e_{j'-2} \in M$.\label{fig415}}
\end{figure}

\item $\tau(e_{i, j} \gets C^*(e_{i, \Cdot})) = \big(\frac 12, 0, \frac 12\big)$:
We denote the two edges of $C^*(e_{i, \Cdot})$ as $e_{i-1, j''-1}$ and $e_{i+1, j'''+1}$, respectively;
clearly, $|(j'' - 1) - (j''' + 1)| \ge 2$.
We have $2 \le |C(C^*(e_{i, \Cdot}))| \le 3$.
The same as in the last case, we have $|C(C^*(e_{i, \Cdot}))| = 3$,
and furthermore exactly one of $i-2, j''-2, j''-1, j''$ is incident with an edge of $M$,
and exactly one of $i+2, j''', j'''+1, j'''+2$ is incident with an edge of $M$.
Among these $16$ edge combinations of $C(C^*(e_{i, \Cdot}))$, in one of them the two edges of $C(C^*(e_{i, \Cdot}))$ could be parallel to each other,
as shown in Fig.~\ref{fig416a} (this happens when $j''' = j'' + 1$, and $e_{j''}, e_{j'''} \in M$),
where the corresponding configuration of $\mcN_p[C(C^*(e_{i, \Cdot}))]$ is also shown;
one of the other $15$ is shown in Fig.~\ref{fig416b} ($j''' > j'' + 1$, $e_{j''-1}, e_{j'''} \in M$),
where the corresponding configuration of $\mcN_p[C(C^*(e_{i, \Cdot}))]$ is also shown.

\begin{figure}[H]
\captionsetup[subfigure]{justification=centering}
\begin{subfigure}{0.46\textwidth}
\includegraphics[width=\linewidth]{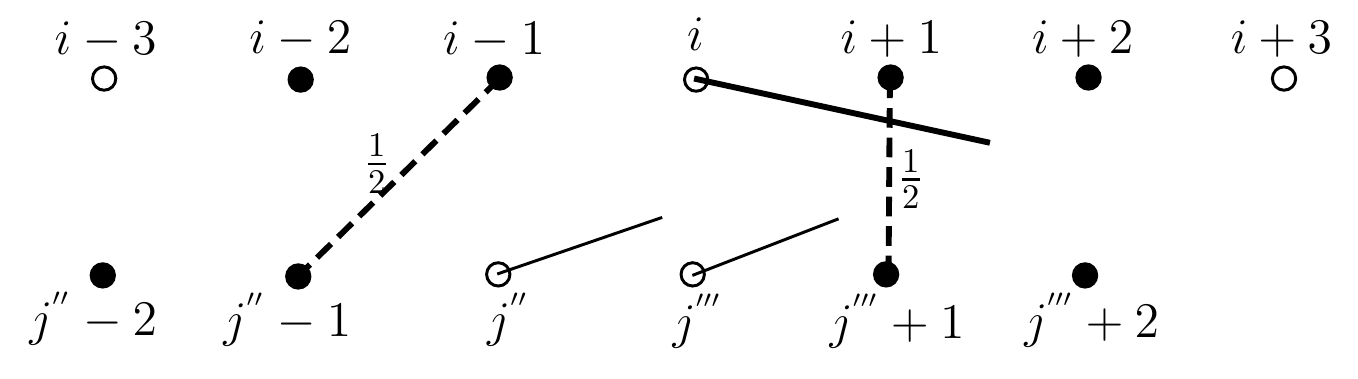} 
\caption{$|C(C^*(e_{i, \Cdot}))| = 3$ and $j''' = j'' + 1$.\label{fig416a}}
\end{subfigure}
\hspace*{\fill}
\begin{subfigure}{0.46\textwidth}
\includegraphics[width=\linewidth]{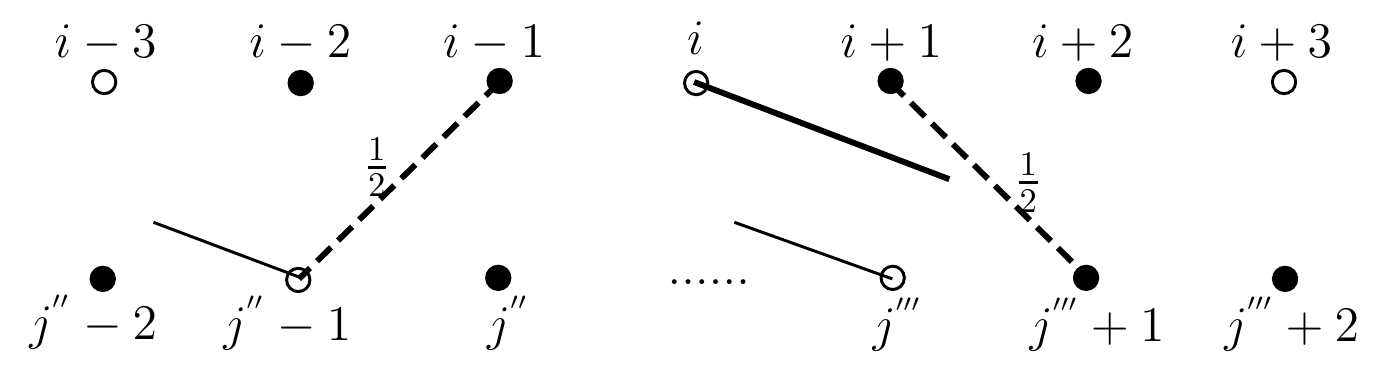} 
\caption{$|C(C^*(e_{i, \Cdot}))| = 3$ and $j''' > j'' + 1$.\label{fig416b}}
\end{subfigure}
\caption{The two possible configurations of $\mcN_p[C(C^*(e_{i, \Cdot}))]$ when $\tau(e_{i, j} \gets C^*(e_{i, \Cdot})) = \big(\frac 12, 0, \frac 12\big)$,
associated with the two possible edge combinations of $C(C^*(e_{i, \Cdot}))$ with $|C(C^*(e_{i, \Cdot}))| = 3$.
The second configuration represents the other $15$ symmetric cases exactly one of $i-2, j''-2, j''-1, j''$ is incident with an edge of $M$
and exactly one of $i+2, j''', j'''+1, j'''+2$ is incident with an edge of $M$, but the two edges are not parallel to each other.\label{fig416}}
\end{figure}
\end{enumerate}

Therefore, we have a total of $27$ configurations of $\mcN_p[C(C^*(e_{i, \Cdot}))]$ associated with all the possible edge combinations of $C(C^*(e_{i, \Cdot}))$,
up to symmetry, for further discussion.

\begin{lemma}
\label{lemma411}
When $e_{i, j}$ is a singleton edge of $M$ with $\omega(e_{i, j}) \ge 3$,
there is at least one parallel edge of $M$ in $C(C^*(e_{i, j}))$ for each of the $8$ possible value combinations of $\tau(e_{i, j} \gets C^*(e_{i, j}))$.
\end{lemma}
\begin{proof}
From Lemma~\ref{lemma404}, for each of the $8$ possible value combinations of $\tau(e_{i, j} \gets C^*(e_{i, j}))$,
there is an entry $1$ in the ordered value combination of $\tau(e_{i, j} \gets C^*(e_{i, \Cdot}))$ or $\tau(e_{i, j} \gets C^*(e_{\Cdot, j}))$.
There are only $5$ possible such ordered value combinations, which are
$\big(\frac 12, 1, \frac 12\big)$, 
$\big(\frac 12, 1, \frac 13\big)$, $\big(1, \frac 12, \frac 13\big)$, 
$\big(1, \frac 12, \frac 14\big)$, and
$\big(\frac 13, 1, \frac 13\big)$.
The above Figs.~\ref{fig401}--\ref{fig404} show that for the first $4$ ordered value combinations, there is at least one parallel edge of $M$ in $C(C^*(e_{i, j}))$.

If $\big(\frac 13, 1, \frac 13\big)$ is the ordered value combination for $\tau(e_{i, j} \gets C^*(e_{i, \Cdot}))$,
then $\tau(e_{i, j} \gets C^*(e_{\Cdot, j}))$ has an ordered value combination
either $\big(\frac 12, \frac 12, \frac 13\big)$ or $\big(\frac 12, \frac 13, \frac 12\big)$.
The above Figs.~\ref{fig406} and \ref{fig409} show that there is at least one parallel edge of $M$ in $C(C^*(e_{\Cdot, j}))$.
This proves the lemma.
\end{proof}

\subsection{An upper bound on $\omega(e)$ for $e \in C(C^*(e_{i, j})) - \{e_{i, j}\}$}
\label{sec4.4}
From Lemma~\ref{lemma410}, in the sequel we always consider the case $e_{i,j}$ is a singleton edge of $M$ with $\omega(e_{i,j}) \ge 3$.

We walk through all the $27$ configurations of $\mcN_p[C(C^*(e_{i, \Cdot}))]$ to
determine an upper bound on $\omega(e)$, for any $e \in C(C^*(e_{i, \Cdot})) - \{e_{i, j}\}$.

\begin{lemma}
\label{lemma412}
For any edge $e \in C(C^*(e_{i, \Cdot})) - \{e_{i, j}\}$, $|C(e^*_{h, \ell})| \ge 2$ for all edges $e^*_{h, \ell} \in C^*(e)$,
if any one of the following five conditions holds:
\begin{enumerate}
\item
	$|C(C^*(e_{i, \Cdot}))| = |C^*(e_{i, \Cdot})| = 3$.
\item
	$|C(C^*(e_{i, \Cdot}))| = |C^*(e_{i, \Cdot})| + 1 = 4$ and
	there is an edge $e^*_{i_1, j_1} \in C^*(e_{\Cdot, j})$ such that $|C(e^*_{i_1, j_1})| = 1$.
\item
	$e \in C(e^*_{i_2, j_2})$ for some $e^*_{i_2, j_2} \in C^*(e_{i, \Cdot})$ with $|C(e^*_{i_2, j_2})| = 2$,
	and there is an edge $e^*_{i_1, j_1} \in C^*(e_{\Cdot, j})$ such that $|C(e^*_{i_1, j_1})| = 1$.
\item
	$e \in C(e^*_{i_1, j_1}) \cup C(e^*_{i_2, j_2})$ for some $e^*_{i_1, j_1}, e^*_{i_2, j_2} \in C^*(e_{i, \Cdot})$
	with $|C(e^*_{i_1, j_1}) \cup C(e^*_{i_2, j_2})| = 2$.
\item
	$e \in C(e^*_{i_2, j_2}) \cup C(e^*_{i_3, j_3})$ for some $e^*_{i_2, j_2}, e^*_{i_3, j_3} \in C^*(e_{i, \Cdot})$
	with $|C(e^*_{i_2, j_2}) \cup C(e^*_{i_3, j_3})| = 3$,
	and there is an edge $e^*_{i_1, j_1} \in C^*(e_{\Cdot, j})$ such that $|C(e^*_{i_1, j_1})| = 1$.
\end{enumerate}
And consequently, $\omega(e) \le \frac {17}6$.
\end{lemma}
\begin{proof}
We prove by contradiction, and thus assume that there is an edge $e^*_{h, \ell} \in C^*(e)$ such that $|C(e^*_{h, \ell})| = 1$.

If the first condition holds,
then the algorithm $\mcL\mcS$ would replace the three edges of $C(C^*(e_{i, \Cdot}))$ by
$e^*_{h, \ell}$ and the three parallel edges of $C^*(e_{i, \Cdot})$ to expand $M$, a contradiction.

If the second condition holds,
then the algorithm $\mcL\mcS$ would replace the four edges of $C(C^*(e_{i, \Cdot}))$ by
$e^*_{i_1, j_1}$, $e^*_{h, \ell}$, and the three parallel edges in $C^*(e_{i, \Cdot})$ to expand $M$, a contradiction.

If the third condition holds,
then the algorithm $\mcL\mcS$ would replace the two edges of $C(e^*_{i_2, j_2})$ by $e^*_{i_2, j_2}$, $e^*_{i_1, j_1}$, $e^*_{h, \ell}$ to expand $M$,
a contradiction.

If the fourth condition holds,
then the algorithm $\mcL\mcS$ would replace the two edges of $C(e^*_{i_1, j_1}) \cup C(e^*_{i_2, j_2})$ by $e^*_{i_1, j_1}, e^*_{i_2, j_2}$, $e^*_{h, \ell}$
to expand $M$, a contradiction.

If the fifth condition holds,
then the algorithm $\mcL\mcS$ would replace the three edges of $C(e^*_{i_2, j_2}) \cup C(e^*_{i_3, j_3})$ by $e^*_{i_2, j_2}, e^*_{i_3, j_3}$, $e^*_{i_1, j_1}$,
$e^*_{h, \ell}$ to expand $M$, a contradiction.

Therefore, we proved that $|C(e^*_{h, \ell})| \ge 2$ for all edges $e^*_{h, \ell} \in C^*(e)$.
It then follows from Lemma~\ref{lemma403} that $\omega(e) \le 5 \times \frac 12 + \frac 13 = \frac {17}6$.
\end{proof}

\begin{lemma}
\label{lemma413}
For each edge $e \in C(C^*(e_{i, \Cdot})) - \{e_{i, j}\}$ in
Figs.~\ref{fig401}, \ref{fig403}, \ref{fig406a}, \ref{fig406b}, \ref{fig408a}, \ref{fig409}, \ref{fig411a}, \ref{fig411b},
\ref{fig414a}, \ref{fig415a}, \ref{fig415b}, \ref{fig416a} and \ref{fig416b},
$e_{j'-2}$ in Fig.~\ref{fig407a},
$e_{j'-1}$ in Fig.~\ref{fig408b},
$e_{j'}$ in Fig.~\ref{fig410a},
$e_{j'-1}$ in Fig.~\ref{fig411c}, and
$e_{j'}$ in Fig.~\ref{fig413},
its total amount of tokens is $\omega(e) \le \frac {17}6$.
\end{lemma}
\begin{proof}
At least one of the five conditions in Lemma~\ref{lemma412} applies to each of these edges.
For example, in Fig.~\ref{fig401},
for the edge $e_{j'-2}$, the fourth condition of Lemma~\ref{lemma412} holds by setting $(i_1, j_1) := (i-1, j'-1)$ and $(i_2, j_2) := (i, j')$;
for the edge $e_{i+2}$, the fourth condition of Lemma~\ref{lemma412} holds by setting $(i_1, j_1) := (i, j')$ and $(i_2, j_2) := (i+1, j'+1)$.
\end{proof}

\begin{lemma}
\label{lemma414}
For both the edges $e_{i+q}, e_{j'+q} \in C(C^*(e_{i, \Cdot}))$ shown in
Figs.~\ref{fig402}, \ref{fig404}, \ref{fig405}, \ref{fig407a}, \ref{fig408b}, \ref{fig408c}, \ref{fig410b} and \ref{fig411c},
for some $q = 2$ or $-2$,
the total amount of tokens for each of them is at most $\frac {35}{12}$.
\end{lemma}
\begin{proof}
Consider the edge $e_{i+q}$.
If it is a parallel edge of $M$, then it simply cannot fit into any of the $27$ configurations shown in Figs. \ref{fig401}--\ref{fig416},
in which the edge $e_{i,j}$ is a singleton edge of $M$.
(By ``fitting into'' it means the edge $e_{i+q}$ takes up the role of the edge $e_{i,j}$ in the configuration.)
If $e_{i+q}$ is a singleton edge of $M$, we show next that due to the existence of the paired edge $e_{j'+q} \in M$, 
$e_{i+q}$ cannot fit into any of the $27$ configurations shown in Figs. \ref{fig401}--\ref{fig416} either.
This is done by using the edge combinations of $C(C^*(i, \Cdot))$ and the existence of certain edges in $\mcN_p(C(C^*(e_{i, \Cdot}))$.

In more details, we first see that $e_{i+q}$ can only possibly fit into $7$ of the $27$ configurations shown in
Figs.~\ref{fig408a}, \ref{fig410a}, \ref{fig411a}, \ref{fig411b}, \ref{fig412}, \ref{fig413} and \ref{fig414a}, due to the existence of the edge $e_{j'+q} \in M$.
Next, if it were fit in any of them, then in the fitted configuration there is an edge $e_{i-2} \in M$ but
none of the five edges $e_{i-3}, e_{j'-3}, e_{j'-2}, e_{i-1}, e_{j'-1}$ can be in $M$.
This last requirement rules out Fig.~\ref{fig408a} due to $e_{j'-3}, e_{i+3} \in \mcN_p(C(C^*(e_{i, \Cdot})))$;
it rules out Fig.~\ref{fig410a} due to $e_{j'-3} \in \mcN_p(C(C^*(e_{i, \Cdot})))$;
it rules out Figs.~\ref{fig411a}, \ref{fig411b} and \ref{fig412} due to $e_{j'+1} \in C(C^*(e_{i, \Cdot}))$ but
none of $e_{i-2}, e_{j'-2}$ is in $C(C^*(e_{i, \Cdot}))$;
it rules out Fig.~\ref{fig413} due to $e_{i-2}, e_{j'-2}, e_{i+2}, e_{j'+2} \in C(C^*(e_{i, \Cdot}))$;
and it rules out Fig.~\ref{fig414a} due to $e_{j'-1} \in C(C^*(e_{i, \Cdot}))$ and $e_{i+3} \in \mcN_p(C(C^*(e_{i, \Cdot})))$.

Therefore, $\omega(e_{i+q}) < 3$.

Using at most six values from $\{1, \frac 12, \frac 13, \frac 14, \frac 15, \frac 16\}$,
the sum closest but less than $3$ is $1 + \frac 12 + \frac 12 + \frac 12 + \frac 14 + \frac 15 = \frac {59}{20}$.
In order for $\tau(e_{i, j} \gets C^*(e_{i, j}))$ to have a value combination $\big\{1, \frac 12, \frac 12, \frac 12, \frac 14, \frac 15\big\}$,
Lemma \ref{lemma402} says that the value combinations for $\tau(e_{i, j} \gets C^*(e_{i, \Cdot}))$ and $\tau(e_{i, j} \gets C^*(e_{\Cdot, j}))$ are
$\big\{1, \frac 12, \frac 12\big\}$ and $\big\{\frac 12, \frac 14, \frac 15\big\}$.
Furthermore, Lemmas \ref{lemma402} and \ref{lemma406} together state that
the subsequent ordered value combinations are $\big(\frac 12, 1, \frac 12\big)$ and $\big(\frac 12, \frac 14, \frac 15\big)$.
However, $\big(\frac 12, 1, \frac 12\big)$ requires $e_{i, j}$ to be a singleton edge of $M$,
while $\big(\frac 12, \frac 14, \frac 15\big)$ implies $e_{i, j}$ is a parallel edge of $M$,
a contradiction.

The second closest sum to $3$ is $\frac {35}{12}$,
that is the sum of the value combinations $\big\{1, \frac 12, \frac 12, \frac 12, \frac 14, \frac 16\big\}$
(which can be ruled out similarly as in the last paragraph) and $\big\{1, \frac 12, \frac 12, \frac 13, \frac 13, \frac 14\big\}$.
Therefore, $\omega(e_{i+q}) \le \frac{35}{12}$.
\end{proof}

\begin{lemma}
\label{lemma415}
For each edge $e \in C(C^*(e_{i, j})) - \{e_{i, j}\}$ with $\omega(e_{i, j}) \ge 3$,
we have $\omega(e) \le \frac{35}{12}$, except for the following two cases where we have $\omega(e_{i, j}) = 3$:
\begin{enumerate}
\item
	in the configuration shown in Fig.~\ref{fig414b},
	it is possible to have either $\omega(e_{j'-1}) = 3$ (when $|\mcN_p(e_{i-2})| \ge 1$) or $\omega(e_{j'+1}) = 3$ (when $|\mcN_p(e_{i+2})| \ge 1$), but not both;
\item
	in the configuration shown in Fig.~\ref{fig414c},
	it is possible to have either $\omega(e_{j'-1}) = 3$ or $\omega(e_{i-2}) = 3$ (when $|\mcN_p(e_{j'-3})| \ge 1$), but not both.
\end{enumerate}
\end{lemma}
\begin{proof}
Recall that Lemma~\ref{lemma413} settled all the edges of $C(C^*(e_{i, j})) - \{e_{i, j}\}$ in
Figs.~\ref{fig401}, \ref{fig403}, \ref{fig406a}, \ref{fig406b}, \ref{fig408a}, \ref{fig409}, \ref{fig411a}, \ref{fig411b},
\ref{fig414a}, \ref{fig415a}, \ref{fig415b}, \ref{fig416a} and \ref{fig416b},
$e_{j'-2}$ in Fig.~\ref{fig407a},
$e_{j'-1}$ in Fig.~\ref{fig408b},
$e_{j'}$ in Fig.~\ref{fig410a},
$e_{j'-1}$ in Fig.~\ref{fig411c}, and
$e_{j'}$ in Fig.~\ref{fig413};
Lemma~\ref{lemma414} settled all the paired edges $e_{i+q}, e_{j'+q} \in C(C^*(e_{i, \Cdot}))$ in
Figs.~\ref{fig402}, \ref{fig404}, \ref{fig405}, \ref{fig407a}, \ref{fig408b}, \ref{fig408c}, \ref{fig410b} and \ref{fig411c},
for some $q = 2$ or $-2$, and all the edges known to be parallel, including
$e_{j'-2}$ in Fig.~\ref{fig402},
$e_{j'+1}$ in Fig.~\ref{fig404},
$e_{j'+1}$ in Fig.~\ref{fig407a},
$e_{i-2}$ in Fig.~\ref{fig408b},
$e_{j'-1}$ in Fig.~\ref{fig408c},
$e_{j'-2}$ in Fig.~\ref{fig410a},
$e_{j'+1}$ in Fig.~\ref{fig410b},
$e_{j'+1}$ in Fig.~\ref{fig411c},
$e_{j'}, e_{j'+1}, e_{j'+2}$ in Fig.~\ref{fig412},
$e_{j'+1}, e_{j'+2}$ in Fig.~\ref{fig414c}, and
$e_{j'-2}, e_{j'-1}, e_{j'+1}, e_{j'+2}$ in Fig.~\ref{fig414d}.

We therefore are left to prove the lemma for the edges not known to be parallel in Figs.~\ref{fig410a}, \ref{fig412}, \ref{fig413}, \ref{fig414b} and \ref{fig414c}.
We deal with them separately in the following.

\begin{enumerate}
\item
	The edges $e_{i+2}, e_{j'+2}$ in Fig.~\ref{fig410a} and the edges $e_{i-2}, e_{j'-2}, e_{i+2}, e_{j'+2}$ in Fig.~\ref{fig413}, which can be settled the same.

Consider the edge $e_{i+2}$, which can potentially fit into the configuration in Fig.~\ref{fig410a} or Fig.~\ref{fig413}.
In either case, there is an edge $e^*_{i_1, j_1} \in C^*(e_{\Cdot,j})$ such that $|C(e^*_{i_1, j_1})| = 1$ and
there is an edge $e^*_{h_1, \ell_1} \in C^*(e_{i+2})$ such that $|C(e^*_{h_1, \ell_1})| = 1$.
Then the algorithm $\mcL\mcS$ would replace the four edges $e_{i, j}$, $e_{j'}$, $e_{i+2}$, $e_{j'+2}$ by
the five edges $e^*_{i, j'}$, $e^*_{i+1, j'+1}$, $e^*_{i+2, j'+2}$, $e^*_{i_1, j_1}$, $e^*_{h_1, \ell_1}$ to expand $M$, a contradiction.
In summary, $e_{i+2}$ cannot fit into any of the $27$ configurations shown in Figs. \ref{fig401}--\ref{fig416} and thus $\omega(e_{i+2}) \le \frac {35}{12}$.

\item
	The edge $e_{i+2}$ in Fig.~\ref{fig412}.

If $e_{i+2}$ is to fit in, then it can fit only into the configuration in Fig.~\ref{fig412}.
This suggests that $C(C^*(e_{i+2, \Cdot})) = C(C^*(e_{i, \Cdot}))$.
Since there is an edge $e^*_{i_1, j_1} \in C^*(e_{\Cdot, j})$ such that $|C(e^*_{i_1, j_1})| = 1$ and
there is an edge $e^*_{h_1, \ell_1} \in C^*(e_{i+2})$ such that $|C(e^*_{h_1, \ell_1})| = 1$,
the algorithm $\mcL\mcS$ would replace the five edges of $C(C^*(e_{i, \Cdot}))$ by
any six edges from $\big\{e^*_{i-1, j'-1}, e^*_{i, j'}, e^*_{i+1, j'+1}, e^*_{i+2, j'+2}, e^*_{i+3, j'+3}, e^*_{i_1, j_1}, e^*_{h_1, \ell_1}\big\}$
to expand $M$, a contradiction.
In summary, $e_{i+2}$ cannot fit into any of the $27$ configurations shown in Figs. \ref{fig401}--\ref{fig416} and thus $\omega(e_{i+2}) \le \frac {35}{12}$.

\item
	The edges $e_{i-2}$ and $e_{i+2}$ in Fig.~\ref{fig414b}, which can be settled the same.

Consider the edge $e_{i-2}$, which can potentially fit into the configuration in Fig.~\ref{fig414b} or Fig.~\ref{fig414c}.
In either case, all the four edges $e_{i-2}$, $e_{j'-1}$, $e_{i, j}$, $e_{j'+1}$ are singleton edges of $M$,
and there is an edge $e^*_{i_1, j_1} \in C^*(e_{\Cdot, j})$ such that $|C(e^*_{i_1, j_1})| = 1$ and
there is an edge $e^*_{h_1, \ell_1} \in C^*(e_{i-2})$ such that $|C(e^*_{h_1, \ell_1})| = 1$.
Then the algorithm $\mcL\mcS$ would replace these four singleton edges of $M$ by
the edges $e^*_{i_1, j_1}, e^*_{h_1, \ell_1}$ and the two parallel edges $e^*_{i-1, j'-1}, e^*_{i, j'}$ to reduce the singleton edges of $M$, a contradiction.
In summary, $e_{i-2}$ cannot fit into any of the $27$ configurations shown in Figs. \ref{fig401}--\ref{fig416} and thus $\omega(e_{i-2}) \le \frac {35}{12}$.

\item
	The edges $e_{j'-1}$ and $e_{j'+1}$ in Fig.~\ref{fig414b}, which can be settled the same.

Consider the edge $e_{j'-1}$, which can potentially fit into the configuration in Fig.~\ref{fig414b} or Fig.~\ref{fig414c}.
If $e_{j'-1}$ fits into the configuration in Fig.~\ref{fig414b},
then the same as in the last case the algorithm $\mcL\mcS$ would be able to reduce the singleton edges of $M$, a contradiction.
If $e_{j'-1}$ fits into the configuration in Fig.~\ref{fig414c}, then the edge $e_{i-2}$ is a parallel edge of $M$.
From $\tau(e_{i, j} \gets C^*(e_{i, \Cdot})) = \big(\frac 13, \frac 13, \frac 13\big)$,
we conclude that $\omega(e_{i, j}) \le 3$, and consequently $\omega(e_{i, j}) = 3$ and $\omega(e_{j'-1}) = 3$.

It is easy to see that we cannot have both $\omega(e_{j'-1}) = \omega(e_{j'+1}) = 3$,
since otherwise the algorithm $\mcL\mcS$ would be able to expand $M$ by swapping out the five edges of $C(C^*(e_{i, \Cdot}))$, a contradiction.

In summary, we have either $\omega(e_{j'-1}) \le \frac {35}{12}$ or $\omega(e_{j'-1}) = 3$,
the latter of which implies $|\mcN_p(e_{i-2})| \ge 1$ and it is the first case stated in the lemma.

\item
	The edge $e_{i-2}$ in Fig.~\ref{fig414c}.

If $e_{i-2}$ is to fit in, then it can fit only into the configuration in Fig.~\ref{fig414b} or Fig.~\ref{fig414c}.
If $e_{i-2}$ fits into the configuration in Fig.~\ref{fig414b},
then all the four edges $e_{i, j}$, $e_{j'-1}$, $e_{i-2}$, $e_{j'-3}$ are singleton edges of $M$.
Since there is an edge $e^*_{i_1, j_1} \in C^*(e_{\Cdot, j})$ such that $|C(e^*_{i_1, j_1})| = 1$ and
there is an edge $e^*_{h_1, \ell_1} \in C^*(e_{i-2})$ such that $|C(e^*_{h_1, \ell_1})| = 1$,
the algorithm $\mcL\mcS$ would replace these four singleton edges of $M$ by
the edges $e^*_{i_1, j_1}, e^*_{h_1, \ell_1}$ and the two parallel edges $e^*_{i-1, j'-1}, e^*_{i-2, j'-2}$ to reduce the singleton edges, a contradiction.
If $e_{i-2}$ fits into the configuration in Fig.~\ref{fig414c}, then the edge $e_{j'-3}$ is a parallel edge of $M$.
From $\tau(e_{i, j} \gets C^*(e_{i, \Cdot})) = \big(\frac 13, \frac 13, \frac 13\big)$,
we conclude that $\omega(e_{i, j}) \le 3$, and consequently $\omega(e_{i, j}) = 3$ and $\omega(e_{i-2}) = 3$.
In summary, we have either $\omega(e_{i-2}) \le \frac {35}{12}$ or $\omega(e_{i-2}) = 3$, the latter of which implies $|\mcN_p(e_{j'-3})| \ge 1$.

\item
	The edge $e_{j'-1}$ in Fig.~\ref{fig414c}.

If $e_{j'-1}$ is to fit in, then it can fit only into the configuration in Fig.~\ref{fig414b} or Fig.~\ref{fig414c}.
If $e_{j'-1}$ fits into the configuration in Fig.~\ref{fig414b}, then the edge $e_{i-2}$ is a singleton edge of $M$.
If $e_{j'-1}$ fits into the configuration in Fig.~\ref{fig414c}, then the edge $e_{i-2}$ is a parallel edge of $M$.
From $\tau(e_{i, j} \gets C^*(e_{i, \Cdot})) = \big(\frac 13, \frac 13, \frac 13\big)$,
we conclude that $\omega(e_{i, j}) \le 3$, and consequently $\omega(e_{i, j}) = 3$ and $\omega(e_{j'-1}) = 3$.
In summary, we have either $\omega(e_{j'-1}) \le \frac {35}{12}$ or $\omega(e_{j'-1}) = 3$.
\end{enumerate}
It is also easy to see that we cannot have both $\omega(e_{j'-1}) = \omega(e_{i-2}) = 3$ in the last two items,
since otherwise the algorithm $\mcL\mcS$ would be able to expand $M$ by swapping out the five edges of $C(C^*(e_{i, \Cdot}))$, a contradiction.
This is the second case stated in the lemma.
We have proved the lemma.
\end{proof}

\subsection{An upper bound on $\omega(e)$ for $e \in C(C^*(e_{i, j}))$ known to be parallel}
\label{sec4.5}
In this section, we provide a better upper bound on the total amount of tokens received by an edge of $C(C^*(e_{i, j}))$
that is {\em known} to be parallel, for example, in Fig.~\ref{fig402} the edge $e_{j'-2}$ is known parallel but the edge $e_{i+2}$ is not.
Also, from Lemma~\ref{lemma415}, in Fig.~\ref{fig414b} it is possible to have $\omega(e_{j'-1}) = 3$ when $|\mcN_p(e_{i-2})| \ge 1$;
we therefore consider the edge $e_{i-2}$ to be parallel too.
For the same reason, we consider the edge $e_{i+2}$ in Fig.~\ref{fig414b} to be parallel.

\begin{lemma}
\label{lemma416}
For each parallel edge $e \in C(C^*(e_{i, \Cdot}))$ with $\omega(e_{i,j}) \ge 3$,
we have $|C(e^*_{h,\ell})| \ge 2$ for all $e^*_{h,\ell} \in C^*(e)$, except for the following edges:
\begin{enumerate}
\item
	the edge $e_{j'+2}$ in Figs.~\ref{fig404}, \ref{fig407a}, \ref{fig410b}, \ref{fig411c} and \ref{fig412},
\item
	the edges $e_{i-2}, e_{i+2}$ in Fig.~\ref{fig414b},
\item
	the edges $e_{j'+1}, e_{j'+2}$ in Fig.~\ref{fig414c}, and
\item
	the edges $e_{j'-1}, e_{j'-2}, e_{j'+1}, e_{j'+2}$ in Fig.~\ref{fig414d}.
\end{enumerate}
\end{lemma}
\begin{proof}
At least one of the five conditions in Lemma~\ref{lemma412} applies to each of these edges.
For example, in Fig.~\ref{fig402},
for the edge $e_{j'-2}$, the fourth condition of Lemma~\ref{lemma412} holds by setting $(i_1, j_1) := (i-1, j'-1)$ and $(i_2, j_2) := (i, j')$.
\end{proof}

Among all the $27$ configurations in Figs.~\ref{fig401}--\ref{fig416}, we have the following two observations.

\begin{obs}
\label{obs4.1}
If the edge $e_{i-1}$ (or the edge $e_{j'-1}$) is a known parallel edge of $M$ in $C(C^*(e_{i, \Cdot}))$,
and $e^*_{i-1, j'-1} \in M^*$, then $|C(e^*_{i-1, j'-1})| \ge 3$;
if the edge $e_{j'}$ is a known parallel edge of $M$ in $C(C^*(e_{i, \Cdot}))$,
and $e^*_{i, j'} \in M^*$, then $|C(e^*_{i, j'})| \ge 3$;
if the edge $e_{i+1}$ (or the edge $e_{j'+1}$) is a known parallel edge of $M$ in $C(C^*(e_{i, \Cdot}))$,
and $e^*_{i+1, j'+1} \in M^*$, then $|C(e^*_{i+1, j'+1})| \ge 3$.
\end{obs}

\begin{obs}
\label{obs4.2}
When $e^*_{i,j'} \in C^*(e_{i, \Cdot})$,
if the edge $e_{i+p}$ (or $e_{j'+p}$, respectively) is a known parallel edge of $M$ in $C(C^*(e_{i, \Cdot}))$ for some $p = -2, 2$,
then $|C(e^*_{i+p})| \ge 2$ (or $|C(e^*_{j'+p})| \ge 2$, respectively).
\end{obs}

Based on Lemmas \ref{lemma410}, \ref{lemma412}, and Observations \ref{obs4.1} and \ref{obs4.2}, we can prove the following two lemmas.

\begin{lemma}
\label{lemma417}
For any pair of parallel edges $e_{h, \ell}, e_{h+1, \ell+1} \in C(C^*(e_{i, \Cdot}))$,
and an edge $e^* \in C^*(e_{h, \ell}) \cap C^*(e_{h+1, \ell+1})$,
we have $|C(e^*)| \ge 3$ if one of the following three conditions holds.
\begin{enumerate}
\item
	$|C(C^*(e_{i, \Cdot}))| = |C^*(e_{i, \Cdot})| = 3$.
\item
	$|C(C^*(e_{i, \Cdot}))| = |C^*(e_{i, \Cdot})| + 1 = 4$ and there is an edge $e^*_{i_1, j_1} \in C^*(e_{\Cdot, j})$ such that $|C(e^*_{i_1, j_1})| = 1$.
\item
	$e_{h, \ell}, e_{h+1, \ell+1} \in C(e^*_{i_2, j_2}) \cup C(e^*_{i_3, j_3})$
	for some $e^*_{i_2, j_2}, e^*_{i_3, j_3} \in C^*(e_{i, \Cdot})$ with $|C(e^*_{i_2, j_2}) \cup C(e^*_{i_3, j_3})| = 3$,
	and there is an edge $e^*_{i_1, j_1} \in C^*(e_{\Cdot, j})$ such that $|C(e^*_{i_1, j_1})| = 1$.
\end{enumerate}
\end{lemma}
\begin{proof}
We prove by contradiction, and thus assume that there is an edge $e^* \in C^*(e_{h, \ell}) \cap C^*(e_{h+1, \ell+1})$ such that
$C(e^*) = \{e_{h, \ell}, e_{h+1, \ell+1}\}$.

If the first condition holds,
then it follows from Observation \ref{obs4.1} that $e^* \notin C^*(e_{i, \Cdot})$.
In this case, the algorithm $\mcL\mcS$ would replace the three edges of $C(C^*(e_{i, \Cdot}))$ by $e^*$ and the three edges of $C^*(e_{i, \Cdot})$ to expand $M$,
a contradiction.

If the second condition holds,
then again it follows from Observation \ref{obs4.1} that $e^* \notin C^*(e_{i, \Cdot})$.
Also, the edge $e^*_{i_1, j_1}$ is distinct from $e^*$.
In this case, the algorithm $\mcL\mcS$ would replace the four edges in $C(C^*(e_{i, \Cdot}))$ by
$e^*_{i_1, j_1}, e^*$, and the three edges of $C^*(e_{i, \Cdot})$ to expand $M$, a contradiction.

If the third condition holds,
then by Observation \ref{obs4.1} the edge $e^*$ is distinct from $e^*_{i_2, j_2}, e^*_{i_3, j_3}$,
and the edge $e^*_{i_1, j_1}$ is distinct from $e^*$.
In this case, the algorithm $\mcL\mcS$ would replace the three edges of $C(e^*_{i_2, j_2}) \cup C(e^*_{i_3, j_3})$ by
$e^*_{i_2, j_2}, e^*_{i_3, j_3}$, $e^*_{i_1, j_1}$, $e^*$ to expand $M$, a contradiction.
\end{proof}

\begin{lemma}
\label{lemma418}
For any pair of parallel edges $e_{h, \ell}, e_{h+1, \ell+1}$ where $e_{h, \ell} \in C(C^*(e_{i, \Cdot}))$,
there is at most one edge $e^* \in C^*(e_{h, \ell}) \cap C^*(e_{h+1, \ell+1})$ such that $|C(e^*)| = 2$,
if one of the following six conditions holds.
\begin{enumerate}
\item
	$e_{h+1, \ell+1} \notin C(C^*(e_{i, \Cdot}))$ and $|C(C^*(e_{i, \Cdot}))| = |C^*(e_{i, \Cdot})| = 3$.
\item
	$e_{h, \ell} \in C(e^*_{i_1, j_1}) \cup C(e^*_{i_2, j_2})$, $e_{h+1, \ell+1} \notin C(e^*_{i_1, j_1}) \cup C(e^*_{i_2, j_2})$
	for some $e^*_{i_1, j_1}, e^*_{i_2, j_2} \in C^*(e_{i, \Cdot})$ with $|C(e^*_{i_1, j_1}) \cup C(e^*_{i_2, j_2})| = 2$.
\item
	$e_{h+1, \ell+1} \notin C(C^*(e_{i, \Cdot}))$, $|C(C^*(e_{i, \Cdot}))| = |C^*(e_{i, \Cdot})| + 1 = 4$,
	and there is an edge $e^*_{i_1, j_1} \in C^*(e_{\Cdot, j})$ such that $|C(e^*_{i_1, j_1})| = 1$.
\item
	$e_{h, \ell} \in C(e^*_{i_2, j_2}) \cup C(e^*_{i_3, j_3})$, $e_{h+1, \ell+1} \notin C(e^*_{i_2, j_2}) \cup C(e^*_{i_3, j_3})$
	for some $e^*_{i_2, j_2}, e^*_{i_3, j_3} \in C^*(e_{i, \Cdot})$, with $|C(e^*_{i_2, j_2}) \cup C(e^*_{i_3, j_3})| = 3$,
	and there is an edge $e^*_{i_1, j_1} \in C^*(e_{\Cdot, j})$ such that $|C(e^*_{i_1, j_1})| = 1$.
\item
	$e_{h, \ell} \in C(e^*_{i_2, j_2})$, $e_{h+1, \ell+1} \notin C(e^*_{i_2, j_2})$
	for some $e^*_{i_2, j_2} \in C^*(e_{i, \Cdot})$ with $|C(e^*_{i_2, j_2})| = 2$,
	and there is an edge $e^*_{i_1, j_1} \in C^*(e_{\Cdot, j})$ such that $|C(e^*_{i_1, j_1})| = 1$.
\item
	$e_{h+1, \ell+1} \in C(C^*(e_{i, \Cdot}))$, $|C(C^*(e_{i, \Cdot}))| = |C^*(e_{i, \Cdot})| + 2 = 5$,
	and there is an edge $e^*_{i_1, j_1} \in C^*(e_{\Cdot, j})$ such that $|C(e^*_{i_1, j_1})| = 1$.
\end{enumerate}
\end{lemma}
\begin{proof}
We prove by contradiction, and thus assume that there are two edges $e^*_{h_1,\ell_1}, e^*_{h_2,\ell_2} \in C^*(e_{h, \ell}) \cap C^*(e_{h+1, \ell+1})$ such that
$C(e^*_{h_1,\ell_1}) = C(e^*_{h_2,\ell_2}) = \{e_{h, \ell}, e_{h+1, \ell+1}\}$.

If the first condition holds,
then due to $e_{h+1, \ell+1} \notin C(C^*(e_{i, \Cdot}))$, none of $e^*_{h_1,\ell_1}, e^*_{h_2,\ell_2}$ is in $C^*(e_{i, \Cdot})$.
In this case, the algorithm $\mcL\mcS$ would replace the edge $e_{h+1, \ell+1}$ and the three edges of $C(C^*(e_{i, \Cdot}))$ by
$e^*_{h_1,\ell_1}, e^*_{h_2,\ell_2}$ and the three edges of $C^*(e_{i, \Cdot})$ to expand $M$,
a contradiction.

The other five conditions can be similarly proved by this kind of contradiction.
\end{proof}

Using Lemmas \ref{lemma417} and \ref{lemma418}, we can prove a better upper bound on $\omega(e)$ for those edges stated in Lemma~\ref{lemma416}.
This better bound is $\frac 52$, a reduction from $\frac {17}6$ stated in Lemma~\ref{lemma412}.

\begin{lemma}
\label{lemma419}
For any parallel edge $e \in C(C^*(e_{i,\Cdot}))$ discussed in Lemma~\ref{lemma416},
its total amount of tokens received $\omega(e)$ can be better bounded, in particular, $\omega(e) \le \frac 52$.
\end{lemma}
\begin{proof}
We enumerate through all these edges in the following:
\begin{enumerate}
\item
	In Fig.~\ref{fig401}, we have $\tau(e_{i, j} \gets C^*(e_{i, \Cdot})) = \big(\frac 12, 1, \frac 12\big)$.
	For the edge $e_{j'-2}$, it is parallel to $e_{j'-3} \notin C(C^*(e_{i, \Cdot}))$.
	By the condition 1 of Lemma \ref{lemma418} and Lemma~\ref{lemma416},
	$\omega(e_{j'-2}) \le 3 \big( \frac 12 + \frac 13 \big) = \frac 52 = 2.5$.
	The same argument applies to the edge $e_{i+2}$.

	In the rest of the proof, we point out only the condition used in the argument to avoid repetition.
	
\item
	In Fig.~\ref{fig402}, $\omega(e_{j'-2}) \le 3 \big( \frac 12 + \frac 13 \big) = \frac 52 = 2.5$,
	due to the condition 2 of Lemma \ref{lemma418}.

\item
	In Fig.~\ref{fig403}, $\omega(e_{j'+1}), \omega(e_{j'+2}) \le 2 \big( \frac 12 + 2 \times \frac 13 \big) = \frac 73 \approx 2.333$,
	due to the condition 1 of Lemma \ref{lemma417}.

\item
	In Fig.~\ref{fig404}, $\omega(e_{j'+1}) \le \big( \frac 12 + \frac 14 + \frac 13 \big) + \big( 2 \times \frac 12 + \frac 13 \big) = \frac{29}{12} \approx 2.417$,
	due to the condition 2 of Lemma \ref{lemma417}.

\item
	In Fig.~\ref{fig406a}, $\omega(e_{j'+1}), \omega(e_{j'+2}) \le 2 \big( \frac 12 + 2 \times \frac 13 \big) = \frac 73 \approx 2.333$,
	due to the condition 2 of Lemma \ref{lemma417}.
\item
	In Fig.~\ref{fig406b}, $\omega(e_{j'-2}) \le 3 \big( \frac 12 + \frac 13 \big) = \frac 52 = 2.5$,
	due to the condition 3 of Lemma \ref{lemma418};

	$\omega(e_{i+2}) \le 4 \times \frac 13 + 2 \times \frac 12 = \frac 73 \approx 2.333$,
	due to the condition 3 of Lemma \ref{lemma418}.

\item
	In Fig.~\ref{fig407a}, $\omega(e_{j'+1}) \le \big( \frac 12 + \frac 14 + \frac 13 \big) + \big( 2 \times \frac 12 + \frac 13 \big) = \frac{29}{12} \approx 2.417$,
	due to the condition 4 of Lemma \ref{lemma418}.

\item
	In Fig.~\ref{fig408a}, $\omega(e_{j'-2}), \omega(e_{i+2}) \le 4 \times \frac 13 + 2 \times \frac 12 = \frac 73 \approx 2.333$,
	due to the condition 3 of Lemma \ref{lemma418}.
\item
	In Fig.~\ref{fig408b}, $\omega(e_{j'-1}), \omega(e_{j'-2}) \le 2 \big( \frac 12 + 2 \times \frac 13 \big) = \frac 73 \approx 2.333$,
	due to the condition 3 of Lemma \ref{lemma417}.
\item
	In Fig.~\ref{fig408c}, $\omega(e_{i-2}) \le 2 \big( \frac 12 + 2 \times \frac 13 \big) = \frac 73 \approx 2.333$,
	due to the condition 4 of Lemma \ref{lemma418}.

\item
	In Fig.~\ref{fig409}, $\omega(e_{j'+1}), \omega(e_{j'+2}) \le 5 \times \frac 13 + \frac 12 = \frac{13}{6} \approx 2.167$,
	due to the condition 2 of Lemma \ref{lemma417}.

\item
	In Fig.~\ref{fig410a}, $\omega(e_{j'-2}) \le 4 \times \frac 13 + 2 \times \frac 12 = \frac 73 \approx 2.333$,
	due to the condition 4 of Lemma \ref{lemma418}.
\item
	In Fig.~\ref{fig410b}, $\omega(e_{j'+1}) \le \big( \frac 12 + \frac 14 + \frac 13 \big) + \big( 2 \times \frac 12 + \frac 13 \big) = \frac{29}{12} \approx 2.417$,
	due to the condition 5 of Lemma \ref{lemma418}.

\item
	In Fig.~\ref{fig411a}, $\omega(e_{j'}), \omega(e_{j'+2}) \le \big( \frac 12 + \frac 13 + \frac 14 \big) + \big( \frac 12 + 2 \times \frac 13 \big) = \frac 94 = 2.25$,
	due to the condition 2 of Lemma \ref{lemma417};

	$\omega(e_{j'+1}) \le \big( \frac 13 + \frac 14 + \frac 13 \big) + \big( \frac 12 + 2 \times \frac 13 \big) = \frac{25}{12} \approx 2.083$,
	due to the condition 2 of Lemma \ref{lemma417}.
\item
	In Fig.~\ref{fig411b}, $\omega(e_{j'}), \omega(e_{i+2}) \le \big( \frac 12 + \frac 13 + \frac 14 \big) + \big( \frac 12 + 2 \times \frac 13 \big) = \frac 94 = 2.25$,
	due to the condition 2 of Lemma \ref{lemma417};

	$\omega(e_{j'+1}) \le \big( \frac 13 + \frac 14 + \frac 13 \big) + \big( \frac 12 + 2 \times \frac 13 \big) = \frac{25}{12} \approx 2.083$,
	due to the condition 2 of Lemma \ref{lemma417}.
\item
	In Fig.~\ref{fig411c}, $\omega(e_{j'+1}) \le \big( \frac 12 + \frac 13 + \frac 14 \big) + \big( \frac 12 + 2 \times \frac 13 \big) = \frac 94 = 2.25$,
	due to the condition 4 of Lemma \ref{lemma418}.

\item
	In Fig.~\ref{fig412}, $\omega(e_{j'}) \le \big( \frac 12 + \frac 13 + \frac 15 \big) + \big( \frac 12 + 2 \times \frac 13 \big) = \frac{11}{5} = 2.2$,
	due to the condition 3 of Lemma \ref{lemma417};

	$\omega(e_{j'+1}) \le \big( 2 \times \frac 13 + \frac 15 \big) + \big( \frac 12 + 2 \times \frac 13 \big) = \frac{61}{30} \approx 2.033$,
	due to the condition 3 of Lemma \ref{lemma417}.

\item
	In Fig.~\ref{fig414a}, $\omega(e_{j'}) \le 5 \times \frac 13 + \frac 12 = \frac{13}{6} \approx 2.167$,
	due to the condition 2 of Lemma \ref{lemma417};

	$\omega(e_{j'-1}) \le 2 \big( \frac 12 + 2 \times \frac 13 \big) = \frac 73 \approx 2.333$,
	due to the condition 2 of Lemma \ref{lemma417};

	$\omega(e_{i+2}) \le 4 \times \frac 13 + 2 \times \frac 12 = \frac 73 \approx 2.333$,
	due to the condition 3 of Lemma \ref{lemma418}.

\item
	In Fig.~\ref{fig415a}, $\omega(e_{j'-2}) \le 3 \times \frac 12 + 3 \times \frac 13 = \frac 52 = 2.5$,
	due to the condition 4 of Lemma \ref{lemma418}.

\item
	In Fig.~\ref{fig415b}, $\omega(e_{j'+1}) \le 2 \times \frac 12 + 3 \times \frac 13 = 2$,
	due to the condition 4 of Lemma \ref{lemma418} and no edge of $M^*$ incident at ${j'+1}$.

\item
	In Fig.~\ref{fig416a}, $\omega(e_{j''}), \omega(e_{j'''}) \le 1 + 3 \times \frac 12 = \frac 52 = 2.5$,
	simply due to no edge of $M^*$ incident at ${j''}$ and $e_{j'''}$ where $j''' = j'' + 1$.
\end{enumerate}
Note the maximum value among the above is $\frac 52 = 2.5$.
The lemma is proved.
\end{proof}

The next lemma is on the parallel edges excluded from Lemma~\ref{lemma419}.

\begin{lemma}
\label{lemma420}
For each of following parallel edge $e \in C(C^*(e_{i, \Cdot}))$ with $\omega(e_{i,j}) \ge 3$,
we have 
\begin{enumerate}
\item
	for the edge $e_{j'+2}$ in Figs.~\ref{fig404}, \ref{fig407a} and \ref{fig410b}, $\omega(e_{j'+2}) \le \frac {29}{12}$;
\item
	for the edge $e_{j'+2}$ in Figs.~\ref{fig411c} and \ref{fig412}, $\omega(e_{j'+2}) \le \frac {35}{12}$;
\item
	for the edges $e_{i-2}, e_{i+2}$ in Fig.~\ref{fig414b}, either $\omega(e_{i-2}), \omega(e_{i+2}) \le \frac {35}{12}$,
	or $\omega(e_{i-2}) \le \frac {13}6$ when $\omega(e_{j'-1}) = 3$,
	or $\omega(e_{i+2}) \le \frac {13}6$ when $\omega(e_{j'+1}) = 3$;
\item
	for the edges $e_{j'+1}, e_{j'+2}$ in Fig.~\ref{fig414c}, $\omega(e_{j'+1}) \le \frac {13}6$ and $\omega(e_{j'+2}) \le \frac 73$;
\item
	for the edges $e_{j'-1}, e_{j'-2}, e_{j'+1}, e_{j'+2}$ in Fig.~\ref{fig414d}, $\omega(e) \le \frac {35}{12}$.
\end{enumerate}
\end{lemma}
\begin{proof}
We first note that in items 2) and 5) we do not succeed in getting a better bound, and thus quote the existing bounds.
More specifically, for the edge $e_{j'+2}$ in Fig.~\ref{fig411c}, $\omega(e_{j'+2}) \le \frac {35}{12}$ is from Lemma \ref{lemma414};
for the others, $\omega(e) \le \frac {35}{12}$ is from Lemma \ref{lemma415}.

In the rest of the proof, we let $e^*_{i_1, j_1}$ denote the edge of $C^*(e_{i, j})$ such that $|C(e^*_{i_1, j_1})| = 1$.

\begin{enumerate}
\item
	The edge $e_{j'+2}$ in Figs.~\ref{fig404}, \ref{fig407a} and \ref{fig410b}.


	One sees that for the edge $e_{j'+2}$ in Fig.~\ref{fig410b},
	its $\omega(e_{j'+2})$ is larger (or worst) when $\mcN_p(e_{j'+2}) = \emptyset$ than when $\mcN_p(e_{j'+2}) \ne \emptyset$.
	We therefore consider the worse case when $\mcN_p(e_{j'+2}) = \emptyset$;
	this way, all three edges can be discussed exactly the same (ignoring the incidence information of $i-2$ and $j'-2$ in $M$).

	Assume the edge $e_{j'+2}$ is incident at $h$, i.e., $e_{h, j'+2} := e_{j'+2}$.
	We consider the case where $|C^*(e_{h, j'+2})| \ge 5$, as otherwise $\omega(e_{j'+2}) \le 1 + \frac 12 + \frac 12 + \frac 14 < \frac {29}{12}$.
	When there is an edge $e^*_{h_1, \ell_1} \in C^*(e_{h, j'+2})$ such that $|C(e^*_{h_1, \ell_1})| = 1$,
	then either $h_1 = h+1$ or $\ell_1 = j'+3$.
	If there is an edge $e^*_{h}$ of $M^*$ incident at $h$, then $|C(e^*_{h})| \ge 3$,
since otherwise the algorithm $\mcL\mcS$ would be able to expand $M$ by swapping the four edges $e_{i,j}, e_{h-1,j'+1}, e_{h,j'+2}, e_{i+2}$ of $C(C^*(e_{i, \Cdot}))$
by the five edges $e^*_{h}, e^*_{h_1,\ell_1}, e^*_{i_1,j_1}, e^*_{i, j'}, e^*_{i+1,j'+1}$;
for the same reason, if there is an edge $e^*_{h-1}$ of $M^*$ incident at $h-1$, then $|C(e^*_{h-1})| \ge 3$.
These together say that the value combination of $\tau(e_{h, j'+2} \gets C^*(e_{h, \Cdot}))$ is impossible to have two values $\ge \frac 12$.
Therefore, if $|C^*(e_{h, j'+2})| = 5$, we have $\omega(e_{h, j'+2}) \le 1 + \frac 12 + \frac 13 + \frac 13 + \frac 14 = \frac {29}{12}$,
due to Lemmas \ref{lemma401}, \ref{lemma403}, and $|C(e^*_{i+1, j'+1})| = 4$ and $|C(e^*_{i+2, j'+2})| \ge 3$. 
	When there is no edge $e^* \in C^*(e_{h, j'+2})$ such that $|C(e^*)| = 1$,
if $|C^*(e_{h, j'+2})| = 5$, then we have $\omega(e_{h, j'+2}) \le 4 \times \frac 12 + \frac 14 < \frac {29}{12}$,
due to $|C(e^*_{i+1, j'+1})| = 4$. 

	We next consider $|C^*(e_{h, j'+2})| = 6$ and
$C^*(e_{h, j'+2}) = \{e^*_{i+1, j'+1}$, $e^*_{i+2, j'+2}$, $e^*_{i+3, j'+3}$, $e^*_{h-1, \ell-1}$, $e^*_{h, \ell}$, $e^*_{h+1, \ell+1}\}$, for some $\ell$.
If there is an edge of $C^*(e_{h, j'+2})$ conflicts only one edge of $M$, then this edge has to be $e^*_{h+1,\ell+1}$.
Then we have $|C(e^*_{i+2, j'+2})| \ge 4$, since otherwise the algorithm $\mcL\mcS$ would replace
the four edges $e_{i, j}$, $e_{i+2}$, $e_{h-1, j'+1}$, $e_{h, j'+2}$ by
the five edges $e^*_{i_1, j_1}, e^*_{i, j'}, e^*_{i+1, j'+1}, e^*_{i+2, j'+2}, e^*_{h+1, \ell+1}$ to expand $M$, a contradiction.
For a similar reason, we have $|C(e^*_{h, \ell})| \ge 3$ and then $|C(e^*_{h-1, \ell-1})| \ge 4$.
It follows that $|C(e^*_{i+3, j'+3})| \ge 3$ and $|C(e^*_{h, \ell})| = 3$.
Therefore, we have $\omega(e_{h, j'+2}) \le \big( 2 \times \frac 14 + \frac 13 \big) + \big( 1 + \frac 13 + \frac 14 \big) = \frac{29}{12}$.
When there is no edge of $C^*(e_{h, j'+2})$ conflicts only one edge of $M$,
since we cannot have both $|C(e^*_{h, \ell})| = |C(e^*_{h-1, \ell-1})| = 2$,
we have $\omega(e_{h, j'+2}) \le \big( \frac 14 + \frac 13 + \frac 12 \big) + \big( 2 \times \frac 12 + \frac 13 \big) = \frac{29}{12}$.

Note that in the above proof we did not use the incidence information of $i-2$ and $j'-2$ in $M$.
In summary, we have $\omega(e_{j'+2}) \le \frac{29}{12} \approx 2.417$ in Figs.~\ref{fig404}, \ref{fig407a} and \ref{fig410b}.

\item[{\textsf{3.}}]
	The edges $e_{i-2}$ and $e_{i+2}$ in Fig.~\ref{fig414b}, which can be discussed exactly the same.

Recall that there is an edge $e^*_{i_1, j_1} \in C^*(e_{\Cdot, j})$ such that $|C(e^*_{i_1, j_1})| = 1$.

From Lemma \ref{lemma415}, we know that when $\omega(e_{j'-1}) = 3$, $e_{i-2}$ has to be a parallel edge of $M$;
if $e_{i-2}$ is a singleton then $\omega(e_{i-2}) \le \frac{35}{12}$,
and if $\omega(e_{j'-1}) < 3$, then $\omega(e_{j'-1}) \le \frac{35}{12}$.
We consider in the following $\omega(e_{j'-1}) = 3$.

Assume the edge $e_{j'-1}$ is incident at $h'$, i.e., $e_{h', j'-1} := e_{j'-1}$.
Thus we have
$|C(e^*_{i-2, j'-2})| = 3$ (that is, no edge of $M$ incident at $j'-3$),
$|C(C^*(e_{\Cdot, j'-1}))| = 5$, and
there is an edge $e^*_{h'_1, j'_1} \in C^*(e_{h', \Cdot})$ such that $|C(e^*_{h'_1, j'_1})| = 1$.

Assume the edge $e_{i-2}$ is incident at $\ell$, i.e., $e_{i-2,\ell} := e_{i-2}$.
We observe first that if there is an edge $e^*_{\ell}$ of $M^*$ incident at $\ell$, then $|C(e^*_{\ell})| \ge 3$,
since otherwise the algorithm $\mcL\mcS$ would be able to expand $M$ by swapping the five edges of $C(C^*(e_{\Cdot, j'-1}))$ by six edges including $e^*_{\ell}$;
for the same reason, if there is an edge $e^*_{\ell-1}$ of $M^*$ incident at $\ell-1$, then $|C(e^*_{\ell-1})| \ge 3$;
if there is an edge $e^*_{i-3} = e^*_{i-3, j'-3}$ of $M^*$ incident at $i-3$, then $|C(e^*_{i-3})| \ge 3$.
This says that the value combination of $\tau(e_{i-2, \ell} \gets C^*(e_{i-2, \ell}))$ is impossible to have two values $\ge \frac 12$.

If there is an edge of $C^*(e_{i-2, \ell})$ conflicting only one edge of $M$, then this edge has to be $e^*_{\ell+1}$.
In this case, the algorithm $\mcL\mcS$ would replace the five edges of $C(C^*(e_{\Cdot, j'-1}))$ by
the edges $e^*_{\ell+1}, e^*_{h'_1, j'_1}, e^*_{i_1, j_1}$ and the three edges of $C^*(e_{\Cdot, j'-1})$ to expand $M$, a contradiction.
Therefore, there is no edge of $C^*(e_{i-2, \ell})$ conflicting only one edge of $M$.
It follows that if $|C^*(e_{i-2,\ell})| \le 5$, we have $\omega(e_{i-2,\ell}) \le \frac 12 + 4 \times \frac 13 = \frac {11}6$.

We next assume $C^*(e_{i-2, \ell}) = \big\{e^*_{i-1, j'-1}, e^*_{i-2, j'-2}, e^*_{i-3, j'-3}, e^*_{h-1, \ell-1}, e^*_{h, \ell}, e^*_{h+1, \ell+1}\}$, for some $h$.
Note that every one of $e^*_{i-3,j'-3}, e^*_{h-1,\ell-1}, e^*_{h,\ell}$ conflicts both edges $e_{i-3,\ell-1}$ and $e_{i-2,\ell}$.
If there is one of them conflicting only these two edges of $M$, then the five edges of $C(C^*(e_{\Cdot, j'-1}))$ can be replaced by six edges to expand $M$.
It thus follows that all three $|C(e^*_{i-3, j'-3})|$, $|C(e^*_{h-1, \ell-1})|$, $|C(e^*_{h, \ell})| \ge 3$;
and subsequently $\omega(e_{i-2, \ell}) \le 5 \times \frac 13 + \frac 12 = \frac {13}6 \approx 2.167$.

In summary, we have for $e_{i-2}$ in Fig.~\ref{fig414b}:
if $\omega(e_{j'-1}) \le \frac{35}{12}$ then $\omega(e_{i-2}) \le \frac{35}{12}$;
if $\omega(e_{j'-1}) = 3$ then $\omega(e_{i-2}) \le \frac {13}6$.
Similarly, we have for $e_{i+2}$ in Fig.~\ref{fig414b}:
if $\omega(e_{j'+1}) \le \frac{35}{12}$ then $\omega(e_{i+2}) \le \frac{35}{12}$;
if $\omega(e_{j'+1}) = 3$ then $\omega(e_{i+2}) \le \frac {13}6$.

\item[\textsf{4.1.}]
	The edge $e_{j'+1}$ in Fig.~\ref{fig414c}.

Recall that there is an edge $e^*_{i_1, j_1} \in C^*(e_{\Cdot, j})$ such that $|C(e^*_{i_1, j_1})| = 1$.

From Lemma \ref{lemma415}, we know that if $\omega(e_{j'-1}) < 3$, then $\omega(e_{j'-1}) \le \frac{35}{12}$.
We consider in the following $\omega(e_{j'-1}) = 3$.
Assume the edge $e_{j'-1}$ is incident at $h'$, i.e., $e_{h', j'-1} := e_{j'-1}$.
From $\omega(e_{j'-1}) = 3$,
there is an edge $e^*_{h'_1, j'_1} \in C^*(e_{h', \Cdot})$ such that $|C(e^*_{h'_1, j'_1})| = 1$.

Assume the edge $e_{j'+1}$ is incident at $h$, i.e., $e_{h, j'+1} := e_{j'+1}$.
We observe first that if there is an edge $e^*_{h}$ of $M^*$ incident at $h$, then $|C(e^*_{h})| \ge 3$,
since otherwise the algorithm $\mcL\mcS$ would be able to expand $M$ by swapping the five edges of $C(C^*(e_{i, \Cdot}))$ by six edges including $e^*_{h}$;
for the same reason, if there is an edge $e^*_{h+1}$ of $M^*$ incident at $h+1$, then $|C(e^*_{h+1})| \ge 3$.
This says that the value combination of $\tau(e_{h, j'+2} \gets C^*(e_{h, \Cdot}))$ is impossible to have two values $\ge \frac 12$.

If there is an edge of $C^*(e_{h, j'+1})$ conflicting only one edge of $M$, then this edge has to be $e^*_{h-1}$.
In this case, the algorithm $\mcL\mcS$ would replace the five edges of $C(C^*(e_{i, \Cdot}))$ by
the edges $e^*_{h-1}, e^*_{h'_1, j'_1}, e^*_{i_1, j_1}$ and the three edges of $C^*(e_{i, \Cdot})$ to expand $M$, a contradiction.
Therefore, there is no edge of $C^*(e_{h, j'+1})$ conflicting only one edge of $M$.
It follows that if $|C^*(e_{h, j'+1})| \le 5$, we have $\omega(e_{h, j'+1}) \le \frac 12 + \frac 13 + \frac 12 + \frac 13 + \frac 13 = 2$.

We next assume $C^*(e_{h, j'+1}) = \big\{e^*_{h-1, \ell-1}, e^*_{h, \ell}, e^*_{h+1, \ell+1}, e^*_{i, j'}, e^*_{i+1, j'+1}, e^*_{i+2, j'+2}\}$, for some $\ell$.
Note that every one of $e^*_{i+2,j'+2}, e^*_{h,\ell}, e^*_{h+1,\ell+1}$ conflicts both edges $e_{h,j'+1}$ and $e_{h+1,j'+2}$.
If there is one of them conflicting only these two edges of $M$, then the five edges of $C(C^*(e_{i, \Cdot}))$ can be replaced by six edges to expand $M$.
It thus follows that all three $|C(e^*_{i+2, j'+2})|$, $|C(e^*_{h, \ell})|$, $|C(e^*_{h+1, \ell+1})| \ge 3$;
and subsequently $\omega(e_{h, j'+1}) \le 5 \times \frac 13 + \frac 12 = \frac{13}6 \approx 2.167$.

From Lemma \ref{lemma415}, we also know that if $\omega(e_{i-2}) < 3$, then $\omega(e_{i-2}) \le \frac{35}{12}$.
We consider $\omega(e_{i-2}) = 3$,
which implies that there is an edge $e^*_{h'_1, j'_1} \in C^*(e_{i-2})$ such that $|C(e^*_{h'_1, j'_1})| = 1$.
It follows by the same argument as in the above that $\omega(e_{j'+1}) \le \frac {13}6$.

In summary, we have for $e_{j'+1}$ in Fig.~\ref{fig414c}:
if both $\omega(e_{j'-1}), \omega(e_{i-2}) \le \frac{35}{12}$ then $\omega(e_{j'+1}) \le \frac{35}{12}$ too;
otherwise, $\omega(e_{j'+1}) \le \frac {13}6$.

\item[\textsf{4.2.}]
	The edge $e_{j'+2}$ in Fig.~\ref{fig414c}.

	The argument here is the same as in the last item (4.1) to consider $\omega(e_{j'-1}) = 3$.

Assume the edge $e_{j'+2}$ is incident at $h$, i.e., $e_{h, j'+2} := e_{j'+2}$.

We observe first, for the same reasons, that if there is an edge $e^*_{h}$ of $M^*$ incident at $h$, then $|C(e^*_{h})| \ge 3$;
if there is an edge $e^*_{h+1}$ of $M^*$ incident at $h+1$, then $|C(e^*_{h+1})| \ge 3$;
there is no edge of $C^*(e_{h, j'+2})$ conflicting only one edge of $M$;
if there is an edge $e^*_{j'+2}$ of $M^*$ incident at $j'+2$, then $|C(e^*_{j'+2})| \ge 2$;
if there is an edge $e^*_{j'+3}$ of $M^*$ incident at $j'+3$, then $|C(e^*_{j'+3})| \ge 2$.
These together say that the value combination of $\tau(e_{h, j'+2} \gets C^*(e_{h, j'+2}))$ is impossible to have a value $1$,
and it is impossible to have three values $\ge \frac 12$.
It follows that $\omega(e_{h, j'+2}) \le 2 \times \frac 12 + 4 \times \frac 13 = \frac 73 \approx 2.333$.

In summary, we have for $e_{j'+2}$ in Fig.~\ref{fig414c}:
if both $\omega(e_{j'-1}), \omega(e_{i-2}) \le \frac{35}{12}$ then $\omega(e_{j'+2}) \le \frac{35}{12}$ too;
otherwise, $\omega(e_{j'+1}) \le \frac 73$.
\end{enumerate}
This finishes the proof of the lemma.
\end{proof}

Lemma~\ref{lemma405} states the $12$ possible value combinations of $\tau(e_{i, j} \gets C^*(e_{i, \Cdot}))$ with $\omega(e_{i, j}) \ge 3$,
which are
$\big\{1, \frac 12, \frac 12\big\}$,
$\big\{1, \frac 12, \frac 13\big\}$,
$\big\{1, \frac 12, \frac 14\big\}$,
$\big\{1, \frac 13, \frac 13\big\}$,
$\big\{\frac 12, \frac 12, \frac 13\big\}$,
$\big\{\frac 12, \frac 12, \frac 14\big\}$,
$\big\{\frac 12, \frac 13, \frac 13\big\}$,
$\big\{\frac 12, \frac 13, \frac 14\big\}$,
$\big\{\frac 12, \frac 13, \frac 15\big\}$,
$\big\{\frac 12, \frac 14, \frac 14\big\}$,
$\big\{\frac 13, \frac 13, \frac 13\big\}$, and
$\big\{\frac 12, \frac 12, 0\big\}$.
We next count the minimum number of known-to-be parallel edges of $M$ in $C(C^*(e_{i, \Cdot}))$ and
use Lemmas~\ref{lemma419} and \ref{lemma420} to upper bound their $\omega(\cdot)$ values respectively, for each combination.

\begin{enumerate}
\item
	$\big\{1, \frac 12, \frac 12\big\}$:
	there are at least $2$ parallel edges of $M$, each with $\omega(\cdot) \le \frac 52 = 2.5$ (by Lemma~\ref{lemma419});
\item
	$\big\{1, \frac 12, \frac 13\big\}$:
	there is at least $1$ parallel edge of $M$, with $\omega(\cdot) \le \frac 52 = 2.5$ (by Lemma~\ref{lemma419});
\item
	$\big\{1, \frac 12, \frac 14\big\}$:
	there are at least $2$ parallel edges of $M$, each with $\omega(\cdot) \le \frac{29}{12} \approx 2.417$ (by Lemmas~\ref{lemma419}, \ref{lemma420});
\item
	$\big\{1, \frac 13, \frac 13\big\}$:
	no parallel edge;
\item
	$\big\{\frac 12, \frac 12, \frac 13\big\}$:
	there are at least $2$ parallel edges of $M$, each with $\omega(\cdot) \le \frac 52 = 2.5$ (by Lemma~\ref{lemma419});
\item
	$\big\{\frac 12, \frac 12, \frac 14\big\}$:
	there are at least $2$ parallel edges of $M$, each with $\omega(\cdot) \le \frac{29}{12} \approx 2.417$ (by Lemmas~\ref{lemma419}, \ref{lemma420});
\item
	$\big\{\frac 12, \frac 13, \frac 13\big\}$:
	there is at least $1$ parallel edge of $M$, with $\omega(\cdot) \le \frac 73 \approx 2.333$ (by Lemma~\ref{lemma419});
\item
	$\big\{\frac 12, \frac 13, \frac 14\big\}$: 
	there are three possible cases,
	\begin{enumerate}
	\item
		there is at least $1$ parallel edge of $M$ with $\omega(\cdot) \le \frac 73 \approx 2.333$ (by Lemma~\ref{lemma419}), or
	\item
		there are at least $2$ parallel edges of $M$, each with $\omega(\cdot) \le \frac{29}{12} \approx 2.417$ (by Lemmas~\ref{lemma419}, \ref{lemma420}), or
	\item
		there are at least $2$ parallel edges of $M$, one with $\omega(\cdot) \le \frac 94 = 2.25$ and the other with $\omega(\cdot) \le \frac {35}{12} \approx 2.917$
		(by Lemmas~\ref{lemma419}, \ref{lemma420});
	\end{enumerate}
\item
	$\big\{\frac 12, \frac 13, \frac 15\big\}$:
	there are at least $2$ parallel edges of $M$, each with $\omega(\cdot) \le \frac{11}{5} = 2.2$ (by Lemma~\ref{lemma419});
\item
	$\big\{\frac 12, \frac 14, \frac 14\big\}$:
	no parallel edge;
\item
	$\big\{\frac 13, \frac 13, \frac 13\big\}$: 
	there are five possible cases,
	\begin{enumerate}
	\item
		there is no singleton edge other than $e_{i,j}$ with $\omega(\cdot) \ge 3$, no parallel edge;
	\item
		there is no singleton edge other than $e_{i,j}$ with $\omega(\cdot) \ge 3$,
		but there is at least $1$ parallel edge of $M$ with $\omega(\cdot) \le \frac {35}{12} \approx 2.917$ (by Lemmas~\ref{lemma419}, \ref{lemma420}),
	\item
		there is no singleton edge other than $e_{i,j}$ with $\omega(\cdot) \ge 3$,
		but there are at least $2$ parallel edges of $M$, each with $\omega(\cdot) \le \frac {35}{12} \approx 2.917$ (by Lemmas~\ref{lemma419}, \ref{lemma420}),
	\item
		there is one singleton edge other than $e_{i,j}$ with $\omega(\cdot) = 3$,
		accompanied by at least $1$ parallel edge of $M$ with $\omega(\cdot) \le \frac{13}{6} \approx 2.167$ (by Lemma~\ref{lemma420}),
	\item
		there is one singleton edge other than $e_{i,j}$ with $\omega(\cdot) = 3$,
		accompanied by at least $2$ parallel edges of $M$,
		one with $\omega(\cdot) \le \frac {13}6 \approx 2.167$ and the other with $\omega(\cdot) \le \frac 73 \approx 2.333$ (by Lemma~\ref{lemma420});
	\end{enumerate}
\item
	$\big\{\frac 12, \frac 12, 0\big\}$:
	no parallel edge.
\end{enumerate}

Lemma~\ref{lemma404} states the $8$ possible value combinations of $\tau(e_{i, j} \gets C^*(e_{i, j}))$ with $\omega(e_{i, j}) \ge 3$,
which are
$\big\{1, \frac 12, \frac 12, \frac 12, \frac 12, \frac 13\big\}$, 
$\big\{1, \frac 12, \frac 12, \frac 12, \frac 12, \frac 14\big\}$,
$\big\{1, \frac 12, \frac 12, \frac 12, \frac 13, \frac 13\big\}$,
$\big\{1, \frac 12, \frac 12, \frac 12, \frac 13, \frac 14\big\}$,
$\big\{1, \frac 12, \frac 12, \frac 12, \frac 13, \frac 15\big\}$, 
$\big\{1, \frac 12, \frac 12, \frac 12, \frac 14, \frac 14\big\}$,
$\big\{1, \frac 12, \frac 12, \frac 13, \frac 13, \frac 13\big\}$, and
$\big\{1, \frac 12, \frac 12, \frac 12, \frac 12, 0\big\}$.
These combinations give rise to $\omega(e_{i, j}) = \frac{10}{3}$, $\frac{13}{4}$, $\frac{19}{6}$, $\frac{37}{12}$, $\frac{91}{30}$, $3$, $3$ and $3$ respectively.
Based on the above list, we conclude the minimum number of known-to-be parallel edges of $M$ in $C(C^*(e_{i, j}))$ for each combination,
using the cut-off upper bound $2.5$ on their $\omega(\cdot)$ values, as follows.

\begin{enumerate}
\item
	$\big\{1, \frac 12, \frac 12, \frac 12, \frac 12, \frac 13\big\}$ ($\omega(e_{i, j}) = \frac{10}{3}$):
	there are at least $4$ parallel edges of $M$; 
\item
	$\big\{1, \frac 12, \frac 12, \frac 12, \frac 12, \frac 14\big\}$ ($\omega(e_{i, j}) = \frac{13}{4}$):
	there are at least $4$ parallel edges of $M$; 
\item
	$\big\{1, \frac 12, \frac 12, \frac 12, \frac 13, \frac 13\big\}$ ($\omega(e_{i, j}) = \frac{19}{6}$):
	there are at least $3$ parallel edges of $M$; 
\item
	$\big\{1, \frac 12, \frac 12, \frac 12, \frac 13, \frac 14\big\}$ ($\omega(e_{i, j}) = \frac{37}{12}$):
	there are at least $3$ parallel edges of $M$; 
\item
	$\big\{1, \frac 12, \frac 12, \frac 12, \frac 13, \frac 15\big\}$ ($\omega(e_{i, j}) = \frac{91}{30}$):
	there are at least $4$ parallel edges of $M$; 
\item
	$\big\{1, \frac 12, \frac 12, \frac 12, \frac 14, \frac 14\big\}$ ($\omega(e_{i, j}) = 3$):
	there are at least $2$ parallel edges of $M$; 
\item
	$\big\{1, \frac 12, \frac 12, \frac 13, \frac 13, \frac 13\big\}$ ($\omega(e_{i, j}) = 3$):
	there are two possible cases, 
	\begin{enumerate}
	\item
		there is no singleton edge other than $e_{i,j}$ with $\omega(\cdot) \ge 3$,
		but there are at least $2$ parallel edges of $M$; 
	\item
		there is one singleton edge other than $e_{i,j}$ with $\omega(\cdot) = 3$,
		accompanied by at least $3$ parallel edges of $M$; 
	\end{enumerate}
\item
	$\big\{1, \frac 12, \frac 12, \frac 12, \frac 12, 0\big\}$ ($\omega(e_{i, j}) = 3$):
	there are at least $2$ parallel edges of $M$. 
\end{enumerate}

We conclude this section with the following lemma.

\begin{lemma}
\label{lemma421}
Every edge $e_{i, j} \in M$ with $\omega(e_{i, j}) \ge 3$ must be a singleton,
and $\omega(e_{i,j}) \in \big\{\frac{10}{3}, \frac{13}{4}, \frac{19}{6}, \frac{37}{12}, \frac{91}{30}, 3\big\}$.
Furthermore,
\begin{enumerate}
\item
	the existence of an edge $e_{i, j} \in M$ with $\omega(e_{i, j}) = \frac{10}{3}$ or $\frac{13}{4}$ or $\frac{91}{30}$
	is accompanied with at least $4$ parallel edges of $M$ each with $\omega(\cdot) \le 2.5$;
\item
	the existence of an edge $e_{i, j} \in M$ with $\omega(e_{i, j}) = \frac{19}{6}$ or $\frac{37}{12}$
	is accompanied with at least $3$ parallel edges of $M$ each with $\omega(\cdot) \le 2.5$;
\item
	the existence of an edge $e_{i, j} \in M$ with $\omega(e_{i, j}) = 3$
	is accompanied with at least $1.5$ parallel edges of $M$ each with $\omega(\cdot) \le 2.5$.
\end{enumerate}
Each of these accompanying parallel edges must belong to either $C(C^*(e_{i, \Cdot}))$ with $|C^*(e_{i, \Cdot})| = 3$,
or $C(C^*(e_{\Cdot, j}))$ with $|C^*(e_{\Cdot, j})| = 3$, for some $e_{i,j}$.
\end{lemma}

\subsection{An upper bound on the average value of $\omega(e)$}
\label{sec4.6}
Let $M_{\ge 3}$ be the subset of all the edges of $M$ with $\omega(\cdot) \ge 3$, and let $n_s = |M_{\ge 3}|$.
Let $P$ denote the subset of all the accompanying parallel edges of $M$ determined in Lemma~\ref{lemma421}, and let $n_p = |P|$.
From Lemma \ref{lemma410}, every edge of $M_{\ge 3}$ is a singleton,
and thus $M_{\ge 3} \cap P = \emptyset$.

\begin{lemma}
\label{lemma422}
Each edge of $P$ belongs to $C(C^*(e_{i, j}))$ for at most four distinct edges $e_{i,j} \in M_{\ge 3}$.
\end{lemma}
\begin{proof}
Consider an edge $e_{h, \ell} \in P$ and assume that the edge $e_{h+1, \ell+1}$ is also in $M$.

Consider the vertex $d^B_\ell$ at which $e_{h, \ell}$ is incident;
let $e^*_{\ell-2}, e^*_{\ell-1}, e^*_{\ell}, e^*_{\ell+1}, e^*_{\ell+2}$ be the edge of $M^*$ incident at the vertex
$d^B_{\ell-2}$, $d^B_{\ell-1}$, $d^B_{\ell}$, $d^B_{\ell+1}$, $d^B_{\ell+2}$, respectively, if such an edge exists.
Clearly, for any edge $e_{i, j} \in M_{\ge 3}$,
if $C^*(e_{i, j})$ does not contain any of the five edges $e^*_{\ell-2}, e^*_{\ell-1}, e^*_{\ell}, e^*_{\ell+1}, e^*_{\ell+2}$,
then $e_{h, \ell} \notin C(C^*(e_{i, j}))$ (unless $e_{h, \ell} \in C(C^*(e_{i, j}))$ through the symmetric discussion using the vertex $d^A_h$).
We distinguish two cases where $C^*(e_{i, \Cdot})$ contains one of the five edges and $C^*(e_{\Cdot, j})$ contains one of the five edges, respectively.

When $C^*(e_{i, \Cdot})$ contains one of the five edges $e^*_{\ell-2}, e^*_{\ell-1}, e^*_{\ell}, e^*_{\ell+1}, e^*_{\ell+2}$,
we see from all the $27$ configurations of $C(C^*(e_{i, \Cdot}))$ and Lemma~\ref{lemma415} that
neither of the edges $e^*_{\ell}, e^*_{\ell+1}$, if exists, can be incident at the vertex $d^A_i$.
It follows that the vertex $d^A_i$ is an end of one of the three edges $e^*_{\ell-2}, e^*_{\ell-1}, e^*_{\ell+2}$.

When $C^*(e_{\Cdot, j})$ contains one of the five edges $e^*_{\ell-2}, e^*_{\ell-1}, e^*_{\ell}, e^*_{\ell+1}, e^*_{\ell+2}$,
we know that $j = \ell -2$ due to the fact that the edge $e_{i, j}$ is a singleton edge of $M$.

Since no two edges of $M_{\ge 3}$ are adjacent to a common edge of $M^*$, we conclude that there are at most three distinct edges
$e_{i, j} \in M_{\ge 3}$ such that $e_{h, \ell} \in C(C^*(e_{i, j}))$ through the vertex $d^B_\ell$.
Furthermore, if there are such three distinct edges, then one is incident at $v^B_{\ell-2}$,
one is adjacent to $e^*_{\ell-1}$ (but not incident at $v^B_{\ell-1}$),
and the other is adjacent to $e^*_{\ell+2}$ (but not incident at $v^B_{\ell+2}$);
the five edges $e^*_{\ell-2}, e^*_{\ell-1}, e^*_{\ell}, e^*_{\ell+1}, e^*_{\ell+2}$ all exist, so do the extra two edges $e^*_{\ell-3}$ and $e^*_{\ell+3}$,
and these seven edges of $M^*$ are consecutively parallel.

The three edges $e^*_{\ell-1}, e^*_{\ell}, e^*_{\ell+1}$ of $M^*$ are conflicting with
only the three edges of $M_{\ge}$ and the two parallel edges $e_{h, \ell}, e_{h+1, \ell+1}$ of $M$;
and for each of these three edges of $M_{\ge}$, there is another distinct edge of $M^*$ conflicting with only this edge of $M_{\ge}$.
In other words, there are six edges of $M^*$ conflicting with only the three edges of $M_{\ge}$ and the two parallel edges $e_{h, \ell}, e_{h+1, \ell+1}$ of $M$,
a contradiction as the algorithm $\mcL\mcS$ would swap them to expand $M$.

This proves that there are at most two distinct edges
$e_{i, j} \in M_{\ge 3}$ such that $e_{h, \ell} \in C(C^*(e_{i, j}))$ through the vertex $d^B_\ell$.
Symmetrically, we can prove that there are at most two distinct edges
$e_{i, j} \in M_{\ge 3}$ such that $e_{h, \ell} \in C(C^*(e_{i, j}))$ through the vertex $d^A_h$.
Therefore, there are at most four distinct edges
$e_{i, j} \in M_{\ge 3}$ such that $e_{h, \ell} \in C(C^*(e_{i, j}))$.
\end{proof}

Using Lemma~\ref{lemma421}, assume there is a fraction of $x n_s$ edges of $M_{\ge 3}$ each accompanied with $4$ parallel edges of $P$;
there is a fraction of $y n_s$ edges of $M_{\ge 3}$ each accompanied with $3$ parallel edges of $P$;
and there is a fraction of $(1 - x - y) n_s$ edges of $M_{\ge 3}$ each accompanied with $1.5$ parallel edges of $P$,
where $x \ge 0, y \ge 0, 1-x-y \ge 0$.
From Lemma~\ref{lemma422}, we have
\[
4 n_p \ge 4 x n_s + 3 y n_s + 1.5 (1 - x - y) n_s = (1.5 + 2.5 x + 1.5 y) n_s,
\]
which gives
\begin{equation}
\label{eq9}
\frac {n_p}{n_s} \ge \frac {1.5 + 2.5 x + 1.5 y}4,
\end{equation}
and the average amount of tokens for all the edges of $M_{\ge 3} \cup P$ is, using Equation~\ref{eq9},
\begin{equation}
\label{eq10}
\overline{\omega(e)} \le \frac {2.5 n_p + \frac {10}3 x n_s + \frac {19}6 y n_s + 3 (1 - x - y) n_s}{n_p + n_s}
	\le \frac 52 + \frac {12 + 8x}{33 + 15x}
	\le \frac {35}{12}.
\end{equation}

Lemma \ref{lemma415} tells that every other edge of $M$ has its $\omega(\cdot) \le \frac{35}{12}$ too.
Therefore, the average amount of tokens for all the edges of $M$ is no greater than $\frac{35}{12}$.
We have thus proved the following theorem.

\begin{theorem}
\label{thm423}
The algorithm $\mcL\mcS$ is an $O(n^{13})$-time $\frac{35}{12}$-approximation for both the {\sc MCBM} and the {\sc Max-Duo} problems.
\end{theorem}

\section{Lower bounds on the locality gap for the algorithm $\mcL\mcS$}
\label{sec5}
In this section, we present two instances of the {\sc MCBM} and {\sc Max-Duo} problems, respectively,
to show that the approximation ratio of the algorithm $\mcL\mcS$ has a lower bound of $\frac{13}{6} > 2.166$ for {\sc MCBM}
and a lower bound of $\frac 53 > 1.666$ for {\sc Max-Duo}.

\subsection{An instance of {\sc MCBM}}
Consider the bipartite graph $G = (V^A, V^B, E)$ shown in Fig.~\ref{fig501},
where $V^A = \{1, 2, \ldots, 26\}$, $V^B = \{1', 2', \ldots, 26'\}$,
and $E$ is the set of all the edges in solid and dashed lines.
One can see that the set of $26$ consecutive parallel edges (in dashed lines) is an optimal solution $M^*$ to the {\sc MCBM} problem on $G$.
Let $M$ be the maximal compatible matching shown as solid lines in Fig.~\ref{fig501}, and assume it is the starting matching for the algorithm $\mcL\mcS$ on $G$.

\begin{figure}[H]
\centering
\includegraphics[width=0.9\linewidth]{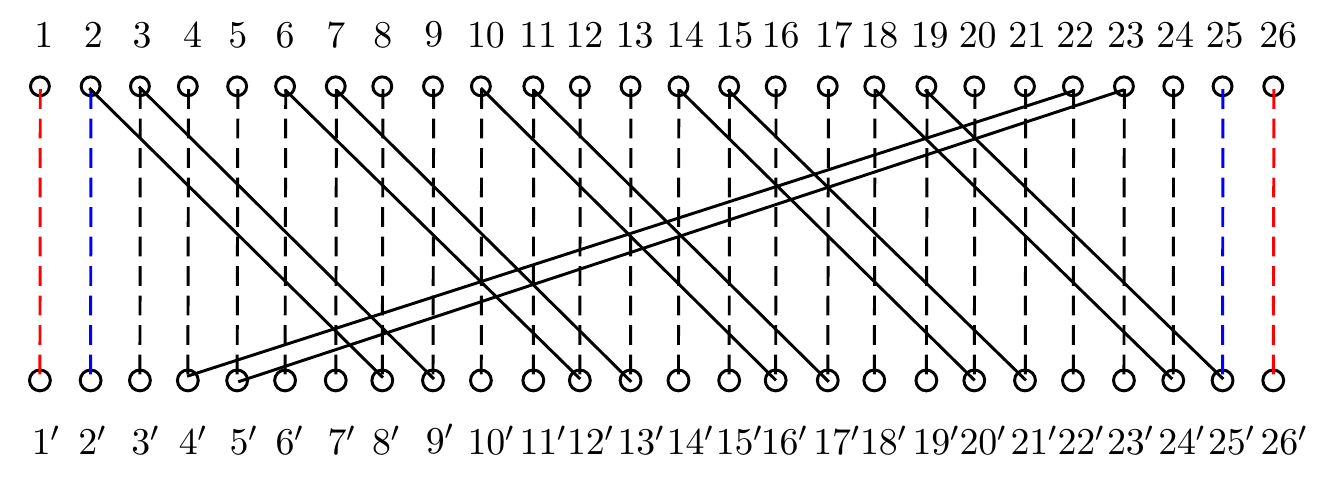} 
\caption{A bipartite graph $G = (V^A, V^B, E)$,
where $V^A = \{1, 2, \ldots, 26\}$, $V^B = \{1', 2', \ldots, 26'\}$, and $E = M \cup M^*$.
$M$ consists of the $12$ edges in solid lines; it is a maximal compatible matching in $G$;
$M^*$ consists of the $26$ edges in dashed lines; it is an optimal compatible matching to the {\sc MCBM} problem on $G$.\label{fig501}}
\end{figure}

\begin{lemma}
\label{lemma501}
$M$ cannot be further improved by the algorithm $\mcL\mcS$ due to the following reasons:
\begin{enumerate}
\item $M$ is a maximal compatible matching in $G$;
\item all the edges of $M$ are parallel edges;
\item for any $e \in M$, there is at most one edge of $M^*$ compatible with all the edges in $M - \{e\}$;
\item for any $2$ edges $e_1, e_2 \in M$, there are at most two edges of $M^*$ compatible with all the edges in $M - \{e_1, e_2\}$;
\item for any $3$ edges $e_1, e_2, e_3 \in M$, there are at most three edges in $M^*$ compatible with all the edges in $M - \{e_1, e_2, e_3\}$;
\item for any $4$ edges $e_1, \ldots, e_4 \in M$, there are at most four edges in $M^*$ compatible with all the edges in $M - \{e_1, \ldots, e_4\}$;
\item for any $5$ edges $e_1, \ldots, e_5 \in M$, there are at most five edges in $M^*$ compatible with all the edges in $M - \{e_1, \ldots, e_5\}$.
\end{enumerate}
\end{lemma}
\begin{proof}
The first two items are trivial.
We note that the second item implies that $M$ cannot be further improved by the algorithm using the operation {\sc Reduce-5-By-5}.
We next show that $M$ cannot be further improved by the algorithm using the operation {\sc Replace-5-By-6}.

We examine whether some edges of $M$ can be swapped out for more edges from $M^*$ by {\sc Replace-5-By-6}.
For ease of presentation, we partition $M^*$ into 3 subsets
$M^{*1} = \{(1, 1'), (26, 26')\}$,
$M^{*2} = \{(2, 2'), (25, 25')\}$, and
$M^{*3} = \{(3, 3'), (4, 4'), \ldots, (24, 24')\}$ (in Fig.~\ref{fig501}, their edges are colored red, blue, black, respectively).
We have the following observations.

\begin{obs}
\label{obs5.1}
To swap for an edge of $M^{*1}$, a unique edge of $M$ has to be swap out.
In details, $(2, 8')$ needs to be swapped out for $(1, 1')$,
and $(19, 25')$ needs to be swapped out for $(26, 26')$.
\end{obs}

\begin{obs}
\label{obs5.2}
To swap for an edge of $M^{*2}$, a unique pair of parallel edges of $M$ has to be swap out.
In details, $(2, 8')$ and $(3, 9')$ need to be swapped out for $(2, 2')$,
and $(19, 25')$ and $(18, 24')$ need to be swapped out for $(25, 25')$.
\end{obs}

\begin{obs}
\label{obs5.3}
To swap for an edge of $M^{*3}$, a unique triplet of edges of $M$ has to be swap out.
In details, when $i = 3 \ (\mbox{mod} \ 4)$, $e_{i-1}, e_i, e_{{i+1}'}$ need to be swapped out for $(i, i')$;
when $i = 0 \ (\mbox{mod} \ 4)$, $e_{i-1}, e_{i'}, e_{{i+1}'}$ need to be swapped out for $(i, i')$;
when $i = 1 \ (\mbox{mod} \ 4)$, $e_{{i-1}'}, e_{i'}, e_{{i+1}}$ need to be swapped out for $(i, i')$;
when $i = 2 \ (\mbox{mod} \ 4)$, $e_{{i-1}'}, e_i, e_{{i+1}}$ need to be swapped out for $(i, i')$.
\end{obs}

The Observations \ref{obs5.1}, \ref{obs5.2}, \ref{obs5.3} prove trivially the items 3--5 of the lemma.

To prove the item 6, we next see what a subset of four edges of $M$ can do.
If these four edges are able to swap in one edge of $M^{*3}$,
then by Observation~\ref{obs5.3} this edge of $M^{*3}$ actually requires three out of the four edges.
Note that these particular three edges are not able to swap for any other edge of $M^{*3}$.
If they contain a unique pair of parallel edges of $M$ in Observation~\ref{obs5.2},
then they can swap in three edges one from each of $M^{*1}, M^{*2}, M^{*3}$,
and the fourth edge either forms with two of them to form another triplet to swap in an other edge of $M^{*3}$,
or it is able to swap in the other edge of $M^{*1}$.
If these particular three edges do not contain a unique pair of parallel edges of $M$ in Observation~\ref{obs5.2},
then they can swap in only the edge of $M^{*3}$,
and the fourth edge can 
either form with one of them to form a unique pair of parallel edges of $M$ in Observation~\ref{obs5.2} to swap in two edges one from each of $M^{*1}, M^{*2}$,
and/or form with two of them to form another triplet to swap in an other edge of $M^{*3}$,
or it is able to swap in the other edge of $M^{*1}$.
Therefore, these four edges can swap in the best case four edges $(1, 1'), (2, 2'), (3, 3')$ (or $(26, 26'), (25, 25'), (24, 24')$, respectively)
and another edge of $M^{*1} \cup M^{*3}$.
If these four edges are not able to swap in any edge of $M^{*3}$, then they can swap in the best case four edges of $M^{*1} \cup M^{*2}$.

To prove the item 7, we next see what a subset of five edges of $M$ can do.
If these five edges are able to swap in two edges of $M^{*3}$,
then by Observation~\ref{obs5.3} these two edges of $M^{*3}$ actually require at least four out of the five edges.
Note that these particular four edges are only able to swap for four edges of $M^*$, including the above two edges of $M^{*3}$
and the other two edges must be either $(1, 1'), (2, 2')$ or $(26, 26'), (25, 25')$.
Then the fifth edge either forms with two of these particular four edges to form another triplet to swap in an other edge of $M^{*3}$,
or it is able to swap in the other edge of $M^{*1}$.
If these five edges are not able to swap in at least two edges of $M^{*3}$,
then they can swap in the best case one edge of $M^{*3}$ and four edges of $M^{*1} \cup M^{*2}$.
\end{proof}

\begin{theorem}
\label{thm502}
There is a lower bound of $\frac{13}{6} > 2.166$ on the locality gap of the algorithm $\mcL\mcS$ for the {\sc MCBM} problem.
\end{theorem}
\begin{proof}
By Lemma \ref{lemma501}, if the matching $M$ is fed as the starting matching to the algorithm $\mcL\mcS$,
then the algorithm terminates without modifying it.
Note that we have $|M| = 12$ and $|M^*| = 26$, we conclude that the algorithm can not do better than $\frac {13}6$ in the worst case.
\end{proof}

\begin{remark}
Using our amortized analysis, let $C(e^*)$ be the subset of edges of $M$ conflicting with the edge $e^* \in M^*$.
Then in the above instance, we have $|C(e^*)| = 1, 2, 3$ for $e^* \in M^{*1}, M^{*2}, M^{*3}$, respectively.
The maximum total amount of tokens received by the edges of $M$ is achieved at $(2, 8')$,
where $\omega((2, 8')) = 1 + \frac 12 + 4 \times \frac 13 = \frac{17}{6} \approx 2.833$.
This maximum is quite close to the approximation ratio $\frac {35}{12} \approx 2.917$ of the algorithm $\mcL\mcS$,
which is also the maximum possible $\omega(\cdot)$ value for the parallel edges of $M$.
(Recall that $\frac {35}{12} = 1 + 2 \times \frac 12 + 2 \times \frac 13 + \frac 14$.)
\end{remark}

\begin{remark}
We may have variants of the algorithm $\mcL\mcS$ by substituting the operation {\sc Replace-5-By-6}
with the similarly defined operation {\sc Replace-$\rho$-By-($\rho+1$)}, for any $\rho \ge 1$
(with or without the operation {\sc Reduce-5-By-6}).

For $\rho = 1, 2, 3, 4$, we can construct similar instances to show the corresponding lower bounds on the locality gap for the variants on the {\sc MCBM} problem.
\begin{itemize}
\item $\rho = 4$.

We can construct a similar instance $G = (V^A, V^B, E)$, except that $|V^A| = |V^B| = 22$ and $|M| = 10$.
The performance ratio of the algorithm is $\frac{11}{5} = 2.2$.
 
\item $\rho = 3$.

We can construct two different instances for $G = (V^A, V^B, E)$.
One is similar to the previous two, except that $|V^A| = |V^B| = 18$ and $|M| = 8$.
In the other we have $|V^A| = |V^B| = 9$ and $|M| = 4$, such that the four edges of $M$ are consecutively parallel.
The performance ratio of the algorithm on both instances is $\frac{9}{4} = 2.25$.

\item $\rho = 2$.

We can construct an instance similar to the second instance when $\rho = 3$, which is a graph $G = (V^A, V^B, E)$ with $|V^A| = |V^B| = 8$ and $|M| = 3$,
such that the three edges of $M$ are consecutively parallel.
The performance ratio of the algorithm is $\frac 83 \approx 2.667$.

\item $\rho = 1$.

We can construct a similar graph $G = (V^A, V^B, E)$ with $|V^A| = |V^B| = 7$ and $|M| = 2$,
such that the two edges of $M$ are parallel.
The performance ratio of the algorithm is $\frac 72 = 3.5$.
This is essentially the $3.5$-approximation by Boria {\it et al.}~\cite{BCC16}, and the instance shows that the performance ratio is tight.
\end{itemize}
\end{remark}

\subsection{An instance of {\sc Max-Duo}}
In the instance of {\sc Max-Duo}, we have two identical length-$11$ strings $A = (a, b, c, d, e, f, b, c, d$, $e, g) = B$.
We construct the corresponding bipartite graph $G = (V^A, V^B, E)$ (shown in Fig.~\ref{fig502}),
where $V^A = \{1, 2, \ldots, 10\}$ and $V^B = \{1', 2', \ldots, 10'\}$.
Since $A = B$, each pair of duos represented by the vertices $i$ and $i'$ are the same, for $i = 1, 2, \ldots, 10$.
Thus it is easy to see that there is an optimal solution $M^*$ to {\sc MCBM} on $G$,
which consists of all the $10$ edges in dashed lines shown in Fig.~\ref{fig502}.
Let $M$ be the compatible matching consisting of the six edges in solid lines in Fig.~\ref{fig502}.

\begin{figure}[H]
\centering
\includegraphics[width=0.5\linewidth]{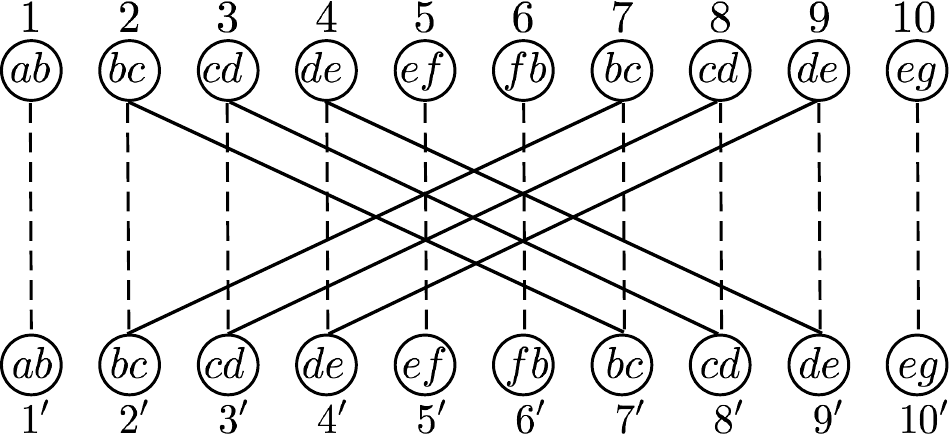}
\caption{The corresponding bipartite graph $G = (V^A, V^B, E)$ constructed from two identical strings $A = B = (a, b, c, d, e, f, b, c, d, e, g)$,
where $V^A = \{1, 2, \ldots, 10\}$, $V^B = \{1', 2', \ldots, 10'\}$, and $E = M \cup M^*$.
$M$ consists of the six edges in solid lines and it is a compatible matching;
$M^*$ consists of the ten edges in dashed lines and it is an optimal compatible matching to the {\sc MCBM} problem on $G$.\label{fig502}}
\end{figure}

\begin{lemma}
\label{lemma503}
$M$ cannot be further improved by the algorithm $\mcL\mcS$ due to the following reasons:
\begin{enumerate}
\item
	$M$ is a local maximal compatible matching;
\item
	all edges of $M$ are parallel edges;
\item
	for any $e \in M$, there is no edge in $M^*$ compatible with all the $5$ edges in $M - \{e\}$;
\item
	for any $2$ edges $e_1, e_2 \in M$, there are at most $2$ edges in $M^*$ compatible with all the $4$ edges in $M - \{e_1, e_2\}$;
\item
	for any $3$ edges $e_1, e_2, e_3 \in M$, there are at most $2$ edges in $M^*$ compatible with all the $3$ edges in $M - \{e_1, e_2, e_3\}$;
\item
	for any $4$ edges $e_1, \ldots, e_4 \in M$, there are at most $4$ edges in $M^*$ compatible with both of the $2$ edges in $M - \{e_1, \ldots, e_4\}$;
\item
	for any $5$ edges $e_1, \ldots, e_5 \in M$, there are $4$ edges in $M^*$ compatible with the only edge in $M - \{e_1, \ldots, e_5\}$.
\end{enumerate}
\end{lemma}
\begin{proof}
The first three items are trivial.
We note that the second item implies that $M$ cannot be further improved by the algorithm using the operation {\sc Reduce-5-By-5}.
We partition $M$ into three subsets $M^2 = \{(2, 7'), (7, 2')\}$, $M^3 = \{(3, 8'), (8, 3')\}$, and $M^4 = \{(4, 9'), (9, 4')\}$.

For any $e \in M$, we see that $|C^*(e)| = 6$, implying there are only $4$ edges of $M^*$ compatible with $e$;
this proves the 7th observation.

For the two edges $e_1, e_2$ in the same part of the partition, we see that $C^*(e_1) = C^*(e_2)$,
implying that if one of them is in $M$, then none of the six edges of $C^*(e_1)$ can be compatible with it.
(This proves again the item 3.)
Also, we see that the edges $(1, 1'), (6, 6')$ are not compatible with the edges and only these edges in $M^2$;
and the edges $(10, 10'), (5, 5')$ are not compatible with the edges and only these edges in $M^4$.

To prove the item 4, we see that when the two edges $e_1, e_2 \in M$ are not in the same part,
then from the last paragraph there is no edge in $M^*$ compatible with all the $4$ edges in $M - \{e_1, e_2\}$;
when the two edges $e_1, e_2 \in M$ are in the same part,
then either there are two edges in $M^*$ compatible with all the $4$ edges in $M - \{e_1, e_2\}$, if this part is $M^2$ or $M^4$,
or otherwise there is no edge in $M^*$ compatible with all the $4$ edges in $M - \{e_1, e_2\}$.

For the item 5, any three edges $e_1, e_2, e_3 \in M$ cannot take up two separate parts,
and therefore there are at most $2$ edges in $M^*$ compatible with all the $3$ edges in $M - \{e_1, e_2, e_3\}$.

For any $4$ edges $e_1, \ldots, e_4 \in M$, if they take up two parts,
then there are exactly $4$ edges of $M^*$ compatible with the $2$ edges in $M - \{e_1, \ldots, e_4\}$, which belong to the same part;
if they do not take up two parts, then there are at most $2$ edges of $M^*$ compatible with the $2$ edges in $M - \{e_1, \ldots, e_4\}$.
This proves the item 6, and completes the proof of the lemma.
\end{proof}

\begin{theorem}
\label{thm504}
There is a lower bound of $\frac 53 > 1.666$ on the locality gap of the algorithm $\mcL\mcS$ for the {\sc Max-Duo} problem.
\end{theorem}
\begin{proof}
By Lemma \ref{lemma503}, if the matching $M$ is fed as the starting matching to the algorithm $\mcL\mcS$,
then the algorithm terminates without modifying it.
Note that we have $|M| = 6$ and $|M^*| = 10$, we conclude that the algorithm can not do better than $\frac 53$ in the worst case.
\end{proof}


\section{Conclusions}
\label{sec6}
We studied the {\sc Max-Duo} problem, the complement of the well studied {\sc MCSP} problem.
Motivated by an earlier local search algorithm,
we presented an improved heuristics $\mcL\mcS$ for a more general {\sc MCBM} problem,
that uses one operation to increase the cardinality of the solution
and another novel operation to reduce the singleton edges in the solution.
The heuristics is iterative and has a time complexity $O(n^{13})$, where $n$ is the length of the input strings.
Through an amortized analysis,
we are able to show that the proposed algorithm $\mcL\mcS$ has an approximation ratio of at most ${35}/{12} < 2.917$.
This improves the current best $3.25$-approximation for both problems, and breaks the barrier of $3$.
In a companion paper, we are able to design a $(1.4 + \epsilon)$-approximation for $2$-{\sc Max-Duo},
a restricted version in which every letter of the alphabet occurs at most twice in each input string.
Together, we improved all current best approximability results for the {\sc Max-Duo} problem.

We also showed that there is a lower bound of ${13}/{6} > 2.166$ and $5/3 > 1.666$ on the locality gap of the algorithm $\mcL\mcS$
for the {\sc MCBM} and the {\sc Max-Duo} problems, respectively.

We remark that the time complexity of the algorithm $\mcL\mcS$ can possibly be reduced using appropriate data structures.
For the performance ratio, one would likely do a better analysis by examining more cases with large $\omega(\cdot)$ values,
which we are looking into.
On the other hand, it is interesting to investigate whether or not swapping more edges can lead to a better approximation.

\subparagraph*{Acknowledgements.}
All authors are supported by NSERC Canada.
Additionally, 
YC is supported by the NSFC Grants No. 11401149, 11201105, 11571252 and 11571087, and the China Scholarship Council Grant No. 201508330054;
TL is supported by the NSFC Grant No. 71371129 and the PSF China Grant No. 2016M592680;
GL is supported by the NSFC Grant No. 61672323.


\end{document}